\documentclass[twoside,11pt]{article}
\usepackage{jmlr2e}

\usepackage[english]{babel}
\usepackage{amsmath}
\usepackage{amstext}
\usepackage{amssymb}
\usepackage{mathtools}
\usepackage{algorithm}
\usepackage{algorithmic}
\usepackage{color}
\usepackage[dvipsnames]{xcolor}
\usepackage{ulem}

\def\R{\mathbb{R}}

\def\rP{\mathbb{P}}
\def\rE{\mathbb{E}}

\def\spn{\mathop{\rm span}}
\def\Lip{\mathop{\rm Lip}}

\def\av{\mathop{\rm {av}}}

\def\argmin{\mathop{\rm arg\, min}}

\def\F{{\mathcal F}}

\def\C{{\mathcal C}}

\def\P{{\mathcal P}}

\def\tpi{{\tilde \pi}}

\def\hQ{{\hat Q}}
\def\hmu{{\hat \mu}}

\def\sX{{\mathsf X}}

\def\sA{{\mathsf A}}

\def\sE{{\mathsf E}}

\def\sW{{\mathsf W}}

\def\tbeta{\tilde{\beta}}

\newtheorem{assumption}{Assumption}

\jmlrheading{1}{2022}{NA}{04/21}{NA}{NA}{Berkay Anahtarci, Can Deha Kariksiz, and Naci Saldi}

\ShortHeadings{Learning MFGs with Discounted and Average Costs}{Anahtarci, Kariksiz, and Saldi}
\firstpageno{1}

\allowdisplaybreaks

\begin{document}

\title{Learning Mean-Field Games with Discounted and Average Costs}

\author{\name Berkay Anahtarci \email berkay.anahtarci@ozyegin.edu.tr\\
       \addr Department of Natural and Mathematical Sciences\\
       \"{O}zye\u{g}in University\\
       \.{I}stanbul, Turkey
       \AND
       \name Can Deha Kariksiz \email deha.kariksiz@ozyegin.edu.tr\\     
       \addr Department of Natural and Mathematical Sciences\\
       \"{O}zye\u{g}in University\\
       \.{I}stanbul, Turkey
       \AND
       \name Naci Saldi \email naci.saldi@bilkent.edu.tr \\
       \addr Department of Mathematics\\
       Bilkent University\\
       Ankara, Turkey
       }
       
\editor{NA}

\maketitle

\begin{abstract}
 We consider learning approximate Nash equilibria for discrete-time mean-field games with stochastic nonlinear state dynamics subject to both average and discounted costs. To this end, we introduce a mean-field equilibrium (MFE) operator, whose fixed point is a mean-field equilibrium, i.e., equilibrium in the infinite population limit. We first prove that this operator is a contraction, and propose a learning algorithm to compute an approximate mean-field equilibrium by approximating the MFE operator with a random one. Moreover, using the contraction property of the MFE operator, we establish the error analysis of the proposed learning algorithm. We then show that the learned mean-field equilibrium constitutes an approximate Nash equilibrium for finite-agent games.

\end{abstract}


\section{Introduction}\label{sec1}

We consider learning approximate Nash equilibria in discrete-time stochastic dynamic games with a large population of identical agents in a mean-field interaction. The usual approach to analyse these game models is to study the infinite-population limit of the problem. This idea was utilized in the works of \cite{HuMaCa06}, \cite{LaLi07}, where mean-field games (MFGs) were introduced to obtain an approximate Nash equilibria for continuous-time differential games with a large number of agents interacting via a mean-field term (i.e., the empirical distribution of the local states). For studies of continuous-time mean-field games with various models and cost functions, 
see  \cite{HuCaMa07,TeZhBa14,Hua10,BeFrPh13,Ca11,CaDe13,GoSa14,MoBa16}.

Our goal in this paper is to learn approximate Nash equilibria for a class of stochastic dynamic games by considering stationary mean-field games in the infinite population limit. In particular, we establish an algorithm to learn \emph{stationary} or \emph{oblivious} mean-field equilibrium (see \cite{WeBeRo05,WeBeRo08}) in the infinite population limit and make use of the learned equilibrium in the finite agent setting as an approximate Nash equilibrium. 

In stationary mean-field games, a generic agent competes against a stationary mean-field term that time-homogeneously models the collective behaviour of other agents (\cite{WeBeRo05}), and therefore solves a Markov decision process (MDP) with a constraint on the stationary distribution of the state. The stationary mean-field equilibrium, which consists of a policy and a state measure, is the concept of equilibrium in the infinite-population limit. This pair must satisfy the Nash certainty equivalence (NCE) principle (\cite{HuMaCa06}), which states that the policy should be optimal under a given state measure, and when the generic agent applies this policy the resulting stationary distribution of the agent's state must be the same as the state measure. The existence of stationary mean-field equilibrium can be established via Kakutani's fixed point theorem under quite mild assumptions. Furthermore, it can be shown that when the number of agents is large enough, the policy in stationary mean-field equilibrium is an approximate Nash equilibrium for a finite-agent setting (\cite{AdJoWe15}). 


In the literature, \cite{WeBeRo10} propose an algorithm for computing oblivious equilibrium in a stationary mean-field industry dynamics model. \cite{AdJoWe15} study a stationary mean-field game model with a countable state-space under an infinite-horizon discounted-cost criterion. 
\cite{HuMa19} consider stationary mean-field games with binary action space, where they establish the existence and the uniqueness of the stationary mean-field equilibrium. \cite{WeLi22} consider stationary mean-field games with a continuum of states and actions, and establish a novel uniqueness result for stationary mean-field equilibrium 
\cite{GoMoSo10} study both stationary and non-stationary mean-field games with a finite state space over a finite horizon and establish the existence and uniqueness of the mean-field equilibrium for both cases. \cite{ElLiNi13,MoBa15,NoNa13,MoBa16-cdc} consider discrete-time mean-field games with linear state dynamics. \cite{SaBaRaSIAM,SaBaRaMOR1} consider a discrete-time non-stationary mean-field game with Polish state and action spaces under the discounted-cost optimality criterion for fully-observed case and partially-observed case, respectively. \cite{SaBaRaMOR2} consider a discrete-time risk-sensitive non-stationary mean-field game with Polish state and action spaces. \cite{Bis15,Wie19,WiAl05,Sal20} study discrete-time non-stationary mean-field games subject to the average-cost optimality criterion. 

We point out that, except the linear model and the paper of \cite{WeBeRo10}, the studies mentioned above only establish the existence and uniqueness of the mean-field equilibrium, and they propose no algorithm with convergence guarantee to compute this mean-field equilibrium when the model is known. In our recent work (\cite{AnKaSa19}), we  study this problem for a very general class of models, propose a value iteration algorithm, and prove the convergence of this algorithm to the stationary mean-field equilibrium. In this current paper, we generalize this algorithm to the model-free setting using fitted $Q$-iteration (see \cite{AnMuSz07}); that is, we propose a learning algorithm to compute an equilibrium solution for discrete-time stationary mean-field games under the discounted-cost and average-cost optimality criteria. 

Learning in stationary mean-field games has become prominent in recent years. In the continuous-time setup, \cite{YiMeMeSh14}  develop a learning algorithm for a mean-field oscillator game model to obtain approximate Nash equilibrium. In the discrete case, \cite{KaFrHa11} consider the learning of equilibrium policy in static anonymous games with countably many players.  \cite{YaLuLiZhZhWa18} establish a learning algorithm for classical stochastic games via mean-field approximation by factoring the $Q$-function in terms of actions. \cite{SuAd18, SuAd19} consider learning gradient-based equilibria in stationary mean-field games and develop a two-time scale stochastic gradient ascent algorithm, respectively.
\cite{GuHuXuZh19} develop a $Q$-learning algorithm to obtain stationary mean-field equilibria for finite state-action stationary mean-field games, where the convergence analysis depends on contractivity assumptions on the operators involved in the algorithm. \cite{ElPeLaGePi19} establish a fictitious play iterative learning algorithm to compact state-action non-stationary mean-field games under finite-horizon discounted cost criterion, where the dynamics of the state and the one-stage cost function satisfy certain structure. They also propose an error analysis of the learning algorithm for the game model with deterministic state dynamics. 
\cite{CaLaTa19} study linear-quadratic mean-field control and establish the convergence of policy gradient algorithm. \cite{FuYaChWa19} develop an actor-critic algorithm to learn mean-field equilibrium for linear-quadratic mean-field games. \cite{YaYeTrXuZh18} consider a mean-field game in which agents can control their transition probabilities without any restriction. In this case, the action space becomes the set of distributions on the state space. Using this specific structure, they can transform a mean-field game into an equivalent deterministic Markov decision process by enlarging the state and action spaces, and then, apply classical reinforcement learning algorithms to compute mean-field equilibrium. \cite{CaLaTa19-b} apply a similar analysis to mean-field control problems, and the convergence of the $Q$-learning algorithm for deterministic systems is established.

In this paper, we develop a learning algorithm that guarantees convergence in a discrete-time stationary mean-field game with nonlinear stochastic state dynamics.
We take into account the average cost criterion, in contrast to earlier research that mainly dealt with discounted cost or finite-horizon total cost criteria. It's also important to note that the majority of the aforementioned works with convergence guarantees focus on finite state and finite action settings, whereas we assume that the action space is a compact and uncountable subset of a finite dimensional Euclidean space. In general, it is more challenging to deal with this assumption. Furthermore, we prove that our algorithm converges to the global stationary mean-field equilibrium rather than local stationary mean-field equilibria. We also establish easily verifiable conditions on the system components for the convergence of the learning algorithm, which is lacking in some of the prior works mentioned above.  

Our learning algorithm performs two-steps in each iteration. In the first step, given any mean-field term, an optimal policy is learned via a fitted $Q$-iteration algorithm. Then, using this optimal policy and the current mean-field term, the next mean-field term is computed in the second step by an empirical estimate of the transition probability, which is obtained via simulation. We prove that the policy obtained by this algorithm is close to the mean-field equilibrium policy, and can therefore be used as an approximate Nash equilibrium for a finite-agent game if the number of agents is sufficiently large.

The error analysis of our learning process depends crucially on the determination of the contraction of the operator providing the stationary mean-field equilibrium. It is necessary to prove that the optimal policy is Lipschitz continuous with respect to the current mean-field term since the optimal policy corresponding to the current mean-field term affects the next mean-field term through the value iteration algorithm. Although establishing the Lipschitz continuity of the optimal $Q$ function with respect to the mean-field term is straightforward, it is quite challenging to do the same for the optimal policy. To overcome this, we assume that the function in the optimality equation is strongly convex and has Lipschitz continuous gradient. In our recent work (\cite{AnKaSa20-ar}) we provide a different approach by introducing a strongly convex regularization function in the one-stage cost function that helps us to obtain Lipschitz continuity of the optimal policy with respect to the mean-field term via duality between strong convexity and smoothness, therefore eliminating the need for strong convexity and smoothness assumptions on the system components when establishing the Lipschitz regularity of the optimal policy with respect to the mean-field term. Although this regularization approach allows us to relax the assumptions on the system components, it adds bias to the equilibrium solution since regularization in general favours randomized policies over deterministic policies, and causes the regularized stationary mean-field equilibrium to deviate from the true stationary mean-field equilibrium as a result of the additional regularization term in the one-stage cost function.

The paper is organized as follows. In Section~\ref{sec3}, we introduce the mean-field game and define the mean-field equilibrium. In Section~\ref{known-model} and Section~\ref{av-known-model}, we introduce MFE operator when the model is known for discounted-cost and average-cost, respectively. In Section~\ref{sec2} and Section~\ref{av-sec2}, we formulate the finite-agent version of the game problem  for discounted-cost and average-cost, respectively. In Section~\ref{unknown-model} and Section~\ref{av-unknown-model}, we propose and perform the error analysis of the learning algorithm for the unknown model and prove that learned mean-field equilibrium constitutes an approximate Nash equilibrium for finite-agent games  for discounted-cost and average-cost, respectively. In Section~\ref{num_ex}, we propose a numerical example.  Section~\ref{conc} concludes the paper.

\smallskip

\noindent\textbf{Notation.} 
For a finite set $\sE$, we let $\P(\sE)$ and $M(\sE)$ denote the set of all probability distributions on $\sE$ and the set of real-valued functions on $\sE$, respectively. In this paper, $\P(\sE)$ is always endowed with $l_1$-norm $\|\cdot\|_1$. We let $m(\cdot)$ denote the Lebesgue measure on appropriate finite dimensional Euclidean space $\R^d$. For any $a \in \R^d$ and $\rho > 0$, let $B(a,\rho) \coloneqq \{b: \|a-b\| \leq \rho\}$, where $\|\cdot\|$ denotes the Euclidean norm. Let $Q: \sE_1 \times \sE_2 \rightarrow \R$, where $\sE_1$ and $\sE_2$ are two sets. Then, we define $Q_{\min}(e_1) \coloneqq \inf_{e_2 \in \sE_2} Q(e_1,e_2)$. The notation $v\sim \nu$ means that the random element $v$ has distribution $\nu$.

\section{Mean-field games and mean-field equilibria}\label{sec3}

In this paper, we consider a discrete-time mean-field game with state space $\sX$ and action space $\sA$. Here, $\sX$ is a finite set with the discrete metric $d_{\sX}(x,y)=1_{\{x \neq y\}}$ and $\sA$ is a convex and compact subset of a finite dimensional Euclidean space $\R^{\dim_{\sA}}$ equipped with the Euclidean norm $\|\,\cdot\,\|$\footnote{In this work, by updating several definitions appropriately, all results are still true if one metrizes the finite dimensional Euclidean space $\R^{\dim_{\sA}}$ with $l_p$-norm for $p\geq 1$.}.   
The state dynamics evolve according to the transition probability $p : \sX \times \sA \times \P(\sX) \to \P(\sX)$; that is, given the current state $x(t)$, action $a(t)$, and state-measure $\mu$, the next state $x(t+1)$ is distributed as follows:
$$
x(t+1) \sim p(\cdot|x(t),a(t),\mu). 
$$
In this model, a policy $\pi$ is a conditional distribution on $\sA$ given $\sX$; that is, $\pi:\sX\rightarrow\P(\sA)$. Let $\Pi$ denote the set of all policies. A policy $\pi$ is deterministic if $\pi(\,\cdot\,|x) = \delta_{f(x)}(\,\cdot\,)$ for some $f:\sX\rightarrow\sA$. Let $\Pi_d$ denote the set of all deterministic policies.

Although we name this model as mean-field game, it is indeed neither a game nor a Markov decision process (MDP) in the strict sense. This model is in between them. Similar to the MDP model, we have a single agent with Markovian dynamics that has an objective function to minimize. However, similar to the game model, this agent should also compete with the collective behaviour of other agents. We model this collective behaviour by an exogenous \textit{state-measure} $\mu \in \P(\sX)$\footnote{In classical mean-field game literature, the exogenous behaviour of the other agents is in general modelled by a state measure-flow $\{\mu_t\}$, $\mu_t \in \P(\sX)$ for all $t$, which means that total population behaviour is non-stationary. In this paper, we only consider the stationary case; that is, $\mu_t = \mu$ for all $t$. Establishing a learning algorithm for the non-stationary case is more challenging and is a future research direction.}. By law of large numbers, this measure $\mu$ should be consistent with the state distribution of this single agent when the agent applies its optimal policy. The precise mathematical description of the problem is given as follows.

If we fix a state-measure $\mu \in \P(\sX)$, which represents the collective behaviour of the other agents, the evolution of the state and action of a generic agent is governed by transition probability $p : \sX \times \sA \times \P(\sX) \to \P(\sX)$, policy $\pi: \sX \rightarrow \P(\sA)$, and initial distribution $\eta_0$ of the state; that is,
\begin{align}
x(0) &\sim \eta_0, \,\,\,
x(t) \sim p(\,\cdot\,|x(t-1),a(t-1),\mu), \text{ } t\geq1, \nonumber \\
a(t) &\sim \pi(\,\cdot\,|x(t)), \text{ } t\geq0. \nonumber
\end{align} 
For this model, a policy $\pi^{*} \in \Pi$ of a generic agent is optimal for $\mu$ if
\begin{align}
J_{\mu}(\pi^{*}) = \inf_{\pi \in \Pi} J_{\mu}(\pi), \nonumber
\end{align}
where 
\begin{align}
J_{\mu}(\pi) &= E^{\pi}\biggl[ \sum_{t=0}^{\infty} \beta^t c(x(t),a(t),\mu) \biggr] \nonumber \\
\intertext{or}
J_{\mu}(\pi) &= \limsup_{T\rightarrow\infty} \frac{1}{T} E^{\pi}\biggl[ \sum_{t=0}^{T-1} c(x(t),a(t),\mu) \biggr]. \nonumber
\end{align}
Here, the first cost function is the discounted-cost with discount factor $\beta \in (0,1)$ and the second cost function is the average-cost. The measurable function $c: \sX \times \sA \times \P(\sX) \rightarrow [0,\infty)$ is the one-stage cost function. With these definitions, to introduce the optimality criteria of the model, we need the following two set-valued mappings.  

The first set-valued mapping $\Psi : \P(\sX) \rightarrow 2^{\Pi}$ is defined as
$$\Psi(\mu) = \{\pi \in \Pi: \pi \text{ is optimal for }  \mu \text{ }\text{when} \text{ } \eta_0 = \mu\}.$$
Hence, $\Psi(\mu)$ is the set of optimal policies for a given state-measure $\mu$ when the initial distribution $\eta_0$ is equal to $\mu$ as well.

We define the second set-valued mapping $\Lambda : \Pi \to 2^{\P(\sX)}$ as follows: for any $\pi \in \Pi$, the state-measure $\mu_{\pi} \in \Lambda(\pi)$ if it satisfies the following fixed point equation: 
\begin{align}
\mu_{\pi}(\,\cdot\,) = \sum_{x \in \sX}\int_{\sA} p(\,\cdot\,|x,a,\mu_{\pi})  \, \pi(da|x) \, \mu_{\pi}(x). \nonumber
\end{align}
Note that if $\mu_{\pi} \in \Lambda(\pi)$ and $\eta_0 = \mu_{\pi}$, then $x(t) \sim \mu_{\pi}$ for all $t\geq 0$ under policy $\pi$. 
If there is no assumption on the transition probability $p : \sX \times \sA \times \P(\sX) \to \P(\sX)$, we may have $\Lambda(\pi) = \emptyset$ for some $\pi$. However, under Assumption~\ref{as1}, we always have  $\Lambda(\pi)$ non-empty, which will be proved in Lemma~\ref{auxlemma}.

We can now define the notion of equilibrium for mean-field games via the mappings $\Psi$, $\Lambda$ as follows.

\begin{definition}
A pair $(\pi_*,\mu_*) \in \Pi \times \P(\sX)$ is a \emph{mean-field equilibrium} if $\pi_* \in \Psi(\mu_*)$ and $\mu_* \in \Lambda(\pi_*)$; that is, $\pi_*$ is an optimal policy for $\mu_*$ and $\mu_*$ is the stationary distribution of the states under policy $\pi_*$ and initial distribution $\mu_*$. 
\end{definition}

In the literature, the existence of mean-field equilibria has been established for the discounted-cost in \cite{SaBaRaSIAM} and for the average-cost in \cite{Wie19,Sal20}. Our aim in this paper is to develop a learning algorithm for computing an approximate mean-field equilibrium in the model-free setting. To that end, we define the following relaxed version of mean-field equilibrium.

\begin{definition}
Let $(\pi_*,\mu_*) \in \Pi_d \times \P(\sX)$ be a \emph{mean-field equilibrium}. A policy $\pi_{\varepsilon} \in \Pi_d$ is an $\varepsilon$-mean-field equilibrium policy if 
$$
\sup_{x \in \sX} \|\pi_{\varepsilon}(x)-\pi_*(x)\| \leq \varepsilon.
$$ 
\end{definition}
Note that in above definition, we require that $\pi_*$ is deterministic. Indeed, this is the case under the assumptions stated below. Therefore, without loss of generality, we can place this restriction on $\pi_*$.

With this definition, our goal now is to learn an $\varepsilon$-mean-field equilibrium policy under the model-free set-up. To this end, we will impose certain assumptions on the components of the mean-field game model. Before doing this, we need to give some definitions. Let us define $M_{\tau}(\sX)$ as the set of real-valued functions on $\sX$ bounded by $\|c\|_{\infty}/(1-\tau)$. Here, $\tau = \beta$ if the objective function is discounted-cost and $\tau=\beta^{\av}$ (see Assumption~\ref{av-as2}) if the objective function is average-cost. Let $F: \sX \times M_{\tau}(\sX) \times \P(\sX) \times \sA \rightarrow \R$ be given by
$$
F: \sX \times M_{\tau}(\sX) \times \P(\sX) \times \sA \ni (x,v,\mu,a) \mapsto 
c(x,a,\mu) + \xi \sum_{y \in \sX} v(y) \, p(y|x,a,\mu) \in \R,
$$
where $\xi=\beta$ if the objective function is discounted-cost and $\xi=1$ if the objective function is average-cost. We may now state our assumptions.

\begin{assumption}
\label{as1}
\begin{itemize}
\item [ ]
\item [(a)] The one-stage cost function $c$ satisfies the following Lipschitz bound:  
\begin{align}
|c(x,a,\mu) - c(\hat{x},\hat{a},\hat{\mu})| &\leq L_1 \, \left(d_{\sX}(x,\hat{x})+ \|a-\hat{a}\| + \|\mu-\hat{\mu}\|_1 \right), 
\end{align}
for every $ x,\hat{x} \in  \sX$, $a, \hat{a} \in \sA$, and  $\mu, \hat{\mu} \in \P(\sX)$.
\item [(b)] The stochastic kernel $p(\,\cdot\,|x,a,\mu)$ satisfies the following Lipschitz bound:
\begin{align}
&\|p(\cdot|x,a,\mu) - p(\cdot|\hat{x},\hat{a},\hat{\mu})\|_1 \leq K_1 \, \left(d_{\sX}(x,\hat{x})+\|a -\hat{a}\| + \|\mu-\hat{\mu}\|_1\right),
\end{align}
for every $ x,\hat{x} \in  \sX$, $a, \hat{a} \in \sA$, and  $\mu, \hat{\mu} \in \P(\sX)$.
\item [(c)] There exists $\alpha > 0$ such that for any $a \in \sA$ and $\delta >0$, we have 
$$
m\left( B(a,\delta) \cap \sA \right) \geq \min\left\{ \alpha \, m(B(a,\delta)), m(\sA) \right\}.
$$  
\item [(d)] For any $x \in \sX$, $v \in M_{\tau}(\sX)$, and $\mu \in \P(\sX)$, $F(x,v,\mu,\cdot)$ is $\rho$-strongly convex. Moreover, the gradient $\nabla F(x,v,\mu,a): \sX \times M_{\tau}(\sX) \times \P(\sX) \times \sA \rightarrow \R^d$ of $F$ with respect to $a$ satisfies the following Lipschitz bound:
\begin{align}
\sup_{a \in \sA} \|\nabla F(x,v,\mu,a) - \nabla F(\hat{x},\hat{v},\hat{\mu},a)\| \leq K_F \, \left(d_{\sX}(x,\hat{x})+\|v-\hat{v}\|_{\infty} + \|\mu-\hat{\mu}\|_1\right), \nonumber  \nonumber
\end{align}
for every $ x,\hat{x} \in  \sX$, $v, \hat{v} \in M_{\tau}(\sX)$, and  $\mu, \hat{\mu} \in \P(\sX)$.
\end{itemize}
\end{assumption}

Let us motivate these assumptions. First, assumptions (a) and (b) are standard conditions in stochastic control theory to obtain a rate of convergence bound for learning algorithms. Assumption (c) is needed to bound the $l_{\infty}$-norm  of Lipschitz continuous functions on $\sA$ with their $l_2$-norm. Assumption (d) is imposed to guarantee Lipschitz continuity of the optimal policy with respect to the corresponding state-measure. Indeed, this condition is equivalent to the standard assumption that guarantees Lipschitz continuity, with respect to unknown parameters, of the optimal solutions of the convex optimization problems  \cite[Theorem 4.51]{BoSh00}.

\begin{example}\label{example1}
Let us consider the industry dynamics model, introduced in \cite{WeBeRo05,WeBeRo08}, where the state $x(t) \in \sX$ of the system gives the quality level of the firm, and the state lives in the finite set $\sX = \{0,\ldots,m\}$. Given the mean-field term $\mu$, the state of the system evolves in the following form:
$$
x(t+1) = \min\{x(t)+h(a(t),\mu,w(t)),m\},
$$
where $a(t) \in \sA = [0,K]$ is the action, which denotes the investment of the agent to increase its quality, $h:\sA\times\P(\sX)\times\sW \rightarrow \sX$, and $w(t) \in \sW$ is the independent noise\footnote{The state dynamics of the model in \cite{WeBeRo05,WeBeRo08} does not depend on the mean-field term $\mu$. For generality, we assume that there is such a dependence in our example.}. In this model, the problem is to maximize the discounted reward, which is equivalent to minimizing the negative of the discounted reward. Therefore, in the minimization formulation, one-stage cost function has the following form:
$$
c(x,a,\mu) = c_2(a) - c_1(x,\mu),
$$
where $c_1(x,\mu)$ is the profit of the firm and $c_2(a)$ is the cost of the investment. For this model, Assumption~\ref{as1}-(c) is true with $\alpha=1$. To have Assumption~\ref{as1}-(d), we need to assume that: (i) $c_2(a)$ is differentiable and $\rho$-strongly convex (this is true if, for instance, $\displaystyle d^2c_2/da^2 \geq \rho$), (ii) $\rP[h(a,\mu,w)=l]$ is convex in $a$ and continuously differentiable in $(a,\mu)$, for all $l \in \sX$. Indeed, for any $x$, $v$, and $\mu$, the function $F(x,v,\mu,a)$ has the following form:
\begin{align}
&F(x,v,\mu,a) = c_2(a)-c_1(x,\mu) + \xi \sum_{y \in \sX} v(y) \, \rP[\min\{x+h(a,\mu,w),m\} = y] \nonumber \\
&= c_2(a)-c_1(x,\mu) + \xi \left[ \sum_{x \leq y <m} v(y) \, \rP[h(a,\mu,w) = y-x] + v(m) \, \rP[h(a,\mu,w) \geq m-x] \right] \nonumber \\
&= c_2(a)-c_1(x,\mu) + \xi \left[\sum_{0 \leq y <m} v(y) \, \rP[h(a,\mu,w) = y-x] + \sum_{l=m-x}^{m} v(m) \, \rP[h(a,\mu,w) = l] \right], \nonumber 
\end{align}
where the last equality is true since $\rP[h(a,\mu,w) = y-x]=0$ if $y<x$. Since $c_2(a)$ is $\rho$-strongly convex and $\rP[h(a,\mu,w)=l]$ is convex in $a$, the function $F(x,v,\mu,a)$ is $\rho$-strongly convex in $a$. Moreover, for every $ x,\hat{x} \in  \sX$, $v, \hat{v} \in M_{\tau}(\sX)$, and  $\mu, \hat{\mu} \in \P(\sX)$, we have \begin{align}
&\sup_{a \in \sA} |\nabla F(x,v,\mu,a) - \nabla F(\hat{x},\hat{v},\hat{\mu},a)| \leq \sup_{a \in \sA} |\nabla F(x,v,\mu,a) - \nabla F(\hat{x},v,\mu,a)|
\nonumber  \\
&\phantom{xxxxx}+\sup_{a \in \sA} |\nabla F(\hat{x},v,\mu,a) - \nabla F(\hat{x},\hat{v},\mu,a)| + \sup_{a \in \sA} |\nabla F(\hat{x},\hat{v},\mu,a) - \nabla F(\hat{x},\hat{v},\hat{\mu},a)| \label{aux-eq0}.
\end{align}
To bound the terms in the sum (\ref{aux-eq0}), let us define the following constants:
\begin{align}
\Theta_1 &\coloneqq \sup_{\mu,a,l} |\nabla_a \rP[h(a,\mu,w)=l]| \nonumber \\
\Theta_2 &\coloneqq \sup_{a,\mu} \sum_{l=0}^{m-1} \left|\nabla_a \rP[h(a,\mu,w)=l] - \nabla_a \rP[h(a,\mu,w)=l+1]\right| \nonumber \\
\Theta_3 &\coloneqq \sup_a \sum_{l=0}^{m} \sup_{\mu} \|\nabla_{a,\mu} \rP[h(a,\mu,w)=l]\| \nonumber \\
\Theta_4 &\coloneqq \sup_{a,\mu} \sum_{l=0}^{m} |\nabla_{a} \rP[h(a,\mu,w)=l]|, \nonumber
\end{align}
where $\nabla_a \rP[h(a,\mu,w)=l]$ and $\nabla_{a,\mu} \rP[h(a,\mu,w)=l]$ are the gradients of $\rP[h(a,\mu,w)=l]$ with respect to $a$ and $(a,\mu)$, respectively. Since $\rP[h(a,\mu,w)=l]$ is continuously differentiable with respect to $(a,\mu)$, and the sets $\sA$ and $\P(\sX)$ are compact, the constants above are well-defined.

Without loss of generality, suppose $x \leq \hat{x}$. Considering the first term in the sum (\ref{aux-eq0}), for all $a \in \sA$ we have 
\begin{align}
&|\nabla F(x,v,\mu,a) - \nabla F(\hat{x},v,\mu,a)| \nonumber \\
&\leq \xi \bigg|\sum_{0 \leq y <m} v(y) \, \nabla_a \rP[h(a,\mu,w) = y-x] + \sum_{m-x\leq l \leq m} v(m) \, \nabla_a \rP[h(a,\mu,w) = l] \nonumber \\
&\phantom{xxxxxxxxxx}\sum_{0 \leq y <m} v(y) \, \nabla_a \rP[h(a,\mu,w) = y-\hat{x}] + \sum_{m-\hat{x}\leq l \leq m} v(m) \, \nabla_a \rP[h(a,\mu,w) = l] \bigg| \nonumber \\
&\leq \xi \bigg| \sum_{0 \leq y <m} v(y) \, \sum_{l=y-\hat{x}}^{y-x-1} \left( \nabla_a \rP[h(a,\mu,w)=l] - \nabla_a \rP[h(a,\mu,w)=l+1] \right)\bigg| \nonumber \\
&\phantom{xxxxxxxxxx} + \xi \, \bigg| \sum_{m-\hat{x} \leq y <m-x} v(m) \, \nabla_a \rP[h(a,\mu,w) = l] \bigg| \nonumber \\
&\leq \xi \, \frac{\|c\|_{\infty}}{1-\tau} \, (\Theta_2+\Theta_1) \, |x-\hat{x}|. \label{aux-eq1}
\end{align}
For the second term in the sum (\ref{aux-eq0}), for all $a \in \sA$ we have 
\begin{align}
&|\nabla F(\hat{x},v,\mu,a) - \nabla F(\hat{x},\hat{v},\mu,a)| \nonumber \\
&\leq \xi \bigg|\sum_{\hat{x} \leq y <m} v(y) \, \nabla_a \rP[h(a,\mu,w) = y-\hat{x}] + \sum_{m-\hat{x}\leq l \leq m} v(m) \, \nabla_a \rP[h(a,\mu,w) = l] \nonumber \\
&\phantom{xxxxxxxxxx}\sum_{\hat{x} \leq y <m} \hat{v}(y) \, \nabla_a \rP[h(a,\mu,w) = y-\hat{x}] + \sum_{m-\hat{x}\leq l \leq m} \hat{v}(m) \, \nabla_a \rP[h(a,\mu,w) = l] \bigg| \nonumber \\
&\leq \xi \, \Theta_4 \, \|v-\hat{v}\|_{\infty}. \label{aux-eq2}
\end{align}
Finally, for the third term in the sum (\ref{aux-eq0}), for all $a \in \sA$, we have
\begin{align}
&|\nabla F(\hat{x},\hat{v},\mu,a) - \nabla F(\hat{x},\hat{v},\hat{\mu},a)| \nonumber \\
&\leq \xi \bigg|\sum_{\hat{x} \leq y <m} \hat{v}(y) \, \nabla_a \rP[h(a,\mu,w) = y-\hat{x}] + \sum_{m-\hat{x}\leq l \leq m} \hat{v}(m) \, \nabla_a \rP[h(a,\mu,w) = l] \nonumber \\
&\phantom{xxxxxxxxxx}\sum_{\hat{x} \leq y <m} \hat{v}(y) \, \nabla_a \rP[h(a,\hat{\mu},w) = y-\hat{x}] + \sum_{m-\hat{x}\leq l \leq m} \hat{v}(m) \, \nabla_a \rP[h(a,\hat{\mu},w) = l] \bigg| \nonumber \\
&\leq \xi \, \frac{\|c\|_{\infty}}{1-\tau}\, \bigg|\sum_{0 \leq l \leq m} \bigg(\nabla_a \rP[h(a,\mu,w) = l] -\nabla_a \rP[h(a,\hat{\mu},w) = l] \bigg) \bigg|. \label{aux-eq3}
\end{align}
By the mean-value theorem, there exists $\tilde{\mu}$ such that $$
\nabla_a \rP[h(a,\mu,w) = l] -\nabla_a \rP[h(a,\hat{\mu},w) = l] = \nabla_{a,\mu} \rP[h(a,\tilde{\mu},w) = l] \cdot (\mu-\hat{\mu})
$$
Hence, (\ref{aux-eq3}) can be bounded from above as follows:
\begin{align}
(\ref{aux-eq3}) \leq \xi \, \frac{\|c\|_{\infty}}{1-\tau}\, \Theta_3 \, \|\mu-\hat{\mu}\|_1. \label{aux-eq4}  
\end{align}
Bringing together the upper bounds in (\ref{aux-eq1}), (\ref{aux-eq2}), and (\ref{aux-eq4}), we get
\begin{align}
&\sup_{a \in \sA} |\nabla F(x,v,\mu,a) - \nabla F(\hat{x},\hat{v},\hat{\mu},a)| \nonumber \\
&\leq \max \bigg\{\xi \, \frac{\|c\|_{\infty \, m}}{1-\tau} \, (\Theta_2+\Theta_1) ,\xi \, \Theta_4,\xi \, \frac{\|c\|_{\infty}}{1-\tau}\, \Theta_3\bigg\} \left(d_{\sX}(x,\hat{x})+\|v-\hat{v}\|_{\infty} + \|\mu-\hat{\mu}\|_1\right) \label{aux-eq5} \\
&\eqqcolon K_F \left(d_{\sX}(x,\hat{x})+\|v-\hat{v}\|_{\infty} + \|\mu-\hat{\mu}\|_1\right). \nonumber
\end{align}
since $|x-\hat{x}| \leq m \, d_{\sX}(x,\hat{x})$. Hence, Assumption~\ref{as1}-(d) is true for this model under the conditions (i) and (ii). Note that the bound in (\ref{aux-eq5}) is fairly crude. By using further properties of the transition probability and the one-stage cost function in addition to conditions (i) and (ii) in specific examples, one can significantly improve this bound. Moreover, instead of the $l_1$-norm on the set of distributions on the state space $\sX$, if we use Wasserstein distance of order 1, it is also possible to improve the bound in (\ref{aux-eq5}) by altering the related analysis according to this distance. 
\end{example}

In this paper, we first consider learning in discounted-cost MFGs. Then, we turn our attention to the average-cost case. In the sequel, we first introduce a mean-field equilibrium (MFE) operator for discounted-cost, which can be used to compute mean-field equilibrium when the model is known. We prove that this operator is a contraction. Then, under model-free setting, we approximate this MFE operator with a random one and establish a learning algorithm.
Using this learning algorithm, we obtain $\varepsilon$-mean-field equilibrium policy with high confidence. To obtain the last result, it is essential that MFE operator is contraction. After we complete the analysis for discounted-cost, we study average-cost setting by applying the same strategy.

Before we move on to the next section, for completeness, let us prove the following result. 

\begin{lemma}\label{auxlemma}
Under Assumption~\ref{as1}, for any $\pi$, the set $\Lambda(\pi)$ is non-empty.  
\end{lemma}  

\begin{proof}
Recall that for any $\pi \in \Pi$, the state-measure $\mu_{\pi} \in \Lambda(\pi)$ if it satisfies the following fixed-point equation: 
\begin{align}
\mu_{\pi}(\,\cdot\,) = \sum_{x \in \sX} \int_{\sA} p(\,\cdot\,|x,a,\mu_{\pi})  \, \pi(da|x) \, \mu_{\pi}(x). \label{ccc}
\end{align}
Let us define the set-valued mapping $L_{\pi}:\P(\sX)\rightarrow 2^{\P(\sX)}$ as follows: given $\mu \in \P(\sX)$, a probability measure $\hat{\mu} \in L_{\pi}(\mu)$ if it is an invariant distribution of the transition probability $\int_{\sA} p(\,\cdot\,|x,a,\mu) \, \pi(da|x)$; that is
$$
\hat{\mu}(\,\cdot\,) = \sum_{x \in \sX} \int_{\sA} p(\,\cdot\,|x,a,\mu) \,\pi(da|x) \, \hat{\mu}(x). 
$$ 
Note that the transition probability $\int_{\sA} p(\,\cdot\,|x,a,\mu) \,\pi(da|x)$ is Feller continuous, and since $\sX$ is finite, the sequence of $n$-step transition probabilities are tight for any $x \in \sX$.
Therefore, we can apply Krylov-Bogoliubov theorem \cite[Theorem 4.17]{Hai06}, and obtain that  $L_{\pi}(\mu)$ is non-empty for each $\mu \in \P(\sX)$. Moreover, $L_{\pi}(\mu)$ is also convex for each $\mu \in \P(\sX)$. If we can prove that $L_{\pi}$ has a closed graph, by Kakutani's fixed point theorem \cite[Corollary 17.55]{AlBo06}, we can conclude that $L_{\pi}$ has a fixed point $\hat{\mu}$; that is, $\hat{\mu}$ satisfies (\ref{ccc}). Hence, $\hat{\mu} \in \Lambda(\pi)$. 

To this end, let $(\mu_n,\hat{\mu}_n) \rightarrow (\mu,\hat{\mu})$, where $\hat{\mu}_n \in L_{\pi}(\mu_n) $ for each $n$. Note that $L_{\pi}$ has a closed graph if $\hat{\mu} \in L_{\pi}(\mu)$. For each $n$, we have 
$$
\hat{\mu}_n(y) = \sum_{x \in \sX} \int_{\sA} p(y|x,a,\mu_n) \,\pi(da|x) \, \hat{\mu}_n(x), \,\, \forall y \in \sX.
$$
For all $y \in \sX$, the left part of the above equation converges to $\mu(y)$ since $\hat{\mu}_n \rightarrow \hat{\mu}$ and the right part of the same equation converges to 
$$
\sum_{x \in \sX} \int_{\sA} p(y|x,a,\mu) \,\pi(da|x) \, \hat{\mu}(x)
$$
since $\hat{\mu}_n \rightarrow \hat{\mu}$ and $\int_{\sA} p(y|\,\cdot\,,a,\mu_n) \,\pi(da|\,\cdot\,)$ converges to $\int_{\sA} p(y|\,\cdot\,,a,\mu) \,\pi(da|\,\cdot\,)$ continuously by Assumption~\ref{as1} \cite[Theorem 3.5]{Lan81}\footnote{Suppose $g$, $g_n$ ($n\geq1$) are uniformly bounded measurable functions on metric space $\sE$. The sequence of functions $g_n$ is said to converge to $g$ continuously if $\lim_{n\rightarrow\infty}g_n(e_n)=g(e)$ for any $e_n\rightarrow e$ where $e \in \sE$. In this case, \cite[Theorem 3.5]{Lan81} states that if $\mu_n \rightarrow \mu$ weakly, then $\int_{\sE}g_n(e) \, \mu_n(de) \rightarrow \int_{\sE} g(e) \, \mu(de)$. If $\sE$ is finite with discrete metric, then weak convergence of probability measures on $\sE$ is equivalent to $l_1$-convergence.}. Hence, we have 
$$
\hat{\mu}(\,\cdot\,) = \sum_{x \in \sX} \int_{\sA} p(\,\cdot\,|x,a,\mu) \,\pi(da|x) \, \hat{\mu}(x).
$$
In other words, $\hat{\mu} \in L_{\pi}(\mu)$, and so, $L_{\pi}$ has a closed graph. This completes the proof. 
\end{proof}

\section{Mean-field equilibrium operator for discounted-cost}\label{known-model}

In this section, we introduce a mean-field equilibrium
(MFE) operator for discounted-cost, whose fixed point is a mean-field equilibrium. We prove that this operator is a contraction. Using this result, we then establish the convergence of the learning algorithm that gives approximate mean-field equilibrium policy. To that end, in addition to Assumption~\ref{as1}, we impose the assumption below. But, before that, let us define the constants:

\begin{align}
c_{\bf m} \coloneqq \|c\|_{\infty}, \,\, Q_{\bf m} \coloneqq \frac{c_{\bf m}}{1-\beta}, \,\, Q_{\Lip} \coloneqq \frac{L_1}{1-\beta K_1/2}. \label{constants1}
\end{align}

\begin{assumption}\label{as2}
We assume that 
$$\hspace{-25pt}\frac{3K_1}{2} \left(1+\frac{K_F}{\rho}\right)+\frac{K_1K_F Q_{\Lip}}{\rho(1-\beta)} < 1,$$
where $Q_{\Lip} > 0$.
\end{assumption}

This assumption is used to ensure that the MFE operator is a contraction, which is crucial to establish the error analysis of the learning algorithm. Note that using Banach fixed point theorem, one can also compute the mean-field equilibrium by applying the MFE operator recursively to obtain successive approximations (i.e., Picard iteration). However,  even if it is restrictive, it is not possible to prove convergence of the learning algorithm without a contraction condition. In addition, imposing a contraction condition is a common method in learning mean-field games (see \cite{GuHuXuZh19,FuYaChWa19}). 

Note that, given any state-measure $\mu$, the value function $J_{\mu}$ of policy $\pi$ with initial state $x$ is given by
$$
J_{\mu}(\pi,x) \coloneqq E^{\pi}\biggl[ \sum_{t=0}^{\infty} \beta^t c(x(t),a(t),\mu) \, \bigg| \, x(0) = x \biggr]. 
$$
Then, the optimal value function is defined as $J_{\mu}^*(x) \coloneqq \inf_{\pi \in \Pi} J_{\mu}(\pi,x)$ for all $x \in \sX$. Using $J_{\mu}^*$, we can characterize the set of optimal policies $\Psi(\mu)$ for $\mu$ as follows. Firstly, $J_{\mu}^*(x)$ is the unique fixed point of the Bellman optimality operator $T_{\mu}$, which is a $\beta$-contraction with respect to the $\|\cdot\|_{\infty}$-norm:
$$
J_{\mu}^*(x) = \min_{a \in \sA} \bigg[c(x,a,\mu) + \beta \sum_{y \in \sX} J_{\mu}^*(y) \, p(y|x,a,\mu) \bigg] \eqqcolon T_{\mu}J_{\mu}^*(x).
$$
Additionally, if $f^*: \sX \rightarrow \sA$ attains the minimum in the equation above for all $x \in \sX$ as follows
\begin{align}
&\min_{a \in \sA} \bigg[c(x,a,\mu) + \beta \sum_{y \in \sX} J_{\mu}^*(y) \, p(y|x,a,\mu) \bigg] = c(x,f^*(x),\mu) + \beta \sum_{y \in \sX} J_{\mu}^*(y) \, p(y|x,f^*(x),\mu), \nonumber
\end{align}
then the policy $\pi^*(a|x) = \delta_{f^*(x)}(a) \in \Pi_d$ is optimal for $\mu$ and for any initial distribution $\eta_0$. We refer the reader to \cite[Chapter 4]{HeLa96} and \cite[Chapter 8]{HeLa99} for these classical results in MDP theory.

We can also obtain a similar characterization by using the optimal $Q$-function instead of the optimal value function $J_{\mu}^*$. Indeed, we define the optimal $Q$-function as
$$
Q_{\mu}^*(x,a) = c(x,a,\mu) + \beta \sum_{y \in \sX} J_{\mu}^*(y) \, p(y|x,a,\mu). 
$$
Note that $Q_{\mu,\min}^*(x) \coloneqq \min_{a \in \sA} Q_{\mu}^*(x,a) = J_{\mu}^*(x)$ for all $x \in \sX$, and so, we have
$$
Q_{\mu}^*(x,a) = c(x,a,\mu) + \beta \sum_{y \in \sX} Q_{\mu,\min}^*(y) \, p(y|x,a,\mu)  \eqqcolon H_{\mu}Q_{\mu}^*(x,a), 
$$
where $H_{\mu}$ is the Bellman optimality operator for $Q$-functions. It is straightforward to prove that $H_{\mu}$ is a $\|\cdot\|_{\infty}$-contraction with modulus $\beta$ and the unique fixed point of $H_{\mu}$ is $Q_{\mu}^*$. Hence, we can develop a $Q$-iteration algorithm to compute $Q_{\mu}^*$, and the policy $\pi^*(a|x) = \delta_{f^*(x)}(a) \in \Pi_d$ is optimal for $\mu$ and for any initial distribution $\eta_0$, if $Q_{\mu}^*(x,f^*(x)) = Q_{\mu,\min}^*(x)$ for all $x \in \sX$. The advantage of $Q$-iteration algorithm is that one can adapt this algorithm to the model-free setting via $Q$-learning.

Let us recall the following fact about $l_1$ norm on the set probability distributions on finite sets \cite[p. 141]{Geo11}. Suppose that there exists a real valued function $g$ on a finite set $\sE$. Then, for any pair of probability distributions $\mu,\nu$ on $\sE$, we have 
\begin{align}
\left|\sum_{e} g(e) \, \mu(e) - \sum_{e} g(e) \, \nu(e) \right| \leq \frac{\spn(g)}{2} \, \|\mu-\nu\|_{1}, \label{tv-bound}
\end{align}
where $\spn(g) \coloneqq \sup_{e \in \sE} g(e) - \inf_{e \in \sE} g(e)$
is the span-seminorm. Using this result, we can prove the following fact about optimal value functions.

\begin{lemma}\label{lip-value}
For any $\mu  \in \P(\sX)$, the optimal value function $Q^{*}_{\mu,\min}$ is $Q_{\Lip}$-Lipschitz continuous; that is, 
$$
|Q^{*}_{\mu,\min}(x)-Q^{*}_{\mu,\min}(y)| \leq Q_{\Lip} \, d_{\sX}(x,y).
$$
\end{lemma}

\proof
Fix any $\mu \in \P(\sX)$. Let $u: \sX \rightarrow \R$ be $K$-Lipschitz continuous for some $0<K<L_1$. Then $g = u/K$ is $1$-Lipschitz continuous and therefore, for all $a \in \sA$ and $z,y \in \sX$ we have
\begin{align}
\biggl | \sum_{x} u(x) p(x|z,a,\mu) - \sum_{x} u(x) p(x|y,a,\mu) \biggr | 
&= K \biggl | \sum_{x} g(x) p(x|z,a,\mu) - \sum_{x} g(x) p(x|y,a,\mu) \biggr | \nonumber \\
&\leq \frac{K}{2} \, \|p(\,\cdot\,|z,a,\mu) - p(\,\cdot\,|y,a,\mu)\|_1 \nonumber \quad \text{(by (\ref{tv-bound}))}\\
&\leq \frac{KK_1}{2} \, d_{\sX}(z,y), \quad \text{(by Assumption~\ref{as1})} \nonumber
\end{align}
since $\sup_x g(x) - \inf_x g(x) \leq 1$. Hence, the contractive operator $T_{\mu}$ maps a $K$-Lipschitz function $u$ to a $L_1+\beta K K_1/2$-Lipschitz function, indeed, for all $z,y \in \sX$
\begin{align}
| T_{\mu}u(z) &- T_{\mu}u(y) | \nonumber \\
&\leq \sup_{a} \biggl \{ |c(z,a,\mu) - c(y,a,\mu)| + \beta \biggl | \sum_{x} u(x) p(x|z,a,\mu) - \sum_{x} u(x) p(x|y,a,\mu) \biggr | \biggr \}\nonumber \\
&\leq L_1 d_{\sX}(z,y) + \beta \frac{K K_1}{2} d_{\sX}(z,y) = \biggl(L_1 + \beta \frac{K K_1}{2}\biggr) d_{\sX}(z,y). \nonumber
\end{align}
Now we apply $T_{\mu}$ recursively to obtain the sequence $\{T_{\mu}^n u\}$ by letting $T_{\mu}^n u = T_{\mu} (T_{\mu}^{n-1} u )$, which converges to the optimal value function $Q^{*}_{\mu,\min}$ by Banach fixed point theorem. Clearly, by mathematical induction, we have for all $n\geq1$, $T_{\mu}^n u$ is $K_n$-Lipschitz continuous, where $K_n = L_1 \sum_{i=0}^{n-1} (\beta K_1/2)^i + K (\beta K_1/2)^n$. Since $K < L_1$, then $K_n \leq K_{n+1}$ for all $n$ and therefore, $K_n \uparrow Q_{\Lip}$. Hence, $T_{\mu}^n u$ is $Q_{\Lip}$-Lipschitz continuous for all $n$, and therefore, $Q^{*}_{\mu,\min}$ is also  $ Q_{\Lip}$-Lipschitz continuous.
\endproof

Before introducing the mean-field equilibrium (MFE) operator, we first define the set ${\cal C}$ on which the $Q$-functions live. We let ${\cal C}$ consist of all $Q$-functions $Q:\sX \times \sA \rightarrow \R$ such that $Q(x,\cdot)$ is $Q_{\Lip}$-Lipschitz and $\rho$-strongly convex for every $x \in \sX$ with $\|Q\|_{\infty} \leq Q_{\bf m}$, and the gradient $\nabla Q$ of $Q$ with respect to $a$ satisfies the bound 
$$
\sup_{a \in \sA} \|\nabla Q(x,a) - \nabla Q(\hat{x},a)\| \leq K_F \, d_{\sX}(x,\hat{x}),
$$
for every $x, \hat{x} \in  \sX$.

The MFE operator defined as a composition of the operators $H_1$ and $H_2$, where $H_1: \P(\sX) \rightarrow {\cal C}$ is defined as $H_1(\mu) = Q_{\mu}^*$ (optimal $Q$-function for $\mu$), and $H_2: \P(\sX) \times {\cal C} \rightarrow \P(\sX)$ is defined as
\begin{align}
H_2(\mu,Q)(\cdot) \coloneqq \sum_{x \in \sX} p(\cdot|x,f_{Q}(x),\mu) \, \mu(x),\label{h2op}
\end{align}
where $f_{Q}(\,\cdot\,) \coloneqq \argmin_{a\in\sA} Q(\,\cdot\,,a)$ is the unique minimizer  of  $Q \in {\cal C}$ by $\rho$-strong convexity.
Here, $H_1$ computes the optimal $Q$-function given the current state-measure, and $H_2$ computes the next state-measure given the current state-measure and the corresponding optimal $Q$-function. Therefore, the MFE operator is given by
\begin{align}
H: \P(\sX) \ni \mu \mapsto H_2\left(\mu,H_1(\mu)\right) \in \P(\sX).\label{mfeop}
\end{align} 
In this section, our goal is to prove that $H$ is a contraction.

\begin{remark}\label{remh2}
Note that we can alternatively define the operator $H_2$ as a mapping from $\C$ to $\P(\sX)$ as follows: $H_2(Q) = \mu$ if $\mu$ satisfies the following fixed point equation:
$$
\mu(\cdot) \coloneqq \sum_{x \in \sX} p(\cdot|x,f_{Q}(x),\mu) \, \mu(x).
$$
Notice that $H_2$ is a well-defined operator since such a state-measure $\mu \in \P(\sX)$ exists for any $Q$ by Lemma~\ref{auxlemma}. Hence, we may define the MFE operator as $H(\mu) \coloneqq H_2(H_1(\mu))$. In this case, one can prove that this operator has the same contraction coefficient as the original MFE operator given in (\ref{mfeop}). However, although the original $H_2 $ operator in (\ref{h2op}) can effortlessly be approximated via computing the empirical estimate of $p(\cdot|x,f_{Q}(x),\mu)$ for each $x \in \sX$, which is possible since $|\sX| < \infty$, approximating the new $H_2$ operator is quite costly. Indeed, we need to compute a fixed point of some equation in this case. Therefore, there is no advantage to replace original $H_2$ with the new one.   
\end{remark}

In the following lemma, we prove that $H_1$ is Lipschitz continuous, which will later used to prove that $H$ operator is a contraction.

\begin{lemma}\label{n-lemma1}
The mapping $H_1$ is Lipschitz continuous on $\P(\sX)$ with the Lipschitz constant $K_{H_1}$, where 
$$\displaystyle K_{H_1} \coloneqq \frac{Q_{\Lip}}{1-\beta}$$.
\end{lemma}

\proof
First of all, $H_1$ is well-defined; that is, it maps any $\mu \in \P(\sX)$ into ${\cal C}$. Indeed, recall that $Q_{\mu}^*$ is the fixed point of the contractive operator $H_{\mu}$:
\begin{align}
Q_{\mu}^*(x,a) = c(x,a,\mu) + \beta \sum_{y \in \sX} Q_{\mu,\min}^*(y) \, p(y|x,a,\mu). \nonumber
\end{align}
Then, using Assumption~\ref{as1}-(a),(b),(d), it is straightforward to prove that $H_1(\mu) \in {\cal C}$. Indeed, the only non-trivial fact is the $Q_{\Lip}$-Lipschitz continuity of $ Q_{\mu}^{*}$ on $\sX \times \sA$. To this end, let $(x,a), (\hat{x},\hat{a}) \in \sX \times \sA$ be arbitrary. Then, 
\begin{align}
|Q_{\mu}^{*}(x,a)&-Q_{\mu}^{*}(\hat{x},\hat{a})| \nonumber \\
&= |c(x,a,\mu) + \beta \sum_{y} Q_{\mu,\min}^{*}(y) p(y|x,a,\mu) - c(\hat{x},\hat{a},\mu) - \beta \sum_{y} Q_{\mu,\min}^{*}(y) p(y|\hat{x},\hat{a},\mu)| \nonumber \\
&\leq L_1 (d_{\sX}(x,\hat{x})+\|a-\hat{a}\|)  + \beta \frac{K_1 Q_{\Lip}}{2} \, (d_{\sX}(x,\hat{x})+\|a-\hat{a}\|),\nonumber   
\end{align}
where the last inequality follows from (\ref{tv-bound}) and Lemma~\ref{lip-value}. Hence, $Q_{\mu}^{*}$ is $Q_{\Lip}$-Lipschitz continuous.

Now, we prove that $H_1$ is $K_{H_1}$-Lipschitz on $\P(\sX)$. For any $\mu,\hmu \in \P(\sX)$, we have 
\begin{align}
&\|H_1(\mu) - H_1(\hmu)\|_{\infty} = \|Q_{\mu}^*-Q_{\hmu}^*\|_{\infty} \nonumber \\
&= \sup_{x,a} \left| c(x,a,\mu) + \beta \sum_{y} \hspace{-3pt} Q_{\mu,\min}^*(y) p(y|x,a,\mu) - c(x,a,\hmu) - \beta \sum_{y} \hspace{-3pt} Q_{\hmu,\min}^*(y)  p(y|x,a,\hmu) \right| \nonumber \\
&\leq L_1 \, \|\mu-\hmu\|_1 + \beta \left| \sum_{y} Q_{\mu,\min}^*(y) p(y|x,a,\mu) - \sum_{y} Q_{\mu,\min}^*(y) p(y|x,a,\hmu) \right| \nonumber \\
&\phantom{xxxxxxxxxxxxxxxxxxx}+ \beta \left| \sum_{y} Q_{\mu,\min}^*(y) p(y|x,a,\hmu) - \sum_{y} Q_{\hmu,\min}^*(y) p(y|x,a,\hmu) \right| \nonumber \\
&\leq L_1 \, \|\mu-\hmu\|_1 + \beta \, \frac{K_1Q_{\Lip}}{2} \, \|\mu-\hmu\|_1 + \beta \, \|Q_{\mu}^*-Q_{\hmu}^*\|_{\infty}, \nonumber
\end{align}
where the last inequality follows from (\ref{tv-bound}) and Lemma~\ref{lip-value}. 
\endproof

Now, using Lemma~\ref{n-lemma1}, we can prove that $H$ is a contraction on $\P(\sX)$.

\begin{proposition}\label{MFE-con}
The mapping $H$ is a contraction with contraction  on $\P(\sX)$ constant $K_{H}$, where
$$
K_H \coloneqq \frac{3K_1}{2} \left(1+\frac{K_F}{\rho}\right)+\frac{K_1K_FK_{H_1}}{\rho}. 
$$
\end{proposition}

\proof
Fix any $\mu, \hmu \in \P(\sX)$. Note that, since $Q_{\mu}^* = H_{\mu} Q_{\mu}^*$, the mapping $f_{Q_{\mu}^*}(x)$ is the unique minimizer of $F(x,Q_{\mu,\min}^*,\mu,\cdot)$. Similarly, $f_{Q_{\hmu}^*}(y)$ is the unique minimizer of $F(y,Q_{\hmu,\min}^*,\hmu,\cdot)$. For any $x,y \in \sX$, we define 
$a = f_{Q_{\mu}^*}(x)$ and $r = f_{Q_{\hmu}^*}(y) - f_{Q_{\mu}^*}(x).$ As $a$ is the unique minimizer of a strongly convex function $F(x,Q_{\mu,\min}^*,\mu,\cdot)$, by the first-order optimality condition, we have 
$$
\nabla \, F\left(x,Q_{\mu,\min}^*,\mu,a\right) \cdot r \geq 0. 
$$
For $a+r$ and $F(y,Q_{\hmu,\min}^*,\hmu,\cdot)$, by first-order optimality condition, we also have   
$$
\nabla \, F\left(y,Q_{\hmu,\min}^*,\hmu,a+r\right) \cdot r \leq 0. 
$$
Therefore, by $\rho$-strong convexity of $F$ in Assumption~\ref{as1}-(d) and \cite[Lemma 3.2]{HaRa19}, we have
\begin{align}
-\nabla F(y,Q_{\hmu,\min}^*,\hmu,a) \cdot r &\geq -\nabla F(y,Q_{\hmu,\min}^*,\hmu,a) \cdot r + \nabla F(y,Q_{\hmu,\min}^*,\hmu,a+r) \cdot r \nonumber \\
&\geq \rho \, \|r\|^2. \label{st1}
\end{align}
Similarly, by Assumption~\ref{as1}-(d), we also have
\begin{align}
-\nabla F(y,Q_{\hmu,\min}^*,\hmu,a) \cdot r &\leq -\nabla F(y,Q_{\hmu,\min}^*,\hmu,a) \cdot r + \nabla F(x,Q_{\mu,\min}^*,\mu,a) \cdot r \nonumber \\
&\leq \|r\| \, \|\nabla F(x,Q_{\mu,\min}^*,\mu,a)-\nabla F(y,Q_{\hmu,\min}^*,\hmu,a)\| \nonumber \\
&\leq K_F \, \|r\| \left(d_{\sX}(x,y) + \|Q_{\mu,\min}^*-Q_{\hmu,\min}^*\|_{\infty} + \|\mu-\hmu\|_1\right) \nonumber \\
&\leq K_F \, \|r\| \left(d_{\sX}(x,y) + \|Q_{\mu}^*-Q_{\hmu}^*\|_{\infty} + \|\mu-\hmu\|_1\right). \label{st2}
\end{align}
Combining (\ref{st1}) and (\ref{st2}) yields 
$$ \rho \, \|r\|^2 \leq K_F \, \|r\|  \left(d_{\sX}(x,y) + \|Q_{\mu}^*-Q_{\hmu}^*\|_{\infty} + \|\mu-\hmu\|_1\right). $$
Since $r = f_{Q_{\hmu}^*}(y) - f_{Q_{\mu}^*}(x),$ we obtain
\begin{align}
\|f_{Q_{\hmu}^*}(y) - f_{Q_{\mu}^*}(x)\| &\leq \frac{K_F}{\rho} \,\left(d_{\sX}(x,y) + \|Q_{\mu}^*-Q_{\hmu}^*\|_{\infty} + \|\mu-\hmu\|_1 \right) \nonumber \\
&= \frac{K_F}{\rho} \,\left(d_{\sX}(x,y) + \|H_1(\mu)-H_1(\hmu)\|_{\infty} + \|\mu-\hmu\|_1 \right) \nonumber \\
&\leq \frac{K_F}{\rho} \,\left(d_{\sX}(x,y) + K_{H_1} \|\mu-\hmu\|_1 + \|\mu-\hmu\|_1 \right). \label{perturbation}
\end{align}
Therefore, $f_{Q_{\mu}^*}(x)$ is Lipschitz continuous with respect to $(x,\mu)$.

Now, using (\ref{perturbation}), we have
\begin{align}
\|H_2(\mu,H_1(\mu)) - H_2(\hmu,H_1(\hmu))\|_1
&= \sum_{y} \, \bigg| \sum_{x} \, p(y|x,f_{Q_{\mu}^*}(x),\mu),\mu) \, \mu(x) \nonumber \\
&\phantom{xxxxxxxxxxxxxx}- \sum_{x} \, p(y|x,f_{Q_{\hmu}^*}(x),\hmu) \, \hmu(x) \biggr| \nonumber \\
&\leq \sum_{y} \, \bigg| \sum_{x} \, p(y|x,f_{Q_{\mu}^*}(x),\mu) \, \mu(x) \nonumber \\
&\phantom{xxxxxxxxxxxxxx}- \sum_{x} \, p(y|x,f_{Q_{\hmu}^*}(x),\hmu) \, \mu(x) \biggr| \nonumber \\
&+ \sum_{y} \, \bigg| \sum_{x} \, p(y|x,f_{Q_{\hmu}^*}(x),\hmu) \, \mu(x) \nonumber \\
&\phantom{xxxxxxxxxxxxxx}- \sum_{x} \, p(y|x,f_{Q_{\hmu}^*}(x),\hmu) \, \hmu(x) \biggr| \nonumber \\
&\overset{(I)}{\leq} \sum_{x} \left\|p(\cdot|x,f_{Q_{\mu}^*}(x),\mu)-p(\cdot|x,f_{Q_{\hmu}^*}(x),\hmu) \right\|_1 \mu(x) \nonumber \\
&\phantom{xxxxxxxxxxxxxxxxxxx}+ \frac{K_1}{2} \left( 1 + \frac{K_F}{\rho}\right) \, \|\mu-\hmu\|_1 \nonumber \\
&\leq K_1 \left( \|f_{Q_{\mu}^*}(x)-f_{Q_{\hmu}^*}(x)\| + \|\mu-\hmu\|_1 \right) \nonumber \\
&\phantom{xxxxxxxxxxxxxxxxxxx}+ \frac{K_1}{2} \left( 1 + \frac{K_F}{\rho}\right) \, \|\mu-\hmu\|_1 \nonumber \\
&\leq  K_H \, \|\mu-\hmu\|_1.
\end{align}
Note that (\ref{perturbation}) and Assumption~\ref{as1}-(b) lead to
$$
\|p(\cdot|x,f_{Q_{\hmu}^*}(x),\hmu)-p(\cdot|y,f_{Q_{\hmu}^*}(y),\hmu)\|_1 \leq K_1 \left( 1 + \frac{K_F}{\rho}\right).
$$
Hence, (I) follows from Lemma~\ref{KoRa08}. This completes the proof. 
\endproof

\begin{remark}
Note that in the MDP theory, it is normally not required to establish the Lipschitz continuity of the optimal policy. Indeed, the Lipschitz continuity of the optimal value function is in general needed, which can be established straightforward as in Lemma~\ref{lip-value}. However, in mean-field games, since the optimal policy $f_{Q^*_{\mu}}$ directly affects the behaviour of the next state-measure through
$$
	H_2(\mu,Q_{\mu}^{*})(\cdot) = \sum_x p(\cdot|x,f_{Q^*_{\mu}},\mu) \, \mu(x),
$$
one must also establish the Lipschitz continuity of the optimal policy $f_{Q^*_{\mu}}$ in this case. This is indeed the key point in the proof of Proposition~\ref{MFE-con}.

In \cite{AnKaSa20-ar}, we establish the Lipschitz continuity of the optimal policy by introducing a regularization term into a one-stage cost function. This significantly relaxes conditions on the system components in Assumption~\ref{as1}-(d) and simplifies the analysis. However, the regularization term adds some bias to the equilibrium solution (i.e., it in general favours randomized policies over deterministic policies) and also causes the regularized stationary mean-field equilibrium to deviate from true stationary mean-field equilibrium as a result of the additional regularization term in the one-stage cost function. 
\end{remark}

Under Assumption~\ref{as1} and Assumption~\ref{as2}, $H$ is a contraction mapping. Therefore, by the Banach fixed point theorem, $H$ has a unique fixed point. Let $\mu_* \in \P(\sX)$ be this unique fixed point and let $Q_{\mu_*}^{*} = H_1(\mu_*)$. Define the policy $\pi_*(\,\cdot\,|x) = \delta_{f_{Q^*_{\mu_*}}(x)}(\,\cdot\,)$. Then, one can prove that the pair $(\pi_*,\mu_*)$ is a mean-field equilibrium. Indeed, note that $(\mu_*,Q_{\mu_*}^*)$ satisfies the following equations
\begin{align}
\mu_{*}(\cdot) &= \sum_{x \in \sX} p(\cdot|x,a,\mu_*) \, \pi_*(a|x) \, \mu_{*}(x), \label{opt2} \\
Q_{\mu_*}^*(x,a) &= c(x,a,\mu_*) + \beta \sum_{y \in \sX} Q_{\mu_*,\min}^*(y) \, p(y|x,a,\mu_*). \label{opt1} 
\end{align}
Here, (\ref{opt1}) implies that $\pi_* \in \Psi(\mu_*)$ and (\ref{opt2}) implies that $\mu_* \in \Lambda(\pi_*)$. Hence, $(\pi_*,\mu_*)$ is a mean-field equilibrium. Hence, we can compute this mean-field equilibrium via applying $H$ recursively starting from arbitrary state-measure. This indeed leads to a value iteration algorithm for computing mean-field equilibrium. However, if the model is unknown; that is the transition probability $p$ and the one-stage cost function $c$ are not available to the decision maker, we replace $H$ with a random operator and establish a learning algorithm via this random operator. To prove the convergence of this learning algorithm, the contraction property of $H$ is crucial, as stated before.

\section{Finite-Agent Game for Discounted-cost}\label{sec2}

The mean-field game model defined in Section~\ref{sec3} is indeed the infinite-population version of the finite-agent game model with mean-field interactions, which will be described in this section. In this model, there are $N$-agents and for every time step $t \in \{0,1,2,\ldots\}$ and every agent $i \in \{1,2,\ldots,N\}$, $x^N_i(t) \in \sX$ and $a^N_i(t) \in \sA$ denote the state and the action of Agent~$i$ at time $t$, respectively. Moreover, 
\begin{align}
e_t^{(N)}(\,\cdot\,) \coloneqq \frac{1}{N} \sum_{i=1}^N \delta_{x_i^N(t)}(\,\cdot\,) \in \P(\sX) \nonumber
\end{align}
denote the empirical distribution of the agents' states at time $t$. For each $t \ge 0$, next states $(x^N_1(t+1),\ldots,x^N_N(t+1))$ of agents have the following conditional distribution given current states $(x^N_1(t),\ldots,x^N_N(t))$ and actions $(a^N_1(t),\ldots,a^N_N(t))$:
\begin{align}
&(x^N_1(t+1),\ldots,x^N_N(t+1)) \sim \bigotimes^N_{i=1} p\big(\,\cdot\,\big|x^N_i(t),a^N_i(t),e^{(N)}_t\big). \nonumber 
\end{align}
A \emph{policy} $\pi$ for a generic agent in this model is a conditional distribution on $\sA$ given $\sX$; that is, agents can only use their individual states to design their actions. The set of all policies for Agent~$i$ is denoted by $\Pi_i$. Hence, under $\pi \in \Pi_i$, the conditional distribution of the action $a_i^N(t)$ of Agent~$i$ at time $t$ given its state $x_i^N(t)$ is 
$$
a_i^N(t) \sim \pi(\,\cdot\,|x_i^N(t)).
$$ 
Therefore, the information structure of the problem is decentralized. 
The initial states $\{x^N_i(0)\}_{i=1}^N$ are independent and identically distributed according to the initial distribution $\eta_0$.

We let ${\boldsymbol \pi}^{(N)} \coloneqq (\pi^1,\ldots,\pi^N)$, $\pi^i \in \Pi_i$, denote an $N$-tuple of policies for all the agents in the game. Under such an $N$-tuple of policies, for Agent~$i$, the discounted-cost is given by
\begin{align}
J_i^{(N)}({\boldsymbol \pi}^{(N)}) &= E^{{\boldsymbol \pi}^{(N)}}\biggl[\sum_{t=0}^{\infty}\beta^{t}c(x_{i}^N(t),a_{i}^N(t),e^{(N)}_t)\biggr]. \nonumber 
\end{align}
Since agents are coupled through their dynamics and cost functions via the empirical distribution of the states, the problem is indeed a classical game problem. Therefore, the standard notion of optimality is a player-by-player one. 

\begin{definition}
An $N$-tuple of policies ${\boldsymbol \pi}^{(N*)}= (\pi^{1*},\ldots,\pi^{N*})$ constitutes a \emph{Nash equilibrium} if
\begin{align}
J_i^{(N)}({\boldsymbol \pi}^{(N*)}) = \inf_{\pi^i \in \Pi_i} J_i^{(N)}({\boldsymbol \pi}^{(N*)}_{-i},\pi^i) \nonumber
\end{align}
for each $i=1,\ldots,N$, where ${\boldsymbol \pi}^{(N*)}_{-i} \coloneqq (\pi^{j*})_{j\neq i}$.
\end{definition}

We note that obtaining a Nash equilibria is in general prohibitive for finite-agent game model due to the decentralized nature of the information structure of the problem and the large number of agents (see \cite[pp. 4259]{SaBaRaSIAM}). Therefore, it is of interest to seek an approximate Nash equilibrium, whose definition is given below.

\begin{definition}
An $N$-tuple of policies ${\boldsymbol \pi}^{(N*)}= (\pi^{1*},\ldots,\pi^{N*})$ constitutes a \emph{$\delta$-Nash equilibrium} if
\begin{align}
J_i^{(N)}({\boldsymbol \pi}^{(N*)}) \leq \inf_{\pi^i \in \Pi_i} J_i^{(N)}({\boldsymbol \pi}^{(N*)}_{-i},\pi^i) + \delta \nonumber
\end{align}
for each $i=1,\ldots,N$, where ${\boldsymbol \pi}^{(N*)}_{-i} \coloneqq (\pi^{j*})_{j\neq i}$.
\end{definition}

In finite-agent mean-field game model, if the number of agents is large enough, one can obtain a $\delta$-Nash equilibrium by studying the infinite-population limit $N\rightarrow\infty$ of the game (i.e., mean-field game). In the infinite-agent case, the empirical distribution of the states can be modelled as an exogenous state-measure, which should be consistent with the distribution of a generic agent by the law of large numbers (i.e., mean-field equilibrium); that is, a generic agent should solve the mean-field game that is introduced in the preceding section. Then, it is possible to prove that if each agent in the finite-agent $N$ game problem adopts the policy in mean-field equilibrium, the resulting $N$-tuple of policies will be an approximate Nash equilibrium for all sufficiently large $N$. This was indeed proved in \cite{SaBaRaSIAM}.  

%

Note that it is also possible to prove that if each agent in the finite-agent game model adopts the $\varepsilon$-mean-field equilibrium policy, the resulting policy will be also an approximate Nash equilibrium for all sufficiently large $N$-agent game models. Indeed, this is the statement of the next theorem.

But before, let us define the following constants:
\begin{align}
&C_1 \coloneqq \left(\frac{3 \, K_1}{2}  + \frac{K_1 \, K_{F}}{2\rho} \right), \, C_2 \coloneqq \left(L_1+\frac{\beta K_1 Q_{\Lip}}{2}\right) \frac{K_1}{1-C_1}, \,
C_3 \coloneqq \left(L_1+\frac{\beta K_1 Q_{\Lip}}{2}\right).\nonumber
\end{align}
Note that by Assumption~\ref{as2}, the constant $C_1$ is strictly less than $1$.

\begin{theorem}\label{old-main-cor}
Let $\pi_{\varepsilon}$ be an $\varepsilon$-mean-field equilibrium policy for the mean-field equilibrium $(\pi_*,\mu_*) \in \Pi_d \times \P(\sX)$ given by the unique fixed point of the MFE operator $H$. Let $\eta_0 \in \Lambda(\pi_{\varepsilon})$. Then, for any $\delta>0$, there exists a positive integer $N(\delta)$ such that, for each $N\geq N(\delta)$, the $N$-tuple of policies ${\boldsymbol \pi}^{(N)} = \{\pi_{\varepsilon},\pi_{\varepsilon},\ldots,\pi_{\varepsilon}\}$ is a $(\delta+\tau\varepsilon)$-Nash equilibrium for the game with $N$ agents, where $\displaystyle \tau \coloneqq \frac{2C_2+C_3}{1-\beta}$.
\end{theorem}

\proof
By an abuse of notation, we denote the deterministic mappings from $\sX$ to $\sA$ that induce policies $\pi_*$ and $\pi_{\varepsilon}$ as $\pi_*$ and $\pi_{\varepsilon}$ as well, respectively. 
Note that in view of (\ref{perturbation}), one can prove that 
\begin{align}
\|\pi_*(x)-\pi_*(y)\| \leq \frac{K_F}{\rho} \, d_{\sX}(x,y).\label{aks}
\end{align}
Let $\mu_{\varepsilon} \in \Lambda(\pi_{\varepsilon})$. Then, we have
\begin{align}
\|\mu_{\varepsilon}-\mu_*\|_1 
&= \sum_{y} \, \bigg| \sum_{x} \, p(y|x,\pi_{\varepsilon}(x),\mu_{\varepsilon}) \, \mu_{\varepsilon}(x) - \sum_{x} \, p(y|x,\pi_*(x),\mu_*) \, \mu_*(x) \biggr| \nonumber \\
&\leq \sum_{y} \, \bigg| \sum_{x} \, p(y|x,\pi_{\varepsilon}(x),\mu_{\varepsilon}) \, \mu_{\varepsilon}(x) - \sum_{x} \, p(y|x,\pi_*(x),\mu_*) \, \mu_{\varepsilon}(x) \biggr| \nonumber \\
&+ \sum_{y} \, \bigg| \sum_{x} \,p(y|x,\pi_*(x),\mu_*) \, \mu_{\varepsilon}(x) - \sum_{x} \, p(y|x,\pi_*(x),\mu_*) \, \mu_*(x) \biggr| \nonumber \\
&\overset{(I)}{\leq} \sum_{x} \left\|p(\cdot|x,\pi_{\varepsilon}(x),\mu_{\varepsilon})-p(\cdot|x,\pi_*(x),\mu_*) \right\|_1 \mu_{\varepsilon}(x)+ \frac{K_1}{2} \left( 1 + \frac{K_{F}}{\rho} \right) \, \|\mu_{\varepsilon}-\mu_*\|_1 \nonumber \\
&\leq K_1 \left( \sup_x \|\pi_{\varepsilon}(x)-\pi_*(x)\| + \|\mu_{\varepsilon}-\mu_*\|_1 \right) + \frac{K_1}{2} \left( 1 + \frac{K_{F}}{\rho} \right) \, \|\mu_{\varepsilon}-\mu_*\|_1 \nonumber \\
&\leq K_1 \, \varepsilon +  \left(\frac{3 \, K_1}{2}  + \frac{K_1 \, K_{F}}{2\rho} \right) \, \|\mu_{\varepsilon}-\mu_*\|_1. \nonumber 
\end{align}
Note that (\ref{aks}) and Assumption~\ref{as1} lead to
\begin{align}
&\|p(\cdot|x,\pi_*(x),\mu_*)-p(\cdot|y,\pi_*(y),\mu_*)\|_1 \leq K_1 \left( 1 + \frac{K_{F}}{\rho} \right) \, d_{\sX}(x,y). \nonumber
\end{align}
Hence, (I) follows from Lemma~\ref{KoRa08}. Therefore, we have 
$$
\|\mu_{\varepsilon}-\mu_*\|_1 \leq \frac{K_1 \, \varepsilon}{1-C_1}.
$$ 

Now, fix any policy $\pi \in \Pi_d$. For any state-measure $\mu$, it is a well-known fact in MDP theory that the value function $J_{\mu}(\pi,\cdot)$ of $\pi$ satisfies the following fixed point equation:
$$
J_{\mu}(\pi,x) = c(x,\pi(x),\mu) + \beta \sum_y J_{\mu}(\pi,y) \, p(y|x,\pi(x),\mu),
$$
for every $x \in \sX.$ Therefore, we have 
\begin{align}
&\|J_{\mu_*}(\pi,\cdot)-J_{\mu_{\varepsilon}}(\pi,\cdot)\|_{\infty} \nonumber \\
&=\sup_{x} \bigg| c(x,\pi(x),\mu_*) + \beta \, \sum_{y} J_{\mu_*}(\pi,y) \, p(y|x,\pi(x),\mu_*) \nonumber \\
&\phantom{xxxxxxxxxx}-c(x,\pi(x),\mu_{\varepsilon}) - \beta \, \sum_{y} J_{\mu_{\varepsilon}}(\pi,y) \, p(y|x,\pi(x),\mu_{\varepsilon})\bigg| \nonumber \\
&\leq L_1 \, \|\mu_*-\mu_{\varepsilon}\|_1 + \beta \sup_{x} \bigg|\sum_{y} J_{\mu_*}(\pi,y) \, p(y|x,\pi(x),\mu_*) - \sum_{y} J_{\mu_*}(\pi,y) \, p(y|x,\pi(x),\mu_{\varepsilon})\bigg| \nonumber \\
&\phantom{xxxxxxxxxxxxx}+ \beta \sup_{x} \bigg|\sum_{y} J_{\mu_*}(\pi,y) \, p(y|x,\pi(x),\mu_{\varepsilon}) - \sum_{y} J_{\mu_{\varepsilon}}(\pi,y) \, p(y|x,\pi(x),\mu_{\varepsilon})\bigg| \nonumber \\
&\overset{(II)}{\leq} \left(L_1+\frac{\beta K_1 Q_{\Lip}}{2}\right) \|\mu_*-\mu_{\varepsilon}\|_1 + \beta \|J_{\mu_*}(\pi,\cdot)-J_{\mu_{\varepsilon}}(\pi,\cdot)\|_{\infty} \nonumber \\
&\leq  \left(L_1+\frac{\beta K_1 Q_{\Lip}}{2}\right) \frac{K_1  \varepsilon}{1-C_1}+ \beta \|J_{\mu_*}(\pi,\cdot)-J_{\mu_{\varepsilon}}(\pi,\cdot)\|_{\infty}. \nonumber 
\end{align}
Here (II) follows from (\ref{tv-bound}) and the fact that $J_{\mu_*}(\pi,\cdot)$ is $Q_{\Lip}$-Lipschitz continuous, which can be proved as in Lemma~\ref{lip-value}.  
Therefore, we obtain 
\begin{align}\label{nneqq1}
\|J_{\mu_*}(\pi,\cdot)-J_{\mu_{\varepsilon}}(\pi,\cdot)\|_{\infty} \leq \frac{C_2 \, \varepsilon}{1-\beta}.
\end{align}

Similarly, we also have  
\begin{align}
\|J_{\mu_*}(\pi_*,\cdot)-J_{\mu_*}(\pi_{\varepsilon},\cdot)\|_{\infty} &=\sup_{x} \bigg| c(x,\pi_*(x),\mu_*) + \beta \, \sum_{y} J_{\mu_*}(\pi_*,y) \, p(y|x,\pi_*(x),\mu_*) \nonumber \\
&\phantom{xxxxxxxx}-c(x,\pi_{\varepsilon}(x),\mu_*) - \beta \, \sum_{y} J_{\mu_*}(\pi_{\varepsilon},y) \, p(y|x,\pi_{\varepsilon}(x),\mu_*)\bigg| \nonumber \\
&\leq L_1  \, \sup_{x} \|\pi_*(x)-\pi_{\varepsilon}(x)\| + \beta \sup_{x} \bigg|\sum_{y} J_{\mu_*}(\pi_*,y) \, p(y|x,\pi_*(x),\mu_*) \nonumber \\
&\phantom{xxxxxxxxxxxxxxx}- \sum_{y} J_{\mu_*}(\pi_*,y) \, p(y|x,\pi_{\varepsilon}(x),\mu_*)\bigg| \nonumber \\
&+ \beta \sup_{x} \bigg|\sum_{y} J_{\mu_*}(\pi_*,y) \, p(y|x,\pi_{\varepsilon}(x),\mu_*) \nonumber \\
&\phantom{xxxxxxxxxxxxxxx}- \sum_{y} J_{\mu_*}(\pi_{\varepsilon},y) \, p(y|x,\pi_{\varepsilon}(x),\mu_*)\bigg| \nonumber \\
&\overset{(III)}{\leq}  \left(L_1+\frac{\beta K_1 Q_{\Lip}}{2}\right) \sup_x \|\pi_*(x)-\pi_{\varepsilon}(x)\| \nonumber \\
&\phantom{xxxxxxxxxxxxxxxxxxxxx}+ \beta \|J_{\mu_*}(\pi_*,\cdot)-J_{\mu_*}(\pi_{\varepsilon},\cdot)\|_{\infty} \nonumber \\
&\leq  \left(L_1+\frac{\beta K_1 Q_{\Lip}}{2}\right) \, \varepsilon+ \beta \|J_{\mu_*}(\pi_*,\cdot)-J_{\mu_*}(\pi_{\varepsilon},\cdot)\|_{\infty}. \nonumber 
\end{align}
\normalsize
Here (III) follows from (\ref{tv-bound}) and the fact that $J_{\mu_*}(\pi_*,\cdot)$ is $Q_{\Lip}$-Lipschitz continuous, which can be proved as in Lemma~\ref{lip-value}.
Therefore, we obtain 
\begin{align}\label{nneqq2}
\|J_{\mu_*}(\pi_*,\cdot)-J_{\mu_*}(\pi_{\varepsilon},\cdot)\|_{\infty} \leq \frac{C_3\varepsilon}{1-\beta}, 
\end{align}
where $C_3 \coloneqq \left(L_1+\frac{\beta K_1 Q_{\Lip}}{2}\right)$. 

Note that we must prove that
\begin{align}
J_i^{(N)}({\boldsymbol \pi}^{(N)}) &\leq \inf_{\pi^i \in \Pi_i} J_i^{(N)}({\boldsymbol \pi}^{(N)}_{-i},\pi^i) + \tau \, \varepsilon + \delta \label{old-eq13}
\end{align}
for each $i=1,\ldots,N$, when $N$ is sufficiently large. As the transition probabilities and the one-stage cost functions are the same for all agents, it is sufficient to prove (\ref{old-eq13}) for Agent~$1$ only. Given $\delta > 0$, for each $N\geq1$, let $\tpi^{(N)} \in \Pi_1$ be a deterministic policy such that
\begin{align*}
J_1^{(N)} (\tpi^{(N)},\pi_{\varepsilon},\ldots,\pi_{\varepsilon}) < \inf_{\pi' \in \Pi_1} J_1^{(N)} (\pi',\pi_{\varepsilon},\ldots,\pi_{\varepsilon}) + \frac{\delta}{3}. 
\end{align*}
On the other hand, by \cite[Theorem 4.10]{SaBaRaSIAM}
\begin{align}
\lim_{N\rightarrow\infty} J_1^{(N)} (\tpi^{(N)},\pi_{\varepsilon},\ldots,\pi_{\varepsilon}) &= \lim_{N\rightarrow\infty} J_{\mu_{\varepsilon}}(\tpi^{(N)}) \nonumber \\
&\geq \lim_{N\rightarrow\infty} J_{\mu_*}(\tpi^{(N)}) - \frac{C_2 \, \varepsilon}{1-\beta} \quad \text{(by (\ref{nneqq1}))}\nonumber \\
&\geq \inf_{\pi' \in \Pi_d} J_{\mu_*}(\pi') - \frac{C_2 \, \varepsilon}{1-\beta} \nonumber \\
&= J_{\mu_*}(\pi_*) - \frac{C_2 \, \varepsilon}{1-\beta} \nonumber \\
&\geq J_{\mu_*}(\pi_{\varepsilon}) - \frac{C_2 \, \varepsilon}{1-\beta} - \frac{C_3 \, \varepsilon}{1-\beta} \quad \text{(by (\ref{nneqq2}))}\nonumber \\
&\geq J_{\mu_{\varepsilon}}(\pi_{\varepsilon}) - \frac{2 \, C_2 \, \varepsilon}{1-\beta} - \frac{C_3 \, \varepsilon}{1-\beta} \quad \text{(by (\ref{nneqq1}))}\nonumber \\
&\eqqcolon J_{\mu_{\varepsilon}}(\pi_{\varepsilon}) - \tau \, \varepsilon \nonumber.
\end{align}
Note that by \cite[Theorem 4.10]{SaBaRaSIAM}, we also have 
$$
\lim_{N\rightarrow\infty} J_1^{(N)} (\pi_{\varepsilon},\pi_{\varepsilon},\ldots,\pi_{\varepsilon}) = J_{\mu_{\varepsilon}}(\pi_{\varepsilon}).
$$
Hence, there exists $N(\delta)$ such that for all $N\geq N(\delta)$, we have
\begin{align}
J_1^{(N)} (\tpi^{(N)},\pi_{\varepsilon},\ldots,\pi_{\varepsilon}) + \frac{\delta}{3} &\geq J_{\mu_{\varepsilon}}(\pi_{\varepsilon})- \tau \, \varepsilon \nonumber \\
J_{\mu_{\varepsilon}}(\pi_{\varepsilon}) + \frac{\delta}{3}&\geq J_1^{(N)} (\pi_{\varepsilon},\pi_{\varepsilon},\ldots,\pi_{\varepsilon}). \nonumber  
\end{align}
Therefore, for all $N \geq N(\delta)$, we obtain
\begin{align}
\inf_{\pi' \in \Pi_1} J_1^{(N)} (\pi',\pi_{\varepsilon},\ldots,\pi_{\varepsilon}) + \delta + \tau \,\varepsilon 
&\geq J_1^{(N)} (\tpi^{(N)},\pi_{\varepsilon},\ldots,\pi_{\varepsilon}) + \frac{2\delta}{3} + \tau \, \varepsilon \nonumber \\
&\geq J_{\mu_{\varepsilon}}(\pi_{\varepsilon}) + \frac{\delta}{3}   \nonumber \\
&\geq J_1^{(N)} (\pi_{\varepsilon},\pi_{\varepsilon},\ldots,\pi_{\varepsilon}). \nonumber
\end{align}
\endproof

Theorem~\ref{old-main-cor} implies that, by learning $\varepsilon$-mean-field equilibrium policy in the infinite-population limit, one can obtain an approximate Nash equilibrium for the finite-agent game problem for which computing or learning the exact Nash equilibrium is in general  prohibitive. 

In the next section, we approximate the MFE operator $H$ introduced in Section~\ref{known-model} via the random operator $\hat{H}$ to develop an algorithm for learning a $\varepsilon$-mean-field equilibrium policy in the model-free setting.

\section{Learning Algorithm for Discounted-cost}\label{unknown-model}

In this section, we develop an offline learning algorithm to learn an approximate mean-field equilibrium policy. 
To this end, we suppose that a generic agent has access to a simulator, which generates a new state $y \sim p(\,\cdot\,|x,a,\mu)$ and gives the cost $c(x,a,\mu)$ for any given state $x$, action $a$, and state measure $\mu$. This is a typical assumption in offline reinforcement learning algorithms.

Each iteration of our learning algorithm has two stages. Using a fitted $Q$-iteration algorithm, we learn the optimal  $Q$-function $Q_{\mu}^*$ for a given state-measure $\mu$ in the first stage by replacing $H_1$ with a random operator $\hat{H}_1$. The $Q$-functions are selected from a fixed function class ${\cal F}$ which can be defined as a collection of neural networks with a specific architecture or a linear span of a finite number of basis functions. There will be an additional representation error in the learning algorithm depending on this choice, which is generally negligible since $Q$-functions in ${\cal C}$ can be well approximated by functions from ${\cal F}$. 

In the second stage, we update the state-measure by approximating the  transition probability via its empirical estimate by replacing $H_2$ with a random operator $\hat{H}_2$. It should be noted that if the alternative $H_2$ operator mentioned in Remark~\ref{remh2} is used, the random operator that approximates this alternative  $H_2$ operator would be more complicated than $\hat{H}_2$. Indeed, in this case, an empirical estimation of the transition probability might be insufficient.

We proceed by introducing the random operator $\hat{H}_1$.
To describe $\hat{H}_1$, we need to give some definitions. Let $m_{\sA}(\cdot) \coloneqq m(\cdot)/m(\sA)$ be the uniform probability measure on $\sA$. Let us choose a probability measure $\nu$ on $\sX$ such that $\min_x \nu(x) > 0$. For instance, one can choose $\nu$ as the uniform distribution over $\sX$. Define $\zeta_0 \coloneqq 1/\sqrt{\min_x \nu(x)}$. We also choose some policy $\pi_b$ such that, for any $x \in \sX$, the distribution $\pi_b(\cdot|x)$ on $\sA$ has density with respect to Lebesgue measure $m$. To simplify the notation, we denote this density by $\pi_b(a|x)$, and assume that it satisfies $\pi_0 \coloneqq \inf_{(x,a) \in \sX\times\sA} \pi_b(a|x) > 0$. Note that the randomized policy $\pi_b$ is used to generate data for the learning algorithm below. In general, given any mean-field term, it is enough to consider deterministic policies for optimality. However, as is typical in reinforcement learning, we employ randomized policies to explore the action space in the training stage. Because of this, $\pi_b$ is introduced in a stochastic manner. We can now define the random operator $\hat{H}_1$.

\begin{algorithm}[H]
\caption{Algorithm $\hat{H}_1$}
\label{H1}
\begin{algorithmic}
\STATE{Input $\mu$, Data size $N$, Number of iterations $L$}
\STATE{Generate i.i.d. samples $\{(x_t,a_t,c_t,y_{t+1})_{t=1}^N\}$ using
$$
x_t \sim \nu, \, a_t \sim \pi_b(\cdot|x_t), \, c_t = c(x_t,a_t,\mu), \, y_{t+1} \sim p(\cdot|x_t,a_t,\mu)
$$
}
\STATE{Start with $Q_0 = 0$}
\FOR{$l=0,\ldots,L-1$}
\STATE{
$$
Q_{l+1} = \argmin_{f \in {\cal F}} \frac{1}{N} \sum_{t=1}^N \frac{1}{m(\sA) \, \pi_b(a_t|x_t)} \left| f(x_t,a_t) - \left[c_t + \beta \min_{a' \in \sA} Q_l(y_{t+1},a') \right]\right|^2
$$
}
\ENDFOR
\RETURN{$Q_L$}
\end{algorithmic}
\end{algorithm}

\begin{remark}
Notice that in Algorithm~$\hat{H_1}$, we use the distribution $\nu$ and policy $\pi_b$ to build an i.i.d. dataset. In fact, instead of using i.i.d. samples, one can use a sample path  $\{x_t,a_t\}_{t=1}^N$ generated by the policy $\pi_b$ instead of using i.i.d. samples by setting $c_t = c(x_t,a_t,\mu)$ and $y_{t+1} = x_{t+1}$. Then, in order to establish the error analysis, we have to assume that under $\pi_b$, the state process $\{x_t\}$ must be strictly stationary and exponentially $\beta$-mixing (see \cite{AnMuSz07}). The main issue in this case, however, is finding a policy $\pi_b$ that meets the mixing condition. Indeed, since exponentially $\beta$-mixing stationary processes forget their history exponentially fast, they behave like i.i.d. processes when there is a sufficiently large time gap between two samples. As a result, the error analysis of the exponential $\beta$-mixing case is very close to that of the i.i.d. case. For more information on the error analysis of $\hat{H}_1$ in the exponentially $\beta$-mixing case, see \cite{AnMuSz07-t,AnMuSz07}.
\end{remark}

We perform an error analysis of the algorithm $\hat{H}_1$ before defining the second stage $\hat{H}_2$. To that end, we define the $l_2$-norm of any $g: \sX \times \sA \rightarrow \R$ as 
$$
\|g\|_{\nu}^2 \coloneqq \sum_{x \in \sX} \int_{\sA} g(x,a)^2 \, m_{\sA}(da) \, \nu(x),
$$
and introduce the constants
\begin{align}
&E({\cal F}) \coloneqq \sup_{\mu \in \P(\sX)} \sup_{Q \in {\cal F}} \inf_{Q' \in {\cal F}} \|Q'-H_{\mu}Q\|_{\nu}, 
\end{align}
and 
\begin{align}
&L_{\bf m} \coloneqq (1+\beta) Q_{\bf m} + c_{\bf m}, \,\,\, C \coloneqq \frac{L_{\bf m}^2}{m(\sA) \, \pi_0}, \,\,\, \gamma = 512 C^2. \nonumber 
\end{align}

Here $E({\cal F})$ describes the representation error of the function class ${\cal F}$. This error is generally small since every $Q$ function in  ${\cal C}$ can be very well approximated using, for example, neural networks with a fixed architecture. As a result, we may consider the error caused by $E({\cal F})$ to be negligible. The error analysis of the algorithm  $\hat{H}_1$ is given by the following theorem. 
We define ${\cal F}_{\min} \coloneqq \{Q_{\min}: Q \in {\cal F}\}$ and let
\begin{align}
\Upsilon = 8 \, e^2 \, (V_{{\cal F}}+1) \, (V_{{\cal F}_{\min}}+1) \, \left(\frac{64 e Q_{\bf m} L_{\bf m} (1+\beta)}{m(\sA) \pi_0}\right)^{V_{{\cal F}}+V_{{\cal F}_{\min}}}, \,\, V = V_{{\cal F}}+V_{{\cal F}_{\min}}.  \nonumber 
\end{align}

\begin{theorem}\label{Thm-H1}
For any $(\varepsilon,\delta) \in (0,1)^2$ and $N \geq m_1(\epsilon,\delta,L)$, with probability at least $1-\delta$ we have
$$
\left\|\hat{H}_1[N,L](\mu) - H_1(\mu)\right\|_{\infty} \leq \varepsilon + \Delta,
$$
if $\frac{\beta^L}{1-\beta} \, Q_{\bf m} < \frac{\varepsilon}{2}$, where
$$
m_1(\varepsilon,\delta,L) \coloneqq \frac{\gamma (2 \Lambda)^{4(\dim_{\sA}+1)}}{\varepsilon^{4(\dim_{\sA}+1)}} \, \ln\left(\frac{\Upsilon (2 \Lambda)^{2V(\dim_{\sA}+1)} L}{\delta \varepsilon^{2V(\dim_{\sA}+1)}}\right),
$$
and 
$$ \Delta \coloneqq \frac{1}{1-\beta} \left[ \frac{m(\sA) (\dim_{\sA}+1)! \zeta_0}{\alpha (2/Q_{\Lip})^{\dim_{\sA}}} \, E({\cal F}) \right]^{\frac{1}{\dim_{\sA}+1}}, \,\, \Lambda \coloneqq \frac{1}{1-\beta} \left[ \frac{ m(\sA) (\dim_{\sA}+1)! \zeta_0}{\alpha (2/Q_{\Lip})^{\dim_{\sA}}} \right]^{\frac{1}{\dim_{\sA}+1}}. \nonumber 
$$

The constant error $\Delta$  is due to the algorithm's representation error $E({\cal F})$, which is generally negligible.

\end{theorem}

\proof
For any real-valued function  $Q(x,a)$, recall the definition 
$$
\|Q\|_{\nu}^2 \coloneqq \sum_{x \in \sX} \int_{\sA} Q(x,a)^2 \, m_{\sA}(da) \, \nu(x). 
$$
Let $Q_l$ be the random $Q$-function at the $l^{\text{th}}$-step of the algorithm. First, we find an upper bound to the following probability
$$
P_0 \coloneqq \rP\left( \|Q_{l+1}-H_{\mu} Q_l\|_{\nu}^2 > E({\cal F})^2 + \varepsilon' \right),
$$
for a given $\varepsilon'>0$. To that end, we define
$$
\hat{L}_N(f;Q) \coloneqq \frac{1}{N} \sum_{t=1}^N \frac{1}{m(\sA) \, \pi_b(a_t|x_t)} \left| f(x_t,a_t) - \left[c_t + \beta \min_{a' \in \sA} Q(y_{t+1},a') \right]\right|^2.
$$ 
The normalization with $\pi_b(a_t|x_t)$ is used here to avoid assigning more weight to the actions that are preferred by the policy, and $m(\sA)$ is introduced for mathematical convenience.

One can show that (see \cite[Lemma 4.1]{AnMuSz07-t}) 
$$
\rE \left[ \hat{L}_N(f;Q) \right] = \|f-H_{\mu}Q\|_{\nu}^2 + L^*(Q) \eqqcolon L(f;Q),
$$
where $L^*(Q)$ is some quantity independent of $f$.  Since we need a similar equation for the average-cost, let us prove it in detail so that we can refer to this proof in the future. Indeed, for each $t=1,\ldots,N$, define 
$$
\hat{Q}_t \coloneqq  c_t + \beta \min_{a' \in \sA} Q(y_{t+1},a'). 
$$ 
Then, 
$$
\rE\left[ \hat{Q}_t \, \big| \, x_t,a_t \right] = H_{\mu} Q(x_t,a_t). 
$$

Note that we can write 
\begin{align}
&\rE\left[\left( f(x_t,a_t) - \left[c_t + \beta \min_{a' \in \sA} Q(y_{t+1},a') \right]\right)^2 \bigg| \, x_t,a_t \right] 
= \rE\left[\left(f(x_t,a_t)-\hat{Q}_t\right)^2 \bigg| \, x_t,a_t \right] \nonumber \\
&\phantom{xxxxxxxxxx}= \rE\left[\left(\hat{Q}_t-H_{\mu} Q(x_t,a_t)\right)^2 \bigg| \, x_t,a_t \right] + \left( f(x_t,a_t) - H_{\mu} Q(x_t,a_t) \right)^2. \nonumber
\end{align}
Dividing each term by $m(\sA) \, \pi_b(a_t|x_t)$, taking the expectation of both sides with respect to $a$ and $x$, and using the \emph{law of iterated expectation} we get 
\begin{align}
&\rE\left[\frac{\left( f(x_t,a_t) - \left[c_t + \beta \min_{a' \in \sA} Q(y_{t+1},a') \right]\right)^2}{m(\sA) \, \pi_b(a_t|x_t)}\right] \nonumber \\ 
&\phantom{xxxx}= \rE\left[\frac{\left(\hat{Q}_t-H_{\mu} Q(x_t,a_t)\right)^2}{m(\sA) \, \pi_b(a_t|x_t)}\right] + \sum_{x_t \in \sX} \int_{\sA} \frac{\left(  f(x_t,a_t) - H_{\mu} Q(x_t,a_t) \right)^2}{m(\sA) \, \pi_b(a_t|x_t)}  \, \pi_b(a_t|x_t) \, m(da_t) \, \nu(x_t) \nonumber \\
&\phantom{xxxx}= \rE\left[\frac{\left(\hat{Q}_t-H_{\mu} Q(x_t,a_t)\right)^2}{m(\sA) \, \pi_b(a_t|x_t)}\right] + \sum_{x_t \in \sX} \int_{\sA} \left(  f(x_t,a_t) - H_{\mu} Q(x_t,a_t) \right)^2  \, m_{\sA}(da_t) \, \nu(x_t) \nonumber \\
&\phantom{xxxx}\eqqcolon L^*(Q) + \|f-H_{\mu}Q\|_{\nu}^2 \eqqcolon L(f;Q). \nonumber 
\end{align}
Since the samples are i.i.d., this establishes the fact.

Using above discussion, we can obtain the following bound 
\begin{align}
\|Q_{l+1} - H_{\mu} Q_l\|_{\nu}^2 - E({\cal F})^2 &\leq \|Q_{l+1} - H_{\mu} Q_l\|_{\nu}^2 - \inf_{f \in {\cal F}} \|f-H_{\mu}Q_l\|_{\nu}^2 \nonumber \\
&= L(Q_{l+1};Q_l) - \inf_{f \in {\cal F}} L(f;Q_l) \nonumber \\
&= L(Q_{l+1};Q_l) - \hat{L}_N(Q_{l+1};Q_l) +\hat{L}_N(Q_{l+1};Q_l) - \inf_{f \in {\cal F}} L(f;Q_l) \nonumber \\
&= L(Q_{l+1};Q_l) - \hat{L}_N(Q_{l+1};Q_l) + \inf_{f \in {\cal F}} \hat{L}_N(f;Q_l) - \inf_{f \in {\cal F}} L(f;Q_l) \nonumber \\
&\leq 2 \sup_{f \in {\cal F}} \left| L(f;Q_l) - \hat{L}_N(f;Q_l) \right| \nonumber \\
&\leq 2 \sup_{f,Q \in {\cal F}} \left| L(f;Q) - \hat{L}_N(f;Q) \right|. \nonumber 
\end{align}
This implies that 
\begin{align}
P_0 \leq \rP \left( \sup_{f,Q \in {\cal F}} \left| L(f;Q) - \hat{L}_N(f;Q) \right| > \frac{\varepsilon'}{2} \right). \label{first} 
\end{align}
For any $f,Q \in {\cal F}$, we define
$$
l_{f,Q}(x,a,c,y) \coloneqq \frac{1}{m(\sA) \pi_b(a|x)} \left| f(x,a) - c - \beta \min_{a' \in \sA} Q(y,a') \right|^2.
$$
Let ${\cal L}_{{\cal F}} \coloneqq \{l_{f,Q}: f,Q \in {\cal F}\}$. Note that $\{z_t\}_{t=1}^N \coloneqq \{(x_t,a_t,c_t,y_{t+1})\}_{t=1}^N$ are i.i.d. and 
$$
\frac{1}{N} \sum_{t=1}^N l_{f,Q}(z_t) = \hat{L}_N(f;Q) \,\, \text{and} \,\, \rE[l_{f,Q}(z_1)] = L(f;Q).
$$
Recall the constant $L_{\bf m} \coloneqq (1+\beta) Q_{\bf m} + c_{\bf m}$. One can prove that $0 \leq l_{f,Q} \leq \frac{L_{\bf m}^2}{m(\sA) \, \pi_0} \eqqcolon C$. Then, by Pollard's Tail Inequality \cite[Theorem 24, p. 25]{Pol84}, we have
\begin{align}
P_0 &\leq \rP \left( \sup_{f,Q \in {\cal F}} \left| \frac{1}{N} \sum_{t=1}^N l_{f,Q}(z_t) - \rE[l_{f,Q}(z_1)] \right| > \frac{\varepsilon'}{2} \right) \nonumber \\
&\leq 8 \, \rE\left[ N_1\left(\frac{\varepsilon'}{16},\{z_t\}_{t=1}^N,{\cal L}_{{\cal F}}\right) \right] \, e^{\frac{-N \, \varepsilon'^2}{512 \, C^2}}. \nonumber 
\end{align}
For any $l_{f,Q}$ and $l_{g,T}$,  we also have  (see \cite[pp. 18]{AnMuSz07-t})
\begin{align}
&\frac{1}{N} \sum_{t=1}^N |l_{f,Q}(z_t) - l_{g,T}(z_t)| \leq \frac{2 L_{\bf m}}{m(\sA) \, \pi_0} \bigg( \frac{1}{N} \sum_{t=1}^N |f(x_t,a_t) - g(x_t,a_t)| \nonumber \\
&\phantom{xxxxxxxxxxxxxxxxxxxxxxxxxx}+ \beta  \frac{1}{N} \sum_{t=1}^N \left|\min_{b \in \sA} Q(y_{t+1},b) - \min_{b \in \sA} T(y_{t+1},b)\right| \bigg). \nonumber
\end{align}
This implies that, for any $\epsilon > 0$, we have 
\begin{align}
&N_1\left(\frac{2 L_{\bf m}}{m(\sA) \, \pi_0}\,(1+\beta)\,\epsilon,\{z_t\}_{t=1}^N,{\cal L}_{{\cal F}}\right) 
\leq N_1\left(\epsilon,\{(x_t,a_t)\}_{t=1}^N,{\cal F}\right) \, N_1\left(\epsilon,\{y_{t+1}\}_{t=1}^N,{\cal F}_{\min}\right) \nonumber \\
&\phantom{xxxxxxxxxxxxxxxxx}\overset{(I)}{\leq} e \, (V_{{\cal F}}+1) \, \left(\frac{2 e Q_{\bf m}}{\epsilon}\right)^{V_{{\cal F}}} \, e \, (V_{{\cal F}_{\min}}+1) \, \left(\frac{2 e Q_{\bf m}}{\epsilon}\right)^{V_{{\cal F}_{\min}}}, 
\end{align}
where (I) follows from Lemma~\ref{cov-num}. Therefore, we have the following bound on the probability $P_0$:
\begin{align}
P_0 \leq 8 \, \left\{ e^2 \, (V_{{\cal F}}+1) \, (V_{{\cal F}_{\min}}+1) \, \left(\frac{64 e Q_{\bf m} L_{\bf m} (1+\beta)}{m(\sA) \pi_0 \varepsilon'}\right)^{V_{{\cal F}}+V_{{\cal F}_{\min}}} \right\} \, e^{\frac{-N \, \varepsilon'^2}{512 \, C^2}}. \label{second}
\end{align}
Recall the constants
\small
\begin{align}
\Upsilon &= 8 \, e^2 \, (V_{{\cal F}}+1) \, (V_{{\cal F}_{\min}}+1) \, \left(\frac{64 e Q_{\bf m} L_{\bf m} (1+\beta)}{m(\sA) \pi_0}\right)^{V_{{\cal F}}+V_{{\cal F}_{\min}}},  
V &= V_{{\cal F}}+V_{{\cal F}_{\min}}, \,\,\, \gamma = 512 C^2.\nonumber
\end{align}
\normalsize
Then, we can write (\ref{second}) as follows
\begin{align}
P_0 \coloneqq \rP\left( \|Q_{l+1}-H_{\mu} Q_l\|_{\nu}^2 > E({\cal F})^2 + \varepsilon' \right) \leq \Upsilon \, \varepsilon'^{-V} \, e^{\frac{-N \varepsilon'^2}{\gamma}} \eqqcolon \frac{\delta'}{L}. \label{third}
\end{align}
Hence, for each $l=0,\ldots,L-1$, with probability at most $\frac{\delta'}{L}$
$$
\|Q_{l+1} - H_{\mu} Q_l\|_{\nu}^2 > \varepsilon' + E({\cal F})^2.
$$
This implies that with probability at most $\frac{\delta'}{L}$
$$
\|Q_{l+1} - H_{\mu} Q_l\|_{\nu} > \sqrt{\varepsilon'} + E({\cal F}).
$$
Using this, we can conclude that with probability at least $1-\delta'$
\begin{align}
&\|Q_L - H_1(\mu)\|_{\infty} \leq \sum_{l=0}^{L-1} \beta^{L-(l+1)} \, \|Q_{l+1} - H_{\mu} Q_l\|_{\infty} + \|H_{\mu}^L Q_0 - H_1(\mu)\|_{\infty} \nonumber \\
&\overset{(II)}{\leq} \sum_{l=0}^{L-1} \beta^{L-(l+1)} \, \left[ \frac{m(\sA) (\dim_{\sA}+1)!\zeta_0}{\alpha (2/Q_{\Lip})^{\dim_{\sA}}} \, \|Q_{l+1} - H_{\mu} Q_l\|_{\nu} \right]^{\frac{1}{\dim_{\sA}+1}} + \frac{\beta^L}{1-\beta} \, Q_{\bf m} \nonumber \\
&\leq \sum_{l=0}^{L-1} \beta^{L-(l+1)} \, \left[ \frac{ m(\sA) (\dim_{\sA}+1)!\zeta_0}{\alpha (2/Q_{\Lip})^{\dim_{\sA}}} \, (\sqrt{\varepsilon'}+E({\cal F})) \right]^{\frac{1}{\dim_{\sA}+1}} + \frac{\beta^L}{1-\beta} \, Q_{\bf m} \nonumber \\
&\leq \frac{1}{1-\beta} \left( \left[ \frac{m(\sA) (\dim_{\sA}+1)!\zeta_0}{\alpha (2/Q_{\Lip})^{\dim_{\sA}}} \, E({\cal F}) \right]^{\frac{1}{\dim_{\sA}+1}} +  \left[ \frac{ m(\sA) (\dim_{\sA}+1)!\zeta_0}{\alpha (2/Q_{\Lip})^{\dim_{\sA}}} \right]^{\frac{1}{\dim_{\sA}+1}} \varepsilon'^{\frac{1}{2(\dim_{\sA}+1)}}\right) \nonumber \\
&\phantom{xxxxxxxxxxxxxxxxxxxxxxxxxxxxxxxxxxxxxxxxx}+\frac{\beta^L}{1-\beta} \, Q_{\bf m}, \nonumber
\end{align}
where (II) follows from Lemma~\ref{li-l2}.
Then, with probability at least $1-\delta'$, we have
\begin{align}
\|Q_L - H_1(\mu)\|_{\infty} \leq \Lambda \varepsilon'^{\frac{1}{2(\dim_{\sA}+1)}} + \Delta + \frac{\beta^L}{1-\beta} \, Q_{\bf m}. \label{fourth}
\end{align}
The result follows by picking $\delta = \delta' \coloneqq L \, \Upsilon \, \varepsilon'^{-V} \, e^{\frac{-N \varepsilon'^2}{\gamma}}$ in (\ref{third}), choosing $\Lambda \varepsilon'^{\frac{1}{2(\dim_{\sA}+1)}} = \varepsilon/2$, and $\frac{\beta^L}{1-\beta} \, Q_{\bf m} = \varepsilon/2$.  
\endproof

\begin{remark}
We use the $\|\,\cdot\,\|_{\nu}$-norm on $Q$-functions until a certain stage in the proof of Theorem~\ref{Thm-H1}, and then we use  Lemma~\ref{li-l2} to go back to the $\|\,\cdot\,\|_{\infty}$-norm.
Notice that Assumption~\ref{as1}-(c) on $\sA$ is needed to accomplish this because the operator $H_1$ becomes a $\beta$-contraction only in terms of the $\|\,\cdot\,\|_{\infty}$-norm. However, without switching from  $\|\,\cdot\,\|_{\nu}$-norm to $\|\,\cdot\,\|_{\infty}$-norm, a similar error analysis in terms of $\|\,\cdot\,\|_{\nu}$-norm can be formed by replacing Assumption~\ref{as1}-(c) with a concentrability assumption (see \cite{MuSz08,AgJiKa19}). To that end, let us define the state-action visitation probability of any policy $\pi$ as
$$
d^{\pi}(x,da) \coloneqq (1-\beta) \sum_{t=0}^{\infty} \rP^{\pi}(x(t)= x, a(t) \in da).
$$  
The concentrability assumption states that the state-action visitation probability $d^{\pi}$ is absolutely continuous with respect to  $\nu(x) \otimes m_{\sA}(da)$ for any $\pi \in \Pi$, and the corresponding densities are uniformly bounded, i.e.,
$$
\sup_{\pi \in \Pi} \left\|\frac{d^{\pi}}{\nu\otimes m_{\sA}}\right\|_{\infty} \leq C,
$$
for some $C$. Under this assumption, the final part of Theorem~\ref{Thm-H1} can be handled with the $\|\,\cdot\,\|_{\nu}$-norm instead of the $\|\,\cdot\,\|_{\infty}$-norm using performance difference lemma \cite[Theorem 15.4]{AgJiKa19}. However, it is not possible to establish the overall error analysis of the learning algorithm using the $\|\,\cdot\,\|_{\nu}$-norm on $Q$-functions under the same set of assumptions on the system components without strengthening Assumption~\ref{as1}-(d). 
\end{remark}

We now give the description of the random operator $\hat{H}_2$, and then, do the error analysis. In this algorithm, the goal is to replace the operator $H_2$, which gives the next state-measure, with $\hat{H}_2$. We achieve this by simulating the transition probability $p(\cdot|x,a,\mu)$ for certain state-measure $\mu $ and policy $\pi$. This is possible since $|\sX|$ is finite.  

\begin{algorithm}[H]
\caption{Algorithm $\hat{H}_2$}
\label{H2}
\begin{algorithmic}
\STATE{Inputs $\left(\mu,Q\right)$, Data size $M$, Number of iterations $|\sX|$}
\FOR{$x \in \sX$}
\STATE{generate i.i.d. samples $\{y_t^x\}_{t=1}^M$ using
$$
y_t^x \sim p(\cdot|x,f_{Q}(x),\mu)
$$
and define 
$$
p_M(\cdot|x,f_{Q}(x),\mu) = \frac{1}{M} \sum_{t=1}^M \delta_{y_t^x}(\cdot).
$$
}
\ENDFOR
\RETURN{$\sum_{x \in \sX} p_M(\cdot|x,f_{Q}(x),\mu) \, \mu(x)$}
\end{algorithmic}
\end{algorithm}

This is the error analysis of the random operator $\hat{H}_2$.

\begin{theorem}\label{Thm-H2}
For any $(\varepsilon,\delta) \in (0,1)^2$, with probability at least $1-\delta$
$$
\left\|\hat{H}_2[M](\mu,Q) - H_2(\mu,Q) \right\|_1 \leq \varepsilon 
$$
if $M \geq m_2(\epsilon,\delta)$, where 
$$
m_2(\epsilon,\delta) \coloneqq \frac{|\sX|^2}{\varepsilon^2} \, \ln\left(\frac{2 \, |\sX|^2}{\delta}\right). 
$$
\end{theorem}

\proof
By Hoeffding Inequality \cite[Theorem 2.1]{HaRa19}, for any $x,y \in \sX$, we have
\begin{align}
\rP \left( \left|p_M(y|x,f_{Q}(x),\mu) - p(y|x,f_{Q}(x),\mu) \right| > \frac{\varepsilon}{|\sX|} \right) \leq 2 e^{\frac{-M \varepsilon^2}{|\sX|^2}}. \nonumber 
\end{align}
Hence, we have
\begin{align}
&\rP\left( \left\|\hat{H}_2[M](\mu,Q) - H_2(\mu,Q) \right\|_1 \leq \varepsilon \right) \nonumber \\
&\geq \rP \left( \sum_{x,y \in \sX} \, \left|p_M(y|x,f_{Q}(x),\mu) - p(y|x,f_{Q}(x),\mu) \right| \, \mu(x) \leq \varepsilon \right) \nonumber \\
&\geq 1- \rP \left( \exists  x,y \in \sX \, \text{ s.t.} \, \left|p_M(y|x,f_{Q}(x),\mu) - p(y|x,f_{Q}(x),\mu) \right| > \frac{\varepsilon}{|\sX|}, \right) \nonumber \\
&\geq 1-2 \, |\sX|^2 \,e^{\frac{-M \varepsilon^2}{|\sX|^2}}. \nonumber  
\end{align}
The result follows by picking $\delta = 2 \, |\sX|^2 \,e^{\frac{-M \varepsilon^2}{|\sX|^2}}$.
\endproof

The overall description of the learning algorithm is given below. In this algorithm, to achieve an approximate mean-field equilibrium policy, we successively apply the random operator $\hat{H}$ which replaces the MFE operator $H$.

\begin{algorithm}[H]
\caption{Learning Algorithm}
\label{Qit}
\begin{algorithmic}
\STATE{Input $\mu_0$, Number of iterations $K$, Parameters of $\hat{H_1}$ and $\hat{H}_2$ $\left(\{[N_k,L_k]\}_{k=0}^{K-1},\{M_k\}_{k=0}^{K-1}\right)$}
\STATE{Start with $\mu_0$}
\FOR{$k=0,\ldots,K-1$}
\STATE{
$\mu_{k+1} = \hat{H}\left([N_k,L_k],M_k\right)(\mu_k) \coloneqq \hat{H}_2[M_k]\left(\mu_k,\hat{H}_1[N_k,L_k](\mu_k)\right)$
}
\ENDFOR
\RETURN{$\mu_K$}
\end{algorithmic}
\end{algorithm}

The current state-measure $\mu_k$ is the input for each iteration $k=0,\ldots,K-1$. In addition, for the random operator $\hat{H_1}$, we choose integers $N_k$ and $L_k$ as the data size and the number of iterations, respectively, and for the random operator  $\hat{H}_2$, we choose an integer $M_k$ as the data size. We first compute an approximate $Q$-function for $\mu_k$ by applying $\hat{H}_1[N_k,L_k](\mu_k)$, and then we compute an approximate next state-measure by applying $\hat{H}_2[M_k](\mu_k,\hat{H}_1[N_k,L_k](\mu_k))$. Since an approximate $Q$-function is used instead of the exact $Q$-function in the second stage of the iteration, there will be an error due to $\hat{H}_1$ in addition to the error resulting from $\hat{H_2}$.

The error analyses of the algorithms  $\hat{H_1}$ and $\hat{H}_2$ have been completed in Theorem~\ref{Thm-H1} and Theorem~\ref{Thm-H2}, respectively. The error analysis for the learning algorithm for the random operator $\hat{H}$, which is a combination of $\hat{H_1}$ and $\hat{H}_2$, is given below. We state the key result of this section as a corollary after the proof of the following theorem.

\begin{theorem}\label{main-result}
Fix any $(\varepsilon,\delta) \in (0,1)^2$. Define 
$$
\varepsilon_1 \coloneqq \frac{\rho \, (1-K_H)^2 \, \varepsilon^2}{64 (K_1)^2}, \,\,
\varepsilon_2 \coloneqq \frac{(1-K_H) \, \varepsilon}{4}.
$$
Let $K,L$ be such that 
\begin{align}
\frac{(K_H)^K}{1-K_H} &\leq \frac{\varepsilon}{2}, \,\, \frac{\beta^L}{1-\beta} Q_{\bf m} \leq \frac{\varepsilon_1}{2}. \nonumber 
\end{align}
Then, pick $N,M$ such that
\begin{align}
N &\geq m_1\left( \varepsilon_1,\frac{\delta}{2K},L \right), \,\, M \geq m_2\left( \varepsilon_2,\frac{\delta}{2K} \right).
\end{align}
Let $\mu_K$ be the output of the learning algorithm with parameters $$\left(\mu_0, K, \{[N,L]\}_{k=0}^K, \{M\}_{k=0}^{K-1} \right).$$ Then, with probability at least $1-\delta$
$$
\|\mu_K - \mu_*\|_1 \leq \frac{2 K_1 \sqrt{\Delta }}{\sqrt{\rho} (1-K_H)} + \varepsilon,
$$
where $\mu_*$ is the state-measure in mean-field equilibrium given by the MFE operator $H$. 
\end{theorem}

\proof
Note that for any $\mu \in \P(\sX)$, $Q \in {\cal F}$, $\hQ \in {\cal C}$, we have
\begin{align}
\|H_2(\mu,Q) - H_2(\mu,\hQ)\|_1 
&= \sum_{y \in \sX} \left| \sum_{x \in \sX} p(y|x,f_{Q}(x),\mu) \, \mu(x) - \sum_{x \in \sX} p(y|x,f_{\hQ}(x),\mu) \, \mu(x) \right| \nonumber \\
&\leq \sum_{x \in \sX} \|p(\cdot|x,f_{Q}(x),\mu)-p(\cdot|x,f_{\hQ}(x),\mu)\|_1 \, \mu(x) \nonumber \\
&\leq \sum_{x \in \sX} K_1 \, \|f_{Q}(x)-f_{\hQ}(x)\| \, \mu(x). \label{fbound}
\end{align}
For all $x \in \sX$, note that the mapping $f_{Q}(x)$ is the minimizer of $Q(x,\,\cdot\,)$ and the mapping $f_{\hQ}(x)$ is the unique minimizer of $\hQ(x,\,\cdot\,)$ by strong convexity. Let us set $a = f_{\hQ}(x)$ and $r = f_{Q}(x) - f_{\hQ}(x).$ As $a$ is the unique minimizer of a strongly convex function $\hQ(x,\,\cdot\,)$, by first-order optimality condition, we have
$$
\nabla \, \hQ\left(x,a\right) \cdot r \geq 0. 
$$
Hence, by strong convexity 
\begin{align}
\hQ(x,a+r) - \hQ(x,a) &\geq \nabla \hQ(x,a) \cdot r + \frac{\rho}{2} \|r\|^2 \nonumber \\
&\geq \frac{\rho}{2} \|r\|^2
\end{align}
For all $x \in \sX$, this leads to
\begin{align}
\|f_{Q}(x) - f_{\hQ}(x)\|^2 &\leq  \frac{2}{\rho} \, \left(\hQ(x,f_{Q}(x)) - \hQ(x,f_{\hQ}(x))\right) \nonumber \\
&= \frac{2}{\rho} \, \left(\hQ(x,f_{Q}(x)) - Q(x,f_{Q}(x)) +Q(x,f_{Q}(x)) - \hQ(x,f_{\hQ}(x))\right) \nonumber \\
&= \frac{2}{\rho} \, \left(\hQ(x,f_{Q}(x)) - Q(x,f_{Q}(x)) + \min_{a \in \sA}Q(x,a) - \min_{a \in \sA} \hQ(x,a)\right) \nonumber \\
&\leq \frac{4}{\rho} \, \|Q-\hQ\|_{\infty}. \label{perturbation2}
\end{align}
Hence, combining (\ref{fbound}) and (\ref{perturbation2}) yields
\begin{align}
\|H_2(\mu,Q) - H_2(\mu,\hQ)\|_1 \leq \frac{2 K_1}{\sqrt{\rho}} \, \sqrt{\|Q-\hQ\|_{\infty}}. \label{mainbound}
\end{align}
Using (\ref{mainbound}) and the fact that $H_1(\mu_k) \in {\cal C}$ and $\hat{H}_1[N,L](\mu_k) \in {\cal F}$, for any $k=0,\ldots,K-1$, we have
\begin{align}
\|H(\mu_k) - \hat{H}([N,L],M)(\mu_k)\|_1 
&\leq \|H_2(\mu_k,H_1(\mu_k)) - H_2(\mu_k,\hat{H}_1[N,L](\mu_k))\|_1 \nonumber \\
&+ \|H_2(\mu_k,\hat{H}_1[N,L](\mu_k)) - \hat{H}_2[M](\mu_k,\hat{H}_1[N,L](\mu_k))\|_1  \nonumber \\
&\leq \frac{2 K_1}{\sqrt{\rho}} \, \sqrt{\|H_1(\mu_k) - \hat{H}_1[N,L](\mu_k)\|_{\infty}} \nonumber \\
&+ \|H_2(\mu_k,\hat{H}_1[N,L](\mu_k)) - \hat{H}_2[M](\mu_k,\hat{H}_1[N,L](\mu_k))\|_1. \nonumber 
\end{align}
The last term is upper bounded by 
$$\frac{2 K_1 \sqrt{\varepsilon_1 + \Delta}}{\sqrt{\rho}} + \varepsilon_2$$
with probability at least $1-\frac{\delta}{K}$ by Theorem~\ref{Thm-H1} and Theorem~\ref{Thm-H2}. Therefore, with probability at least $1-\delta$ 
\begin{align}
\|\mu_K - \mu_*\|_1 &\leq \sum_{k=0}^{K-1} K_H^{K-(k+1)} \, \|\hat{H}([N,L],M)(\mu_k) - H(\mu_k)\|_1 + \|H^K(\mu_0) - \mu_*\|_1 \nonumber \\
&\leq \sum_{k=0}^{K-1} K_H^{K-(k+1)} \left(\frac{2 K_1 \sqrt{ \varepsilon_1 + \Delta}}{\sqrt{\rho}} + \varepsilon_2\right) + \frac{(K_H)^K}{1-K_H} \nonumber \\
&\leq \frac{2 K_1 \sqrt{\Delta }}{\sqrt{\rho} (1-K_H)} + \varepsilon. \nonumber 
\end{align}
This completes the proof.
\endproof

Now, we state the main result of this section. It basically states that, by using the learning algorithm, one can learn an approximate mean-field equilibrium policy. By Theorem~\ref{old-main-cor}, this gives an approximate Nash-equilibrium for the finite-agent game.

\begin{corollary}\label{main-cor}
Fix any $(\varepsilon,\delta) \in (0,1)^2$. Suppose that $K,L,N,M$ satisfy the conditions in Theorem~\ref{main-result}. Let $\mu_K$ be the output of the learning algorithm with parameters $$\left(\mu_0, K, \{[N,L]\}_{k=0}^K, \{M\}_{k=0}^{K-1}\right).$$ Define $\pi_K(x) \coloneqq \argmin_{a \in \sA} Q_K(x,a)$, where $Q_K \coloneqq \hat{H}_1([N,L])(\mu_K)$. Then, with probability at least $1-\delta(1+\frac{1}{2K})$, the policy $\pi_K$ is a $\kappa(\varepsilon,\Delta)$-mean-field equilibrium policy, where
\begin{align}
&\kappa(\varepsilon,\Delta) = \sqrt{\frac{4}{\rho}\left(\frac{\rho^2 \, (1-K_H)^2 \, \varepsilon^2}{64 (K_1)^2} + \Delta + K_{H_1} \left(\frac{2 K_1 \sqrt{\Delta }}{\sqrt{\rho} (1-K_H)} + \varepsilon\right)\right)}. \nonumber
\end{align}
Therefore, by Theorem~\ref{old-main-cor}, an $N$-tuple of policies ${\bf \pi}^{(N)} = \{\pi_K,\pi_K,\ldots,\pi_K\}$ is an $\tau \kappa(\varepsilon,\Delta) + \sigma$-Nash equilibrium for the game with $N \geq N(\sigma)$ agents. 
\end{corollary}

\proof
By Theorem~\ref{Thm-H1} and Theorem~\ref{main-result}, with probability at least $1-\delta(1+\frac{1}{2K})$, we have 
\begin{align}
\|Q_K - H_1(\mu_*)\|_{\infty} &\leq \|Q_K - H_1(\mu_K)\|_{\infty} + \|H_1(\mu_K) - H_1(\mu_*)\|_{\infty} \nonumber \\
&\leq \varepsilon_1 + \Delta + K_{H_1} \|\mu_K - \mu_*\|_1 \nonumber \\
&\leq \varepsilon_1 + \Delta + K_{H_1} \left(\frac{2 K_1 \sqrt{\Delta }}{\sqrt{\rho} (1-K_H)} + \varepsilon\right) \nonumber \\
&= \frac{\rho \, (1-K_H)^2 \, \varepsilon^2}{64  (K_1)^2} + \Delta + K_{H_1} \left(\frac{2 K_1 \sqrt{\Delta }}{\sqrt{\rho} (1-K_H)} + \varepsilon\right). \nonumber
\end{align}
Let $\pi_K(x) \coloneqq \argmin_{a \in \sA} Q_K(x,a)$. Using the same analysis that leads to (\ref{perturbation2}), we can obtain the following bound since $Q_K \in {\cal F}$ and $H_1(\mu_*) \in {\cal C}$:
\begin{align}
\sup_{x \in \sX} \|\pi_K(x)-\pi_*(x)\|^2 &\leq \frac{4}{\rho} \, \|Q_K-H_1(\mu_*)\|_{\infty}. \nonumber 
\end{align}
Hence, with probability at least $1-\delta(1+\frac{1}{2K})$, the policy $\pi_K$ is a $\kappa(\varepsilon,\Delta)$-mean-field equilibrium, where
\begin{align}
&\kappa(\varepsilon,\Delta) = \sqrt{\frac{4}{\rho}\left(\frac{\rho \, (1-K_H)^2 \, \varepsilon^2}{64 (K_1)^2} + \Delta + K_{H_1} \left(\frac{2 K_1 \sqrt{\Delta }}{\sqrt{\rho} (1-K_H)} + \varepsilon\right)\right)}. \nonumber
\end{align}
This completes the proof.
\endproof

\begin{remark}
Note that, in Corollary~\ref{main-cor}, there is a constant $\Delta$, which depends on the representation error $E({\cal F})$. In general, 
$E({\cal F})$ is very small since any $Q$ function in ${\cal C}$ can be approximated quite well by functions in ${\cal F}$. Therefore, $\Delta$ is negligible. In this case, we have the following error bound:
\begin{align}
&\kappa(\varepsilon,0) = \sqrt{\frac{4}{\rho}\left(\frac{\rho \, (1-K_H)^2 \, \varepsilon^2}{64  (K_1)^2} + K_{H_1} \varepsilon\right)}. \nonumber
\end{align}
which goes to zero as $\varepsilon \rightarrow 0$.  
\end{remark}

\section{Mean-field Equilibrium Operator for Average-cost}\label{av-known-model}

In this section, we introduce the corresponding MFE operator for average-cost mean-field games. For the purpose of keeping the notation similar to the discounted-cost case while making the distinctions more apparent, we use the `$\av$' superscript to denote the related quantities in the average-cost setting. For instance, to denote the average-cost of any policy $\pi$ with initial state $x$ under state-measure $\mu$, we use $J_{\mu}^{\av}(\pi,x)$ instead of $J_{\mu}(\pi,x)$. Now, let us state the extra conditions imposed for the average-cost in addition to Assumption~\ref{as1}.

\begin{assumption}\label{av-as2}
\begin{itemize}
\item[ ]
\item[(a)] There exists a sub-probability measure $\lambda$ on $\sX$ such that 
$$
p(\,\cdot\,|x,a,\mu) \geq \lambda(\,\cdot\,)
$$
for all $x,a,\mu$. 
\item[(b)] Let $\beta^{\av} \coloneqq 1-\lambda(\sX)$ and $\displaystyle Q_{\Lip}^{\av} \coloneqq \frac{L_1}{1-K_1/2} > 0$. We assume that 
$$\frac{3K_1}{2} \left(1+\frac{K_F}{\rho}\right)+\frac{K_1K_F Q_{\Lip}^{\av}}{\rho(1-\beta^{\av})} < 1.$$ 
\end{itemize}
\end{assumption}

Note that Assumption~\ref{av-as2}-(a) is the so-called `minorization' condition. Minorization condition was used in the literature for studying the geometric ergodicity of Markov chains (see \cite[Section 3.3]{Her89}). The minorization condition is true when the transition probability satisfies conditions R0, R1(a) and R1(b) in \cite{HeMoRo91} (see also \cite[Remark 3.3]{HeMoRo91} and references therein for further conditions). In general, this condition is restrictive for unbounded state spaces, but it is quite general for compact or finite state spaces. Indeed, the minorization condition was used to study average-cost mean-field games with a compact state space in \cite[Assumption A.3]{Wie19}. Note that Assumption~\ref{av-as2}-(b) is used to ensure that MFE operator is contraction, which is crucial to establish the error analysis of the learning algorithm, and so, cannot be relaxed.

Recall that for the average-cost, given any state-measure $\mu$, the value function $J_{\mu}^{\av}$ of policy $\pi$ with initial state $x$ is given by
$$
J_{\mu}^{\av}(\pi,x) \coloneqq \limsup_{T\rightarrow\infty} \frac{1}{T} E^{\pi}\biggl[ \sum_{t=0}^{T-1} c(x(t),a(t),\mu) \, \bigg| \, x(0) = x \biggr]. 
$$
Then, the optimal value function is defined as 
$$J_{\mu}^{\av,*}(x) \coloneqq \inf_{\pi \in \Pi} J_{\mu}^{\av}(\pi,x).$$ Under Assumption~\ref{as1} and Assumption~\ref{av-as2}, it can be proved that 
$$J_{\mu}^{\av,*}(x) = J_{\mu}^{\av,*}(y) \eqqcolon J_{\mu}^{\av,*}$$ 
for all $x,y \in \sX$, for some constant $J_{\mu}^{\av,*}$; that is, the optimal value function does not depend on the initial state. Furthermore, let $h_{\mu}^*(x)$ be the unique fixed point of the $\beta$-contraction operator $T_{\mu}^{\av}$ with respect to $\|\cdot\|_{\infty}$-norm:
$$
h_{\mu}^*(x) = \min_{a \in \sA} \bigg[c(x,a,\mu) + \sum_{y \in \sX} h_{\mu}^*(y)  q(y|x,a,\mu) \bigg] \eqqcolon T_{\mu}^{\av}h_{\mu}^*(x),
$$
where 
$$q(\,\cdot\,|x,a,\mu) \coloneqq p(\,\cdot\,|x,a,\mu)-\lambda(\,\cdot\,).$$ 
Then, the pair $\left(h_{\mu}^*,\sum_{y\in\sX} h_{\mu}^*(y) \, \lambda(y)\right)$ satisfies the average-cost optimality equation (ACOE):
$$
h_{\mu}^*(x) + \sum_{y\in\sX} h_{\mu}^*(y) \, \lambda(y) = \min_{a \in \sA} \bigg[c(x,a,\mu) +\sum_{y \in \sX} h_{\mu}^*(y) \, p(y|x,a,\mu) \bigg].
$$
Therefore, $J_{\mu}^{\av,*} = \sum_{y\in\sX} h_{\mu}^*(y) \, \lambda(y)$. Additionally, if $f^*: \sX \rightarrow \sA$ attains the minimum in the ACOE, that is,
\begin{align}
&\min_{a \in \sA} \bigg[c(x,a,\mu) + \sum_{y \in \sX} h_{\mu}^*(y) \, p(y|x,a,\mu) \bigg] = c(x,f^*(x),\mu) + \sum_{y \in \sX} h_{\mu}^*(y) \, p(y|x,f^*(x),\mu) \label{av-optim}
\end{align}
for all $x \in \sX$, then the policy $\pi^*(a|x) = \delta_{f^*(x)}(a) \in \Pi_d$ is optimal for any initial distribution. We refer the reader to \cite[Chapter 3]{Her89} for basics of average-cost Markov decision processes, where these classical results can be found.

We can also obtain a similar characterization by using a $Q$-function instead of $h_{\mu}^*$. Indeed, we define the $Q$-function as
$$
Q_{\mu}^{\av,*}(x,a) = c(x,a,\mu) + \sum_{y \in \sX} h_{\mu}^*(y) \, q(y|x,a,\mu). 
$$
Note that $Q_{\mu,\min}^{\av,*}(x) \coloneqq \min_{a \in \sA} Q_{\mu}^{\av,*}(x,a) = h_{\mu}^*(x)$ for all $x \in \sX$, and so, we have
$$
Q_{\mu}^{\av,*}(x,a) = c(x,a,\mu) + \sum_{y \in \sX} Q_{\mu,\min}^{\av,*}(y) \, q(y|x,a,\mu)  \eqqcolon H_{\mu}^{\av}Q_{\mu}^{\av,*}(x,a), 
$$
where $H_{\mu}^{\av}$ is the corresponding operator on $Q$-functions. Hence, the policy $\pi^*(a|x) = \delta_{f^*(x)}(a) \in \Pi_d$ is optimal for $\mu$ and for any initial distribution, if $Q_{\mu}^{\av,*}(x,f^*(x)) = Q_{\mu,\min}^{\av,*}(x)$ for all $x \in \sX$. One can prove that $H_{\mu}^{\av}$ is a $\|\cdot\|_{\infty}$-contraction with modulus $\beta^{\av}$, and so, the unique fixed point of $H_{\mu}^{\av}$ is $Q_{\mu}^{\av,*}$. Indeed, let $Q$ and $\hQ$ be two different $Q$-functions. Then, we have
\begin{align}
\|H_{\mu}^{\av} Q - H_{\mu}^{\av} \hQ\|_{\infty} &\leq \sup_{(x,a) \in \sX\times\sA} \sum_{y \in \sX} |Q_{\min}(y)-\hQ_{\min}(y)| \, q(y|x,a,\mu) \nonumber \\
&\leq \|Q_{\min}-\hQ_{\min}\|_{\infty} \, \sup_{(x,a) \in \sX\times\sA} q(\sX|x,a,\mu) \nonumber \\
&= \beta^{\av} \, \|Q_{\min}-\hQ_{\min}\|_{\infty}. \nonumber 
\end{align}
Hence, using the Banach fixed point theorem, we can develop a $Q$-iteration algorithm to compute $Q_{\mu}^{\av,*}$, the minimum of which gives the optimal policy. The benefit of this algorithm, as in the discounted case, is that it can be adapted to a model-free setting via  $Q$-learning.

Using (\ref{tv-bound}), we now prove the following result.

\begin{lemma}\label{av-lip-value}
For any $\mu$, $Q^{\av,*}_{\mu,\min}$ is $Q_{\Lip}^{\av}$-Lipschitz continuous; that is, 
$$
|Q^{\av,*}_{\mu,\min}(x)-Q^{\av,*}_{\mu,\min}(y)| \leq Q_{\Lip}^{\av} \, d_{\sX}(x,y).
$$
\end{lemma}

\begin{proof}
The proof is exactly the same with the proof of Lemma~\ref{lip-value}. The only difference is the absence of the discount factor $\beta$.
\end{proof}

Before we define MFE operator, let us describe the set of possible $Q$-functions. This set ${\cal C}^{\av}$ is the set of all $Q$-functions $Q:\sX \times \sA \rightarrow \R$ such that $\displaystyle \|Q\|_{\infty} \leq Q_{\bf m}^{\av} \coloneqq c_{\bf m}/(1-\beta^{\av})$,   $Q(x,\cdot)$ is $Q_{\Lip}^{\av}$-Lipschitz and $\rho$-strongly convex for all $x$, and the gradient $\nabla Q(x,a)$ of $Q$ with respect to $a$ satisfies the bound 
$$
\sup_{a \in \sA} \|\nabla Q(x,a) - \nabla Q(\hat{x},a)\| \leq K_F, \, \forall x, \hat{x}.
$$ 
Now, we can define the MFE operator $H^{\av}$. The operator $H^{\av}$ is very similar to $H$; that is, it is a composition of two operators, where the first operator $H_1^{\av}: \P(\sX) \rightarrow {\cal C}^{\av}$ is defined as $H_1^{\av}(\mu) = Q_{\mu}^{\av,*}$ (the unique fixed point of the operator $H_{\mu}^{\av}$). The second operator $H_2^{\av}: \P(\sX) \times {\cal C}^{\av} \rightarrow \P(\sX)$ is defined as
$$
H_2^{\av}(\mu,Q)(\cdot) \coloneqq \sum_{x \in \sX} p(\cdot|x,f_{Q}(x),\mu) \, \mu(x),
$$
where $f_{Q}(x) \coloneqq \argmin_{a\in\sA} Q(x,a)$ is the unique minimizer by $\rho$-strong convexity of $Q$, for any $Q \in {\cal C}^{\av}$. Note that we indeed have $H_2^{\av} = H_2$, where $H_2$ is the operator that computes the new state-measure in discounted-cost. However, to be consistent with the notation used in this section, we keep $H_2^{\av}$ as it is. Using these operators, let us define the MFE operator as a composition:
$$H^{\av}: \P(\sX) \ni \mu \mapsto H_2^{\av}\left(\mu,H_1^{\av}(\mu)\right) \in \P(\sX).$$ 
Our goal is to establish that $H^{\av}$ is contraction. Using (\ref{tv-bound}) and Lemma~\ref{av-lip-value}, we can first prove that $H_1$ is Lipschitz continuous. 

\begin{lemma}\label{av-n-lemma1}
The mapping $H_1^{\av}$ is $K_{H_1^{\av}}$-Lipschitz, where
$\displaystyle K_{H_1^{\av}} \coloneqq \frac{Q_{\Lip}^{\av}}{1-\beta^{\av}}.$
\end{lemma}

\begin{proof}
The proof can be done as in the proof of Lemma~\ref{n-lemma1} by making appropriate modifications. Note that since $H_1^{\av}(\mu) \coloneqq Q_{\mu}^{\av,*}$ is the fixed point of the contraction operator $H_{\mu}^{\av}$, where $H_{\mu}^{\av}$ is given by 
\begin{align}
H_{\mu}^{\av} Q(x,a) = c(x,a,\mu) + \sum_{y \in \sX} Q_{\min}(y) \, q(y|x,a,\mu), \nonumber
\end{align}
by Assumption~\ref{as1}-(a),(b),(d), $H_{\mu}^{\av}$ maps any continuous $Q:\sX \times \sA \rightarrow \R$ into ${\cal C}^{\av}$. Hence, the fixed point $Q_{\mu}^{\av,*}$ of $H_{\mu}^{\av}$ must be in ${\cal C}^{\av}$ (see the proof of Lemma~\ref{n-lemma1}). Therefore, $H_{\mu}^{\av}$ is well-defined. 

%

Let us now prove that $H_1^{\av}$ is $K_{H_1^{\av}}$-Lipschitz. For any $\mu,\hmu \in \P(\sX)$, we have 
\begin{align}
\|H_1^{\av}(\mu) &- H_1^{\av}(\hmu)\|_{\infty} \nonumber \\
&= \sup_{x,a} \bigg| c(x,a,\mu) + \sum_{y} \hspace{-3pt} Q_{\mu,\min}^{\av,*}(y) q(y|x,a,\mu) - c(x,a,\hmu) - \sum_{y} \hspace{-3pt} Q_{\hmu,\min}^{\av,*}(y)  q(y|x,a,\hmu) \bigg| \nonumber \\
&\leq L_1 \, \|\mu-\hmu\|_1 \nonumber \\
&+ \left| \sum_{y} Q_{\mu,\min}^{\av,*}(y) q(y|x,a,\mu) - \sum_{y} Q_{\mu,\min}^{\av,*}(y) q(y|x,a,\hmu) \right| \nonumber \\
&+ \left| \sum_{y} Q_{\mu,\min}^{\av,*}(y) q(y|x,a,\hmu) - \sum_{y} Q_{\hmu,\min}^{\av,*}(y) q(y|x,a,\hmu) \right| \nonumber \\
&\leq L_1 \, \|\mu-\hmu\|_1 + Q_{\Lip}^{\av} \, K_1/2 \, \|\mu-\hmu\|_1 + \beta^{\av} \, \|Q_{\mu}^{\av,*}-Q_{\hmu}^{\av,*}\|_{\infty}, \nonumber
\end{align}
where the last inequality follows from (\ref{tv-bound}), Lemma~\ref{av-lip-value}, and the fact $q(\sX|x,a,\mu) = \beta^{\av}$ for all $x,a,\mu$. 
\end{proof} 

\noindent Now, using Lemma~\ref{av-n-lemma1}, we can prove that $H^{\av}$ is contraction.

\begin{proposition}\label{av-MFE-con}
The mapping $H^{\av}$ is a contraction with contraction  modulus $K_{H^{\av}}$, where
$$
K_{H^{\av}} \coloneqq \frac{3K_1}{2} \left(1+\frac{K_F}{\rho}\right)+\frac{K_1K_FK_{H_1^{\av}}}{\rho}. 
$$
\end{proposition}

\begin{proof}
The proof is exactly the same with the proof of Proposition~\ref{MFE-con}. The only difference is the following: we should replace $K_{H_1}$ in (\ref{perturbation}) with $K_{H_1^{\av}}$. 
\end{proof}

Now, we know that $H^{\av}$ is a contraction mapping under Assumption~\ref{as1} and Assumption~\ref{av-as2}. Therefore, by Banach fixed point theorem, $H^{\av}$ has a unique fixed point $\mu_*^{\av}$. Let 
$$Q_{\mu_*^{\av}}^{\av,*} = H_1^{\av}(\mu_*^{\av}) \quad 
\text{and} \quad 
\pi_*^{\av}(\,\cdot\,|x) = \delta_{f_{Q^{\av,*}_{\mu_*^{\av}}}(x)}(\,\cdot\,).$$ 
Then, the pair $(\pi_*^{\av},\mu_*^{\av})$ is a mean-field equilibrium since $(\mu_*^{\av},Q_{\mu_*^{\av}}^{\av,*})$ satisfy the following equations
\begin{align}
\mu_{*}^{\av}(\cdot) &= \sum_{x \in \sX} p(\cdot|x,a,\mu_*^{\av}) \, \pi_*^{\av}(a|x) \, \mu_{*}^{\av}(x), \label{av-opt2} \\
Q_{\mu_*^{\av}}^*(x,a) &= c(x,a,\mu_*^{\av}) + \sum_{y \in \sX} Q_{\mu_*^{\av},\min}^{\av,*}(y) \, q(y|x,a,\mu_*^{\av}). \label{av-opt1} 
\end{align}
Here, (\ref{av-opt1}) implies that $\pi_*^{\av} \in \Psi(\mu_*^{\av})$ since 
$$f_{Q^{\av,*}_{\mu_*^{\av}}}(x) \coloneqq \argmin Q^{\av,*} _{\mu_*^{\av}}(x,a)$$
for every $x \in \sX$, and (\ref{av-opt2}) implies $\mu_*^{\av} \in \Lambda(\pi_*^{\av})$. Hence, $(\pi_*^{\av},\mu_*^{\av})$ is a mean-field equilibrium. Therefore, since $H^{\av}$ is a contraction, we can compute this mean-field equilibrium by applying $H^{\av}$ recursively starting from arbitrary state-measure.

Note that if the transition probability $p$, the one-stage cost function $c$, and the minorizing sub-probability measure $\lambda$ are not available to the decision maker, we need to replace $H^{\av}$ with a random operator and establish a learning algorithm via this random operator. To prove the convergence of this learning algorithm, the contraction property of $H^{\av}$ is crucial, similar to the discounted-case.

\section{Finite-Agent Game for Average-cost}\label{av-sec2}

The finite-agent game model for average-cost is exactly the same with the model introduced in Section~\ref{sec2} for the discounted-cost case. The only difference is the cost function. Here, under an $N$-tuple of policies ${\boldsymbol \pi}^{(N)} \coloneqq (\pi^1,\ldots,\pi^N)$, for Agent~$i$, the average-cost is given by
\begin{align}
J_i^{\av,(N)}({\boldsymbol \pi}^{(N)}) &= \limsup_{T\rightarrow\infty} \frac{1}{T} E^{{\boldsymbol \pi}^{(N)}}\biggl[\sum_{t=0}^{T-1} c(x_{i}^N(t),a_{i}^N(t),e^{(N)}_t)\biggr]. \nonumber 
\end{align}
Using this, we define Nash equilibrium and $\delta$-Nash equilibrium similarly. 

\begin{definition}
An $N$-tuple of policies ${\boldsymbol \pi}^{(N*)}= (\pi^{1*},\ldots,\pi^{N*})$ constitutes a \emph{Nash equilibrium} if
\begin{align}
J_i^{\av,(N)}({\boldsymbol \pi}^{(N*)}) = \inf_{\pi^i \in \Pi_i} J_i^{\av,(N)}({\boldsymbol \pi}^{(N*)}_{-i},\pi^i) \nonumber
\end{align}
for each $i=1,\ldots,N$. An $N$-tuple of policies ${\boldsymbol \pi}^{(N*)}= (\pi^{1*},\ldots,\pi^{N*})$ constitutes an \emph{$\delta$-Nash equilibrium} if
\begin{align}
J_i^{\av,(N)}({\boldsymbol \pi}^{(N*)}) \leq \inf_{\pi^i \in \Pi_i} J_i^{\av,(N)}({\boldsymbol \pi}^{(N*)}_{-i},\pi^i) + \delta \nonumber
\end{align}
for each $i=1,\ldots,N$.
\end{definition}

As in the discounted-cost case, if the number of agents is large enough in the finite-agent setting, one can obtain $\delta$-Nash equilibrium by considering the infinite-population limit $N\rightarrow\infty$ of the game (i.e., mean-field game). Then, it is possible to prove that if each agent in the finite-agent $N$ game problem adopts the policy in mean-field equilibrium, the resulting $N$-tuple of policies will be an approximate Nash equilibrium for all sufficiently large $N$. This was indeed proved in \cite{Wie19,Sal20}.  
In the below theorem, we prove that if each agent in the finite-agent game model adopts the $\varepsilon$-mean-field equilibrium policy (instead of exact mean-field equilibrium policy), the resulting policy will still be an approximate Nash equilibrium for all sufficiently large $N$-agent game models.

Before we state the theorem, let us define the following constants:
\begin{align}
&C_1^{\av} \coloneqq \left(\frac{3 \, K_1}{2}  + \frac{K_1 \, K_{F}}{2\rho} \right), \, C_2^{\av} \coloneqq \frac{2 c_{\bf m} (K_1)^2}{(1-C_1^{\av}) \lambda(\sX)}, \,
C_3^{\av} \coloneqq \frac{2 c_{\bf m}}{\lambda(\sX)}.\nonumber
\end{align}
Note that by Assumption~\ref{av-as2}, the constant $C_1^{\av}$ is strictly less than $1$.

\begin{theorem}\label{av-old-main-cor}
Let $\pi_{\varepsilon}$ be an $\varepsilon$-mean-field equilibrium policy for the mean-field equilibrium $(\pi_*,\mu_*) \in \Pi_d \times \P(\sX)$ given by the unique fixed point of the MFE operator $H^{\av}$. Let $\eta_0 \in \Lambda(\pi_{\varepsilon})$. Then, for any $\delta>0$, there exists a positive integer $N(\delta)$ such that, for each $N\geq N(\delta)$, the $N$-tuple of policies ${\boldsymbol \pi}^{(N)} = \{\pi_{\varepsilon},\pi_{\varepsilon},\ldots,\pi_{\varepsilon}\}$ is a $(\delta+\tau^{\av}\varepsilon)$-Nash equilibrium for the game with $N$ agents, where $\tau^{\av} \coloneqq 2C_2^{\av}+C_3^{\av}$.
\end{theorem}

\proof
By an abuse of notation, we denote the deterministic mappings from $\sX$ to $\sA$ that induce policies $\pi_*$ and $\pi_{\varepsilon}$ as $\pi_*$ and $\pi_{\varepsilon}$ as well, respectively. 
As in the proof of Theorem~\ref{old-main-cor}, one can prove that 
\begin{align}
\|\pi_*(x)-\pi_*(y)\| &\leq \frac{K_F}{\rho} \, d_{\sX}(x,y) \quad
\text{and} \quad
\|\mu_{\varepsilon}-\mu_*\|_1 \leq \frac{K_1 \, \varepsilon}{1-C_1^{\av}}, \nonumber 
\end{align}
where $\mu_{\varepsilon} \in \Lambda(\pi_{\varepsilon})$ and $C_1^{\av} \coloneqq \left(\frac{3 \, K_1}{2}  + \frac{K_1 \, K_{F}}{2\rho} \right)$. Note that by Assumption~\ref{av-as2}, $C_1^{\av} < 1$.

For any policy $\pi \in \Pi_d$ and state measure $\mu$, Assumption~\ref{av-as2}-(a) (i.e., minorization condition) implies that there exists a unique invariant measure $\nu_{\pi,\mu} \in \P(\sX)$ of the transition probability $P_{\pi,\mu}(\,\cdot\,|x) \coloneqq \sum_x p(\,\cdot\,|x,\pi(x),\mu)$ such that for any initial state $x \in \sX$, we have 
$$
J^{\av}_{\mu}(\pi,x) = \sum_x c(x,\pi(x),\mu) \, \nu_{\pi,\mu}(x),
$$
where the last identity follows from ergodic theorem \cite[Lemma 3.3]{Her89}. Therefore, the value of any policy $\pi$ under $\mu$ does not depend on the initial state. Let us define
$$
J^{\av}_{\mu}(\pi,x) = J^{\av}_{\mu}(\pi,y) \eqqcolon J^{\av}_{\mu}(\pi), \quad \text{for all $x,y \in \sX$.}
$$
Now, fix any policy $\pi \in \Pi_d$. Then, we have 
\begin{align}
|J_{\mu_*}^{\av}(\pi)-J_{\mu_{\varepsilon}}^{\av}(\pi)| &= \left|\sum_x c(x,\pi(x),\mu_*) \, \nu_{\pi,\mu_*}(x) - \sum_x c(x,\pi(x),\mu_{\varepsilon}) \, \nu_{\pi,\mu_{\varepsilon}}(x)\right| \nonumber \\
&\leq c_{\bf m} \, \|\nu_{\pi,\mu_*}-\nu_{\pi,\mu_{\varepsilon}}\|_1.\nonumber
\end{align}
Hence, to bound $|J_{\mu_*}^{\av}(\pi)-J_{\mu_{\varepsilon}}^{\av}(\pi)|$, it is sufficient to bound $\|\nu_{\pi,\mu_*}-\nu_{\pi,\mu_{\varepsilon}}\|_1$. Note that invariant measures $\nu_{\pi,\mu_*}$ and  $\nu_{\pi,\mu_{\varepsilon}}$ satisfy the following fixed point equations
\begin{align}
\nu_{\pi,\mu_*}(\,\cdot\,) &= \sum_x p(\,\cdot\,|x,\pi(x),\mu_*) \, \nu_{\pi,\mu_*}(x) \nonumber \\
\nu_{\pi,\mu_{\varepsilon}}(\,\cdot\,) &= \sum_x p(\,\cdot\,|x,\pi(x),\mu_{\varepsilon}) \, \nu_{\pi,\mu_{\varepsilon}}(x). \nonumber
\end{align}
Hence, we have 
\begin{align}
\|\nu_{\pi,\mu_*}&-\nu_{\pi,\mu_{\varepsilon}}\|_1 = \sum_y \left|\sum_x p(y|x,\pi(x),\mu_*) \, \nu_{\pi,\mu_*}(x)-\sum_x p(y|x,\pi(x),\mu_{\varepsilon}) \, \nu_{\pi,\mu_{\varepsilon}}(x) \right| \nonumber \\
&\leq \sum_y \left|\sum_x p(y|x,\pi(x),\mu_*) \, \nu_{\pi,\mu_*}(x)-\sum_x p(y|x,\pi(x),\mu_*) \, \nu_{\pi,\mu_{\varepsilon}}(x) \right| \nonumber \\
&\phantom{xxxx}+\sum_y \left|\sum_x p(y|x,\pi(x),\mu_*) \, \nu_{\pi,\mu_{\varepsilon}}(x)-\sum_x p(y|x,\pi(x),\mu_{\varepsilon}) \, \nu_{\pi,\mu_{\varepsilon}}(x) \right| \nonumber \\
&\leq \sum_y \left|\sum_x p(y|x,\pi(x),\mu_*) \, \nu_{\pi,\mu_*}(x)-\sum_x p(y|x,\pi(x),\mu_*) \, \nu_{\pi,\mu_{\varepsilon}}(x) \right|\nonumber \\
&\phantom{xxxx}+\sum_x \|p(\,\cdot\,|x,\pi(x),\mu_*) -  p(\,\cdot\,|x,\pi(x),\mu_{\varepsilon})\|_1 \, \nu_{\pi,\mu_{\varepsilon}}(x) \nonumber \\
&\leq \sum_y \left|\sum_x p(y|x,\pi(x),\mu_*) \, \nu_{\pi,\mu_*}(x)-\sum_x p(y|x,\pi(x),\mu_*) \, \nu_{\pi,\mu_{\varepsilon}}(x) \right| + K_1 \, \|\mu_{\varepsilon}-\mu_*\|_1 \nonumber.
\end{align}
Note that \cite[Lemma 3.3]{Her89} implies that for all $x,z \in \sX$, we have
$$
\|p(\,\cdot\,|x,\pi(x),\mu_*) -  p(\,\cdot\,|z,\pi(z),\mu_*\|_1 \leq (2-\lambda(\sX)) \, d_{\sX}(x,z). 
$$
Then, by Lemma~\ref{KoRa08}, we have 
$$
\sum_y \left|\sum_x p(y|x,\pi(x),\mu_*) \, \nu_{\pi,\mu_*}(x)-\sum_x p(y|x,\pi(x),\mu_*) \, \nu_{\pi,\mu_{\varepsilon}}(x) \right| \leq \frac{2-\lambda(\sX)}{2} \, \|\nu_{\pi,\mu_*}-\nu_{\pi,\mu_{\varepsilon}}\|_1. 
$$
Since $1-\lambda(\sX)/2 < 1$, the last inequality gives the following
$$
\|\nu_{\pi,\mu_*}-\nu_{\pi,\mu_{\varepsilon}}\|_1 \leq \frac{2K_1}{\lambda(\sX)} \, \|\mu_{\varepsilon}-\mu_*\|_1
$$ 
Therefore, we obtain 
\begin{align}\label{av-nneqq1}
|J_{\mu_*}^{\av}(\pi)-J_{\mu_{\varepsilon}}^{\av}(\pi)| \leq \frac{2 c_{\bf m} (K_1)^2}{(1-C_1^{\av}) \lambda(\sX)} \varepsilon \eqqcolon C_2^{\av} \, \varepsilon. 
\end{align}
By using a similar analysis as above, we can also obtain the following
$$
\|\nu_{\pi_*,\mu_*}-\nu_{\pi_{\varepsilon},\mu_*}\|_1 \leq 
\frac{2}{\lambda(\sX)} \sum_x \|\pi_*(x)-\pi_{\varepsilon}(x)\| \, \nu_{\pi_{\varepsilon},\mu_*}(x). 
$$
Note that $\sup_x \|\pi_*(x)-\pi_{\varepsilon}(x)\| \leq \varepsilon$ as $\pi_{\varepsilon}$ is $\varepsilon$-mean-field equilibrium policy. Therefore, we obtain 
\begin{align}\label{av-nneqq2}
|J_{\mu_*}^{\av}(\pi_*)-J_{\mu_*}^{\av}(\pi_{\varepsilon})|_{\infty} \leq \frac{2 c_{\bf m}}{\lambda(\sX)}\varepsilon \eqqcolon C_3^{\av} \, \varepsilon. 
\end{align}

Note that we must prove that
\begin{align}
J_i^{\av,(N)}({\boldsymbol \pi}^{(N)}) &\leq \inf_{\pi^i \in \Pi_i} J_i^{\av,(N)}({\boldsymbol \pi}^{(N)}_{-i},\pi^i) + \tau^{\av} \, \varepsilon + \delta \label{av-old-eq13}
\end{align}
for each $i=1,\ldots,N$, when $N$ is sufficiently large. As the transition probabilities and the one-stage cost functions are the same for all agents, it is sufficient to prove (\ref{av-old-eq13}) for Agent~$1$ only. Given $\delta > 0$, for each $N\geq1$, let $\tpi^{(N)} \in \Pi_1$ be a deterministic policy such that
\begin{align*}
J_1^{\av,(N)} (\tpi^{(N)},\pi_{\varepsilon},\ldots,\pi_{\varepsilon}) < \inf_{\pi' \in \Pi_1} J_1^{\av,(N)} (\pi',\pi_{\varepsilon},\ldots,\pi_{\varepsilon}) + \frac{\delta}{3}.
\end{align*}
On the other hand, by \cite[Lemma 8]{Wie19} and \cite[Theorem 4.10]{SaBaRaSIAM} we get
\begin{align}
\lim_{N\rightarrow\infty} J_1^{\av,(N)} (\tpi^{(N)},\pi_{\varepsilon},\ldots,\pi_{\varepsilon}) &= \lim_{N\rightarrow\infty} J_{\mu_{\varepsilon}}^{\av}(\tpi^{(N)}) \nonumber \\
&\geq \lim_{N\rightarrow\infty} J_{\mu_*}^{\av}(\tpi^{(N)}) - C_2^{\av} \varepsilon \quad \text{(by (\ref{av-nneqq1}))}\nonumber \\
&\geq \inf_{\pi' \in \Pi_d} J_{\mu_*}^{\av}(\pi') - C_2^{\av} \varepsilon \nonumber \\
&= J_{\mu_*}^{\av}(\pi_*) - C_2^{\av} \varepsilon \nonumber \\
&\geq J_{\mu_*}^{\av}(\pi_{\varepsilon}) - C_2^{\av} \varepsilon - C_3^{\av} \varepsilon\quad \text{(by (\ref{av-nneqq2}))}\nonumber \\
&\geq J_{\mu_{\varepsilon}}^{\av}(\pi_{\varepsilon}) - 2C_2^{\av} \varepsilon - C_3^{\av} \varepsilon\quad \text{(by (\ref{av-nneqq1}))}\nonumber \\
&\eqqcolon J_{\mu_{\varepsilon}}^{\av}(\pi_{\varepsilon}) - \tau^{\av} \, \varepsilon \nonumber.
\end{align}
Note that by \cite[Theorem 4.10]{SaBaRaSIAM}, we also have 
$$
\lim_{N\rightarrow\infty} J_1^{\av,(N)} (\pi_{\varepsilon},\pi_{\varepsilon},\ldots,\pi_{\varepsilon}) = J_{\mu_{\varepsilon}}^{\av}(\pi_{\varepsilon}).
$$
Hence, there exists $N(\delta)$ such that for all $N\geq N(\delta)$, we have
\begin{align}
J_1^{\av,(N)} (\tpi^{(N)},\pi_{\varepsilon},\ldots,\pi_{\varepsilon}) + \frac{\delta}{3} &\geq J_{\mu_{\varepsilon}}^{\av}(\pi_{\varepsilon})- \tau^{\av} \, \varepsilon \nonumber \\
J_{\mu_{\varepsilon}}^{\av}(\pi_{\varepsilon}) + \frac{\delta}{3}&\geq J_1^{\av,(N)} (\pi_{\varepsilon},\pi_{\varepsilon},\ldots,\pi_{\varepsilon}). \nonumber  
\end{align}
Therefore, for all $N\geq N(\delta)$, we obtain
\begin{align}
\inf_{\pi' \in \Pi_1} J_1^{\av,(N)} (\pi',\pi_{\varepsilon},\ldots,\pi_{\varepsilon}) + \delta + \tau^{\av} \,\varepsilon 
&\geq J_1^{\av,(N)} (\tpi^{(N)},\pi_{\varepsilon},\ldots,\pi_{\varepsilon}) + \frac{2\delta}{3} + \tau^{\av} \, \varepsilon \nonumber \\
&\geq J_{\mu_{\varepsilon}}^{\av}(\pi_{\varepsilon}) + \frac{\delta}{3}   \nonumber \\
&\geq J_1^{\av,(N)} (\pi_{\varepsilon},\pi_{\varepsilon},\ldots,\pi_{\varepsilon}). \nonumber
\end{align}
\endproof

Theorem~\ref{av-old-main-cor} implies that, by learning $\varepsilon$-mean-field equilibrium policy in the infinite-population limit, one can obtain an approximate Nash equilibrium for the finite-agent game problem for which computing or learning the exact Nash equilibrium is in general  prohibitive. In the next section, we approximate the MFE operator $H^{\av}$ introduced in Section~\ref{av-known-model} via random operator $\hat{H}^{\av}$ to develop an algorithm for learning $\varepsilon$-mean-field equilibrium policy in the model-free setting.

\section{Learning Algorithm for Average-cost}\label{av-unknown-model}

In this section, we develop an offline learning algorithm to learn approximate mean-field equilibrium policy. Similar to the discounted-cost case, we assume that a generic agent has access to a simulator, which generates a new state $y \sim p(\,\cdot\,|x,a,\mu)$ and gives a cost $c(x,a,\mu)$ for any given state $x$, action $a$, and state-measure $\mu$.

In this learning algorithm, there are two stages in each iteration. In the first stage, we learn the $Q$-function $Q_{\mu}^{\av,*}$ upto a constant additive factor for a given $\mu$ using fitted $Q$-iteration algorithm. This stage replaces the operator $H_1^{\av}$ with a random operator $\hat{H}_1^{\av}$ that will be described below. Note that as opposed to the discounted-cost case, here, to construct the random operator $\hat{H}_1^{\av}$ that replaces the operator $H_1^{\av}$, we normally need an additional simulator that generates realizations of the minorizing sub-probability measure $\lambda$ in addition to the simulator for the transition probability $p(\,\cdot\,|x,a,\mu)$. However, this simulator is in general not available to the decision maker, since a generic agent does not know this minorizing sub-probability measure in the absence of the transition probability. Therefore, we need to modify the approach used in the discounted-cost case appropriately for the average-cost setup. Indeed, this is achieved by performing convergence analysis of the random operator $\hat{H}_1^{\av}$ using span-seminorm instead of sup-norm on $Q$-functions. Luckily, convergence analysis of the learning algorithm established using sup-norm in discounted-cost case can easily be adapted to the span-seminorm.

We select $Q$-functions from a fixed function class ${\cal F}^{\av}$ such as the set of neural networks with some fixed architecture or linear span of some finite number of basis functions. Depending on this choice, there will be an additional representation error in the learning algorithm, which is in general negligible. Let ${\cal F}_{\min}^{\av} \coloneqq \{Q_{\min}: Q \in {\cal F}^{\av}\}$.

In the second stage, we update the state-measure by approximating the transition probability via its empirical estimate. This stage replaces the operator $H_2^{\av}$ in the model-based algorithm with a random operator $\hat{H}_2^{\av}$. Indeed, since $H_2^{\av}=H_2$, we also have $\hat{H}_2^{\av}=\hat{H}_2$, and so, the error analysis of $\hat{H}_2^{\av}$ is exactly the same with the error analysis of $\hat{H}_2$.

We proceed with the definition of the random operator $\hat{H}_1^{\av}$. To describe $\hat{H}_1^{\av}$, we need to pick a probability measure $\nu$ on $\sX$ and a policy $\pi_b \in \Pi$. Indeed, we can choose $\nu$ and $\pi_b$ as in discounted-cost case. Recall the constants $\zeta_0 \coloneqq 1/\sqrt{\min_x \nu(x)}$ and $\pi_0 \coloneqq \inf_{(x,a) \in \sX\times\sA} \pi_b(a|x) > 0$. Now, we can give the definition of the random operator $\hat{H}_1^{\av}$.

\begin{algorithm}[H]
\caption{Algorithm $\hat{H}_1^{\av}$}
\label{av-H1}
\begin{algorithmic}
\small
\STATE{Input $\mu$, Data size $N$, Number of iterations $L$ }
\STATE{Generate i.i.d. samples $\{(x_t,a_t,c_t,y_{t+1})_{t=1}^N\}$ using
$$
x_t \sim \nu, \, a_t \sim \pi_b(\cdot|x_t), \, c_t = c(x_t,a_t,\mu), \, y_{t+1} \sim p(\cdot|x_t,a_t,\mu)
$$
}
\STATE{Start with $Q_0 = 0$}
\FOR{$l=0,\ldots,L-1$}
\STATE{ 
\begin{align}
&\hspace{-20pt}Q_{l+1} = \argmin_{f \in {\cal F}} \frac{1}{N} \sum_{t=1}^N \displaystyle \frac{1}{m(\sA) \, \pi_b(a_t|x_t)} \left| f(x_t,a_t) - \left[c_t + \min_{a' \in \sA} Q_l(y_{t+1},a') \right]\right|^2  \nonumber 
\end{align}
}
\ENDFOR
\RETURN{$Q_L$}
\end{algorithmic}
\end{algorithm}
\normalsize

Note that if we used the same method as in the discounted-cost case, we should have generated $y_{t+1}$ using $q(\,\cdot\,|x,a,\mu) \coloneqq p(\,\cdot\,|x,a,\mu)-\lambda(\,\cdot\,)$ instead of $p(\,\cdot\,|x,a,\mu)$ since $H_1^{\av}(\mu)$ gives the unique fixed point of the contraction operator $H_{\mu}^{\av}$ on $Q$-functions given by
$$
H_{\mu}^{\av} Q(x,a) = c(x,a,\mu) + \sum_y Q_{\min}(y) \, q(y|x,a,\mu).
$$
However, in general, a generic agent does not have access to a simulator for $\lambda$, and so, we must construct the algorithm as above using the simulator for $p(\,\cdot\,|x,a,\mu)$. As a consequence of this, we perform the error analysis of the above learning algorithm in terms of span-seminorm instead of sup-norm. To this end, for any $\mu$, define the following operator on $Q$-functions:
$$
R_{\mu}^{\av} Q(x,a) \coloneqq c(x,a,\mu) + \sum_{y\in\sX} Q_{\min}(y) \, p(dy|x,a,\mu). 
$$ 
The operator $R_{\mu}^{\av}$ is different from $H_{\mu}^{\av}$ in this case, and it is used in the proof of the error analysis of $\hat{H}_1^{\av}$ in place of $H_{\mu}^{\av}$.

To do error analysis of $\hat{H}_1^{\av}$, we need to define the following constants:
\small
\begin{align}
&E({\cal F})^{\av} \coloneqq \sup_{\mu \in \P(\sX)} \sup_{Q \in {\cal F}^{\av}} \inf_{Q' \in {\cal F}^{\av}} \|Q'-R_{\mu}^{\av}Q\|_{\nu}, \, L_{\bf m}^{\av} \coloneqq 2 Q_{\bf m}^{\av} + c_{\bf m}, \, C^{\av} \coloneqq \frac{(L_{\bf m}^{\av})^2}{m(\sA) \, \pi_0} \nonumber \\
&\Upsilon^{\av} = 8 \, e^2 \, (V_{{\cal F}^{\av}}+1) \, (V_{{\cal F}_{\min}^{\av}}+1) \, \left(\frac{128 e Q_{\bf m}^{\av} L_{\bf m}^{\av}}{m(\sA) \pi_0}\right)^{V_{{\cal F}^{\av}}+V_{{\cal F}_{\min}^{\av}}}, \, V^{\av} = V_{{\cal F}^{\av}}+V_{{\cal F}_{\min}^{\av}}, \, \gamma^{\av} = 512 (C^{\av})^2 \nonumber \\
&\Delta^{\av} \coloneqq \frac{2}{1-\tbeta} \left[ \frac{m(\sA) (\dim_{\sA}+1)!\zeta_0}{\alpha (2/Q_{\Lip}^{\av})^{\dim_{\sA}}} \, E({\cal F})^{\av} \right]^{\frac{1}{\dim_{\sA}+1}}, \, \Lambda^{\av} \coloneqq \frac{2}{1-\tbeta} \left[ \frac{ m(\sA) (\dim_{\sA}+1)!\zeta_0}{\alpha (2/Q_{\Lip}^{\av})^{\dim_{\sA}}} \right]^{\frac{1}{\dim_{\sA}+1}}, \nonumber 
\end{align}
\normalsize
where 
$$
(0,1) \ni \tbeta \coloneqq 1 - \lambda(\sX)/2 \geq \beta^{\av} \coloneqq 1 - \lambda(\sX).
$$

The below theorem gives the error analysis of the algorithm $\hat{H}_1$. Before stating it, let us recall the definition of span-seminorm of any function $g:\sE\rightarrow\R$ defined on some set $\sE$:
$$
\spn(g) \coloneqq \sup_{e \in \sE} g(e) - \inf_{e \in \sE} g(e). 
$$
It is a seminorm because $\spn(g)=0$ if and only if $g$ is a constant function. Moreover, 
\begin{align}
\spn(g) &\coloneqq \sup_{e \in \sE} g(e) - \inf_{e \in \sE} g(e) \nonumber \\
&= \sup_{e \in \sE} g(e) + \sup_{e \in \sE} -g(e) \nonumber \\
&\leq 2 \, \|g\|_{\infty}.\nonumber
\end{align} 
Hence, we can upper bound span-seminorm via sup-norm.

\begin{theorem}\label{av-Thm-H1}
For any $(\varepsilon,\delta) \in (0,1)^2$, with probability at least $1-\delta$, we have 
$$
\spn\left(\hat{H}_1^{\av}[N,L](\mu) - H_1^{\av}(\mu)\right) \leq \varepsilon + \Delta^{\av}
$$
if $\frac{4\tbeta^L}{1-\tbeta} \, Q_{\bf m}^{\av} < \frac{\varepsilon}{2}$ and $N \geq m_1^{\av}(\varepsilon,\delta,L)$, where 
$$
m_1^{\av}(\varepsilon,\delta,L) \coloneqq \frac{\gamma^{\av} (2 \Lambda^{\av})^{4(\dim_{\sA}+1)}}{\varepsilon^{4(\dim_{\sA}+1)}} \, \ln\left(\frac{\Upsilon^{\av} (2 \Lambda^{\av})^{2V^{\av}(\dim_{\sA}+1)} L}{\delta \varepsilon^{2V^{\av}(\dim_{\sA}+1)}}\right). 
$$
Here, the constant error $\Delta^{\av}$ is as a result of the representation error $E({\cal F})^{\av}$ in the algorithm, which is in general negligible.  
\end{theorem}

\begin{proof}
Recall the definition of $\|\,\cdot\,\|_{\nu}$-norm on $Q$-functions:
$$
\|Q\|_{\nu}^2 \coloneqq \sum_{x \in \sX} \int_{\sA} Q(x,a)^2 \, m_{\sA}(da) \, \nu(x). 
$$
For any $\mu$, recall also the definition of the operator on $Q$-functions:
$$
R_{\mu}^{\av} Q(x,a) \coloneqq c(x,a,\mu) + \sum_{y\in\sX} Q_{\min}(y) \, p(dy|x,a,\mu). 
$$
This is very similar to the operator $H_{\mu}^{\av}$, but it is not $\beta^{\av}$-contraction with respect to sup-norm. Indeed, the operator $R_{\mu}^{\av}$ is $\tbeta$-contraction with respect to span-seminorm: 
$$
\spn(R_{\mu}^{\av}Q_1-R_{\mu}^{\av}Q_2) \leq \tbeta \, \spn(Q_1-Q_2),
$$
where 
$$
\tbeta \coloneqq 1 - \lambda(\sX)/2 \geq \beta^{\av} \coloneqq 1 - \lambda(\sX).
$$
Indeed, for any function $v:\sX\rightarrow\R$, one can prove that (see \cite[proof of Lemma 3.3 and proof of Lemma 3.5]{Her89})
\begin{align}
\sum_y v(y) \, p(y|x,a,\mu) - \sum_y v(y) \, p(y|x',a',\mu) \leq \tbeta \, \spn(v). \label{vvv}
\end{align}
Now let $Q_1$ and $Q_2$ be two $Q$-functions. Then, for any $(x,a)$ and $(x',a')$, we have 
\begin{align}
(R_{\mu}^{\av} Q_1 &- R_{\mu}^{\av} Q_2)(x,a)-(R_{\mu}^{\av} Q_1 - R_{\mu}^{\av} Q_2)(x',a') \nonumber \\
&= \sum_y \left\{ Q_{1,\min}(y)-Q_{2,\min}(y) \right\} \, p(y|x,a,\mu)-\sum_y \left\{ Q_{1,\min}(y)-Q_{2,\min}(y) \right\} \, p(y|x',a',\mu) \nonumber \\
&\leq \tbeta \, \spn(Q_{1,\min}-Q_{2,\min}) \quad \text{(by (\ref{vvv}))} \nonumber \\
&\leq \tbeta \, \spn(Q_{1}-Q_{2}). \nonumber 
\end{align}
This implies that 
$$
\spn(R_{\mu}^{\av}Q_1-R_{\mu}^{\av}Q_2) \leq \tbeta \, \spn(Q_1-Q_2),
$$
which means that $R_{\mu}^{\av}$ is span-seminorm $\tbeta$-contraction. Moreover, $H_1^{\av}(\mu) \coloneqq Q^{\av,*}_{\mu}$ is a fixed point of $R_{\mu}^{\av}$ with respect to span-seminorm; that is,
$$
\spn(R_{\mu}^{\av} H_1^{\av}(\mu)-H_1^{\av}(\mu)) = 0.  
$$
Indeed, for all $x,a$, we have 
\begin{align}
R_{\mu}^{\av} H_1^{\av}(\mu)(x,a) &- H_1^{\av}(\mu)(x,a) \nonumber \\
&= c(x,a,\mu) + \sum_y Q^{\av,*}_{\mu,\min} \, p(y|x,a,\mu) -c(x,a,\mu) - \sum_y Q^{\av,*}_{\mu,\min} \, q(y|x,a,\mu) \nonumber\\
&= \sum_y Q^{\av,*}_{\mu,\min} \, \lambda(y) \,\,\, \text{(i.e., constant)} \nonumber
\end{align}
and so, $\spn(R_{\mu}^{\av} H_1^{\av}(\mu)-H_1^{\av}(\mu)) = 0$. Hence, $H_1^{\av}(\mu)$ is a fixed point of $R_{\mu}^{\av}$ with respect to span-seminorm. Since $R_{\mu}^{\av}$ is also $\tbeta$-contraction with respect to span-seminorm, one can also prove that for all $L>1$, we have
\begin{align}
\spn((R_{\mu}^{\av})^L Q - H_1^{\av}(\mu)) &\leq \frac{\tbeta^L}{1-\tbeta} \spn(Q-R_{\mu}^{\av}Q) \nonumber \\
&\leq \frac{2\tbeta^L}{1-\tbeta} \|Q-R_{\mu}^{\av}Q\|_{\infty} \nonumber \\
&\leq \frac{4\tbeta^L}{1-\tbeta} Q_{\bf m}^{\av} \nonumber 
\end{align}
for any $Q \in {\cal C}^{\av}$ as $\|Q\|_{\infty} \leq Q_{\bf m}^{\av}$. 

Now, using above results, we easily complete the proof by using the same techniques as in the proof of Theorem~\ref{Thm-H1}. Let $Q_l$ be the random $Q$-function at the $l^{\text{th}}$-step of the algorithm. First, we find an upper bound to the following probability
$$
P_0 \coloneqq \rP\left( \|Q_{l+1}-R_{\mu}^{\av} Q_l\|_{\nu}^2 > (E({\cal F})^{\av})^2 + \varepsilon' \right),
$$
for a given $\varepsilon'>0$. To that end, we define
\begin{align}
&\hat{L}_N(f;Q) \coloneqq \frac{1}{N} \sum_{t=1}^N \frac{1}{m(\sA) \, \pi_b(a_t|x_t)} \left| f(x_t,a_t) - \left[c_t + \min_{a' \in \sA} Q(y_{t+1},a') \right]\right|^2. \nonumber 
\end{align}
As in the proof of Theorem~\ref{Thm-H1}, one can show that  
$$
\rE \left[ \hat{L}_N(f;Q) \right] = \|f-R_{\mu}^{\av}Q\|_{\nu}^2 + L^{\av,*}(Q) \eqqcolon L^{\av}(f;Q),
$$
where $L^{\av,*}(Q)$ is some quantity independent of $f$.

Now, using exactly the same steps as in the proof of Theorem~\ref{Thm-H1}, we can obtain the following bound on the probability $P_0$:
\begin{align}
P_0 \leq \Upsilon^{\av} \, \varepsilon'^{-V^{\av}} \, e^{\frac{-N \varepsilon'^2}{\gamma^{\av}}} \eqqcolon \frac{\delta'}{L}. \label{av-third}
\end{align}
The only difference is that in this case, we take $\beta=1$. 
Hence, for each $l=0,\ldots,L-1$, with probability at most $\frac{\delta'}{L}$
$$
\|Q_{l+1} - R_{\mu}^{\av} Q_l\|_{\nu}^2 > \varepsilon' + (E({\cal F})^{\av})^2.
$$
This implies that with probability at most $\frac{\delta'}{L}$
$$
\|Q_{l+1} - R_{\mu}^{\av} Q_l\|_{\nu} > \sqrt{\varepsilon'} + E({\cal F})^{\av}.
$$
Using this and the fact that $R_{\mu}^{\av}$ is $\tbeta$-contraction with respect to span-seminorm, we can conclude that with probability at least $1-\delta'$, we have
\begin{align}
&\spn(Q_L-H_1^{\av}(\mu)) \leq \sum_{l=0}^{L-1} \tbeta^{L-(l+1)} \, \spn(Q_{l+1} - R_{\mu}^{\av} Q_l) + \spn((R_{\mu}^{\av})^L Q_0 - H_1^{\av}(\mu)) \nonumber \\
&\leq 2 \left(\sum_{l=0}^{L-1} \tbeta^{L-(l+1)} \, \|Q_{l+1} - R_{\mu}^{\av} Q_l\|_{\infty} + \frac{2\tbeta^L}{1-\tbeta} \, Q_{\bf m}^{\av} \right) \nonumber \\
&\overset{(I)}{\leq} 2 \sum_{l=0}^{L-1} \tbeta^{L-(l+1)} \, \left[ \frac{m(\sA) (\dim_{\sA}+1)!\zeta_0}{\alpha (2/Q_{\Lip}^{\av})^{\dim_{\sA}}} \, \|Q_{l+1} - R_{\mu}^{\av} Q_l\|_{\nu} \right]^{\frac{1}{\dim_{\sA}+1}} +\frac{4\tbeta^L}{1-\tbeta} \, Q_{\bf m}^{\av}\nonumber \\
&\leq 2 \sum_{l=0}^{L-1} \tbeta^{L-(l+1)} \, \left[ \frac{ m(\sA) (\dim_{\sA}+1)!\zeta_0}{\alpha (2/Q_{\Lip}^{\av})^{\dim_{\sA}}} \, (\sqrt{\varepsilon'}+E({\cal F})^{\av}) \right]^{\frac{1}{\dim_{\sA}+1}} +\frac{4\tbeta^L}{1-\tbeta} \, Q_{\bf m}^{\av}\nonumber \\
&\leq \frac{2}{1-\tbeta} \bigg( \left[ \frac{m(\sA) (\dim_{\sA}+1)!\zeta_0}{\alpha (2/Q_{\Lip}^{\av})^{\dim_{\sA}}} \, E({\cal F})^{\av} \right]^{\frac{1}{\dim_{\sA}+1}} \nonumber \\
&\phantom{xxxxxxxxxxxx}+  \left[ \frac{ m(\sA) (\dim_{\sA}+1)!\zeta_0}{\alpha (2/Q_{\Lip}^{\av})^{\dim_{\sA}}} \right]^{\frac{1}{\dim_{\sA}+1}} \varepsilon'^{\frac{1}{2(\dim_{\sA}+1)}}\bigg) +\frac{4\tbeta^L}{1-\tbeta} \, Q_{\bf m}^{\av}, \nonumber
\end{align}
where (I) follows from Lemma~\ref{li-l2}. Therefore, with probability at least $1-\delta'$, we have
\begin{align}
\spn(Q_L-H_1^{\av}(\mu)) \leq \Lambda^{\av} \varepsilon'^{\frac{1}{2(\dim_{\sA}+1)}} + \Delta^{\av} + \frac{4\tbeta^L}{1-\tbeta} \, Q_{\bf m}^{\av}. \label{av-fourth}
\end{align}
Now, the result follows by picking $\delta = \delta' \coloneqq L \, \Upsilon^{\av} \, \varepsilon'^{-V^{\av}} \, e^{\frac{-N \varepsilon'^2}{\gamma^{\av}}}$, $\Lambda^{\av} \varepsilon'^{\frac{1}{2(\dim_{\sA}+1)}} = \varepsilon/2$, and $\frac{4\tbeta^L}{1-\tbeta} \, Q_{\bf m}^{\av} = \varepsilon/2$.  
\end{proof}

We now give the description of the random operator $\hat{H}_2^{\av}$. In this algorithm, the goal is to replace the operator $H_2^{\av}$, which gives the next state-measure, with $\hat{H}_2^{\av}$. Since $H_2^{\av}=H_2$, we also have $\hat{H}_2^{\av}=\hat{H}_2$. Therefore,  the error analysis of $\hat{H}_2^{\av}$ is exactly the same with Theorem~\ref{Thm-H2}.

\begin{algorithm}[H]
\caption{Algorithm $\hat{H}_2^{\av}$}
\label{av-H2}
\begin{algorithmic}
\small
\STATE{Inputs $(\mu,Q)$, Data size $M$, Number of iterations $|\sX|$}
\FOR{$x \in \sX$}
\STATE{generate i.i.d. samples $\{y_t^x\}_{t=1}^M$ using
$$
y_t^x \sim p(\cdot|x,f_{Q}(x),\mu)
$$
and define 
$$
p_M(\cdot|x,f_{Q}(x),\mu) = \frac{1}{M} \sum_{t=1}^M \delta_{y_t^x}(\cdot).
$$
}
\ENDFOR
\RETURN{$\sum_{x \in \sX} p_M(\cdot|x,f_{Q}(x),\mu) \, \mu(x)$}
\end{algorithmic}
\end{algorithm}
\normalsize

This is the error analysis of the random operator $\hat{H}_2^{\av}$.

\begin{theorem}\label{av-Thm-H2}
For any $(\varepsilon,\delta) \in (0,1)^2$, with probability at least $1-\delta$
$$
\left\|\hat{H}_2^{\av}[M](\mu,Q) - H_2^{\av}(\mu,Q) \right\|_1 \leq \varepsilon 
$$
if $M \geq m_2^{\av}(\varepsilon,\delta)$, where 
$$
m_2^{\av}(\varepsilon,\delta) \coloneqq \frac{|\sX|^2}{\varepsilon^2} \, \ln\left(\frac{2 \, |\sX|^2}{\delta}\right). 
$$
\end{theorem}

\begin{proof}
See the proof of Theorem~\ref{Thm-H2}. 
\end{proof}

Below, we give the overall description of the learning algorithm for the average-cost. In this algorithm, we successively apply the random operator $\hat{H}^{\av}$, which replaces MFE operator $H^{\av}$, to obtain approximate mean-field equilibrium policy.

\begin{algorithm}[H]
\caption{Learning Algorithm}
\label{av-Qit}
\begin{algorithmic}
\small
\STATE{Input $\mu_0$, Number of iterations $K$, Parameters of $\hat{H}_1^{\av}$ and $\hat{H}_2^{\av}$ $\left(\{[N_k,L_k]\}_{k=0}^{K-1},\{M_k\}_{k=0}^{K-1}\right)$}
\normalsize
\STATE{Start with $\mu_0$}
\FOR{$k=0,\ldots,K-1$}
\STATE{
$\mu_{k+1} = \hat{H}^{\av}\left([N_k,L_k],M_k\right)(\mu_k) \coloneqq \hat{H}_2^{\av}[M_k]\left(\mu_k,\hat{H}_1^{\av}[N_k,L_k](\mu_k)\right)$
}
\ENDFOR
\RETURN{$\mu_K$}
\end{algorithmic}
\end{algorithm}

Above, we have completed the error analyses of the operators $\hat{H}_1^{\av}$ and $\hat{H}_2^{\av}$ in Theorem~\ref{av-Thm-H1} and Theorem~\ref{av-Thm-H2}, respectively. Since the random operator $\hat{H}^{\av}$ is a composition of $\hat{H}_1^{\av}$ with $\hat{H}_2^{\av}$, we can obtain the following error analysis for the operator $\hat{H}^{\av}$.

\begin{theorem}\label{av-main-result}
Fix any $(\varepsilon,\delta) \in (0,1)^2$. Define 
$$
\varepsilon_1 \coloneqq \frac{\rho \, (1-K_{H^{\av}})^2 \, \varepsilon^2}{32 (K_1)^2}, \,\,
\varepsilon_2 \coloneqq \frac{(1-K_{H^{\av}}) \, \varepsilon}{4}.
$$
Let $K,L$ be such that 
\begin{align}
\frac{(K_{H^{\av}})^K}{1-K_{H^{\av}}} &\leq \frac{\varepsilon}{2}, \,\, \frac{4\tbeta^L}{1-\tbeta} Q_{\bf m}^{\av} \leq \frac{\varepsilon_1}{2}. \nonumber 
\end{align}
Then, pick $N,M$ such that
\begin{align}
N &\geq m_1^{\av}\left( \varepsilon_1,\frac{\delta}{2K},L \right), \,\, M \geq m_2^{\av}\left( \varepsilon_2,\frac{\delta}{2K} \right).
\end{align}
Let $\mu_K$ be the output of the learning algorithm $\hat{H}^{\av}$ with inputs $$\left(K, \{[N,L]\}_{k=0}^K, \{M\}_{k=0}^{K-1}, \mu_0 \right).$$ Then, with probability at least $1-\delta$
$$
\|\mu_K - \mu_*^{\av}\|_1 \leq \frac{K_1 \sqrt{2 \Delta^{\av} }}{\sqrt{\rho} (1-K_{H^{\av}})} + \varepsilon,
$$
where $\mu_*^{\av}$ is the state-measure in mean-field equilibrium given by the MFE operator $H^{\av}$. 
\end{theorem}

\begin{proof}
The proof is similar to the proof of Theorem~\ref{main-result}. The key difference is that we perform the analysis in terms of span-seminorm in place of sup-norm. 

For any $\mu \in \P(\sX)$, $Q \in {\cal F}^{\av}$, $\hQ \in {\cal C}^{\av}$, we have
\begin{align}
\|H_2^{\av}(\mu,Q) - H_2^{\av}(\mu,\hQ)\|_1 
&= \sum_{y \in \sX} \left| \sum_{x \in \sX} p(y|x,f_{Q}(x),\mu) \, \mu(x) - \sum_{x \in \sX} p(y|x,f_{\hQ}(x),\mu) \, \mu(x) \right| \nonumber \\
&\leq \sum_{x \in \sX} K_1 \, \|f_{Q}(x)-f_{\hQ}(x)\| \, \mu(x). \label{av-fbound}
\end{align}
Now, using exactly the same steps as in the proof of Theorem~\ref{main-result}, for any $x \in \sX$, we have
$$
\|f_{Q}(x) - f_{\hQ}(x)\|^2 \leq \frac{2}{\rho} \, \left(\hQ(x,f_{Q}(x)) - \hQ(x,f_{\hQ}(x))\right) 
$$
For any $x \in \sX$, this leads to
\begin{align}
\|f_{\hQ}(x) - f_{Q}(x)\|^2 &\leq \frac{2}{\rho} \, \bigg(\hQ(x,f_{Q}(x)) - Q(x,f_{Q}(x)) +Q(x,f_{Q}(x)) - \hQ(x,f_{\hQ}(x))\bigg) \nonumber \\
&= \frac{2}{\rho} \, \bigg(\hQ(x,f_{Q}(x)) - Q(x,f_{Q}(x)) + \min_{a \in \sA}Q(x,a) - \min_{a \in \sA} \hQ(x,a)\bigg) \nonumber \\
&\leq \frac{2}{\rho} \, \bigg(\sup_{(z,a) \in \sX\times\sA} (\hQ(z,a) - Q(z,a)) + \sup_{(z,a) \in \sX\times\sA} (Q(z,a) - \hQ(z,a))\bigg) \nonumber \\
&= \frac{2}{\rho} \, \spn(Q-\hQ). \label{av-perturbation2}
\end{align}
Therefore, here, we can also perform a similar analysis as in the proof of Theorem~\ref{main-result} using span-seminorm in place of sup-norm. 

Now, combining (\ref{av-fbound}) and (\ref{av-perturbation2}) yields
\small
\begin{align}
\|H_2^{\av}(\mu,Q) - H_2^{\av}(\mu,\hQ)\|_1 \leq \frac{\sqrt{2} K_1}{\sqrt{\rho}} \, \sqrt{\spn(Q-\hQ)}. \label{av-mainbound}
\end{align}
\normalsize
Using (\ref{av-mainbound}) and the fact that $H_1^{\av}(\mu_k) \in {\cal C}^{\av}$ and $\hat{H}_1^{\av}[N,L](\mu_k)) \in {\cal F}^{\av}$, for any $k=0,\ldots,K-1$, we have
\begin{align}
\|H^{\av}(\mu_k) - \hat{H}^{\av}([N,L],M)(\mu_k)\|_1 &\leq \|H_2^{\av}(\mu_k,H_1^{\av}(\mu_k)) - H_2^{\av}(\mu_k,\hat{H}_1^{\av}[N,L](\mu_k))\|_1 \nonumber \\
&+ \|H_2^{\av}(\mu_k,\hat{H}_1^{\av}[N,L](\mu_k)) - \hat{H}_2^{\av}[M](\mu_k,\hat{H}_1^{\av}[N,L](\mu_k))\|_1  \nonumber \\
&\leq \frac{\sqrt{2}K_1}{\sqrt{\rho}} \, \sqrt{\spn(H_1^{\av}(\mu_k) - \hat{H}_1^{\av}[N,L](\mu_k))} \nonumber \\
&+ \|H_2^{\av}(\mu_k,\hat{H}_1^{\av}[N,L](\mu_k)) - \hat{H}_2^{\av}[M](\mu_k,\hat{H}_1^{\av}[N,L](\mu_k))\|_1. \nonumber 
\end{align}
The last term is upper bounded by 
$$\frac{K_1 \sqrt{2 (\varepsilon_1 + \Delta^{\av})}}{\sqrt{\rho}} + \varepsilon_2$$
with probability at least $1-\frac{\delta}{K}$ by Theorem~\ref{av-Thm-H1} and Theorem~\ref{av-Thm-H2}. Therefore, with probability at least $1-\delta$ 
\begin{align}
\|\mu_K - \mu_*^{\av}\|_1 &\leq \sum_{k=0}^{K-1} K_{H^{\av}}^{K-(k+1)} \, \|\hat{H}^{\av}([N,L],M)(\mu_k) - H^{\av}(\mu_k)\|_1 + \|(H^{\av})^K(\mu_0) - \mu_*^{\av}\|_1 \nonumber \\
&\leq \sum_{k=0}^{K-1} K_{H^{\av}}^{K-(k+1)} \left(\frac{K_1 \sqrt{2  (\varepsilon_1 + \Delta^{\av})}}{\sqrt{\rho}} + \varepsilon_2\right) + \frac{(K_{H^{\av}})^K}{1-K_{H^{\av}}} \nonumber \\
&\leq \frac{K_1 \sqrt{2\, \Delta^{\av} }}{\sqrt{\rho} (1-K_{H^{\av}})} + \varepsilon. \nonumber 
\end{align}
This completes the proof.
\end{proof}

Now, we state the main result of this section. It states that, by using a learning algorithm, one can learn approximate mean-field equilibrium policy. By Theorem~\ref{av-old-main-cor}, this gives an approximate Nash-equilibrium for the finite-agent game.

\begin{corollary}\label{av-main-cor}
Fix any $(\varepsilon,\delta) \in (0,1)^2$. Suppose that $K,L,N,M$ satisfy the conditions in Theorem~\ref{av-main-result}. Let $\mu_K$ be the output of the earning algorithm with inputs $$\left(K, \{[N,L]\}_{k=0}^K, \{M\}_{k=0}^{K-1}, \mu_0 \right).$$ Define $\pi_K(x) \coloneqq \argmin_{a \in \sA} Q_K(x,a)$, where $Q_K \coloneqq \hat{H}_1^{\av}([N,L])(\mu_K)$. Then, with probability at least $1-\delta(1+\frac{1}{2K})$, the policy $\pi_K$ is a $\kappa^{\av}(\varepsilon,\Delta)$-mean-field equilibrium policy, where
\begin{align}
&\kappa^{\av}(\varepsilon,\Delta) = \sqrt{\frac{2}{\rho}\left(\frac{\rho \, (1-K_{H^{\av}})^2 \, \varepsilon^2}{32 (K_1)^2} + \Delta + 2 K_{H_1^{\av}} \left(\frac{K_1 \sqrt{2\, \Delta }}{\sqrt{\rho} (1-K_{H^{\av}})} + \varepsilon\right)\right)}. \nonumber
\end{align}
Therefore, by Theorem~\ref{av-old-main-cor}, an $N$-tuple of policies ${\bf \pi}^{(N)} = \{\pi_K,\pi_K,\ldots,\pi_K\}$ is a $\tau^{\av} \kappa^{\av}(\varepsilon,\Delta) + \sigma$-Nash equilibrium for the game with $N \geq N(\sigma)$ agents. 
\end{corollary}

\begin{proof}
By Theorem~\ref{av-main-result}, with probability at least $1-\delta(1+\frac{1}{2K})$, we have 
\begin{align}
\spn(Q_K - H_1^{\av}(\mu_*^{\av})) &\leq \spn(Q_K - H_1^{\av}(\mu_K)) + 2 \|H_1^{\av}(\mu_K) - H_1^{\av}(\mu_*^{\av})\|_{\infty} \nonumber \\
&\leq \varepsilon_1 + \Delta^{\av} + 2 K_{H_1^{\av}} \|\mu_K - \mu_*^{\av}\|_1 \nonumber \\
&\leq \varepsilon_1 + \Delta^{\av} + 2 K_{H_1^{\av}} \left(\frac{K_1 \sqrt{2 \, \Delta^{\av} }}{\sqrt{\rho} (1-K_{H^{\av}})} + \varepsilon\right) \nonumber \\
&= \frac{\rho \, (1-K_{H^{\av}})^2 \, \varepsilon^2}{32 (K_1)^2} + \Delta^{\av} + 2 K_{H_1^{\av}} \left(\frac{K_1 \sqrt{2, \Delta^{\av} }}{\sqrt{\rho} (1-K_{H^{\av}})} + \varepsilon\right). \nonumber
\end{align}
Let $\pi_K(x) \coloneqq \argmin_{a \in \sA} Q_K(x,a)$. Using the same analysis that leads to (\ref{av-perturbation2}), we can obtain the following bound:
\begin{align}
\sup_{x \in \sX} \|\pi_K(x)-\pi_*(x)\|^2 &\leq \frac{2}{\rho} \, \spn(Q_K-H_1^{\av}(\mu_*^{\av})). \nonumber 
\end{align}
Hence, with probability at least $1-\delta(1+\frac{1}{2K})$, the policy $\pi_K$ is a $\kappa(\varepsilon,\Delta)$-mean-field equilibrium policy.
\end{proof}

\begin{remark}
Note that, in Corollary~\ref{av-main-cor}, there is a constant $\Delta^{\av}$, which depends on the representation error $E({\cal F})^{\av}$. In general, $E({\cal F})^{\av}$ is negligible. Hence, in this case, we have the following error bound:
\begin{align}
&\kappa^{\av}(\varepsilon,0) = \sqrt{\frac{2}{\rho}\left(\frac{\rho \, (1-K_{H^{\av}})^2 \, \varepsilon^2}{32 (K_1)^2} + 2 K_{H_1^{\av}} \varepsilon\right)}. \nonumber
\end{align}
which goes to zero as $\varepsilon \rightarrow 0$.  
\end{remark}

\section{Numerical Examples}\label{num_ex}

In this section, we present two numerical examples in the case of discounted cost and average cost, respectively, to demonstrate the applicability of our learning algorithm.

\subsection{Discounted Cost}\label{num_ex_discount}

We consider the mean-field game that was introduced in Example~\ref{example1}, where we take $\sX = [0, 0.1, 0.,2 , \ldots, 1]$, $\sA=[0,1]$, $c_2(a) = \rho \, a^2$, and $c_1(x,\mu) = \eta \, x \, (1-\xi \, \langle \mu \rangle)$ with $\langle \mu \rangle$ denoting the mean of $\mu$. We also take 
$\rP[h(a,\mu,w)=0.1] = \kappa \, a \, (1-\gamma \, \langle \mu \rangle)$, $\rP[h(a,\mu,w)=0] = \kappa \, a \, \gamma \, \langle \mu \rangle$, and $\rP[h(a,\mu,w)=-0.1] = \kappa \, a$; that is, the state can only go one unit up, go one unit down, or remain the same. Here, the constants $\rho, \eta, \xi, \kappa, \gamma$ are all non-negative. As the conditions (i) and (ii) are clearly satisfied in Example~\ref{example1} for these particular choices of system components, Assumption~\ref{as1} holds true in this case. In the numerical experiments, we use the following values for the parameters:
$
	\eta=2, \quad \xi = 0.4, \quad \rho=1 \quad
	\kappa=1, \quad \gamma=0.4, \quad \beta=0.9.  
$
We run the learning algorithm using the following parameters: $N=10000, L=50, M=1000, K=50$. The output of the learning algorithm contains the average of the state-measure (i.e., mean-field distribution) and mean-field equilibrium policies for states $x=0.1$ and $x=0.6$. In the fitted $Q$-iteration algorithm, we pick the function class $\F$ as two-layer neural networks with $10$ hidden units. We use neural network fitting tool of MATLAB.
In particular, we use `fittnet', `train', and `net' functions of MATLAB, where `Levenberg-Marquardt' is picked as the training algorithm and the transfer function is chosen as `hyperbolic tangent sigmoid transfer function'. The parameters of the neural network fitting tool of MATLAB are set to default values. We also run the value iteration algorithm using MFE operator $H$ to find the true average of state-measure and mean-field equilibrium policies for states $x=0.1$ and $x=0.6$. Then, we compare the learned outputs with outputs of the value iteration algorithm. Figures~\ref{meanfield_comp}, \ref{policy_comp}, and \ref{cost_comp} show this comparison. It can be seen that learned outputs converge to the outputs of the value iteration algorithm. 

\begin{figure}[h]
\centering
	\includegraphics[scale=0.5]{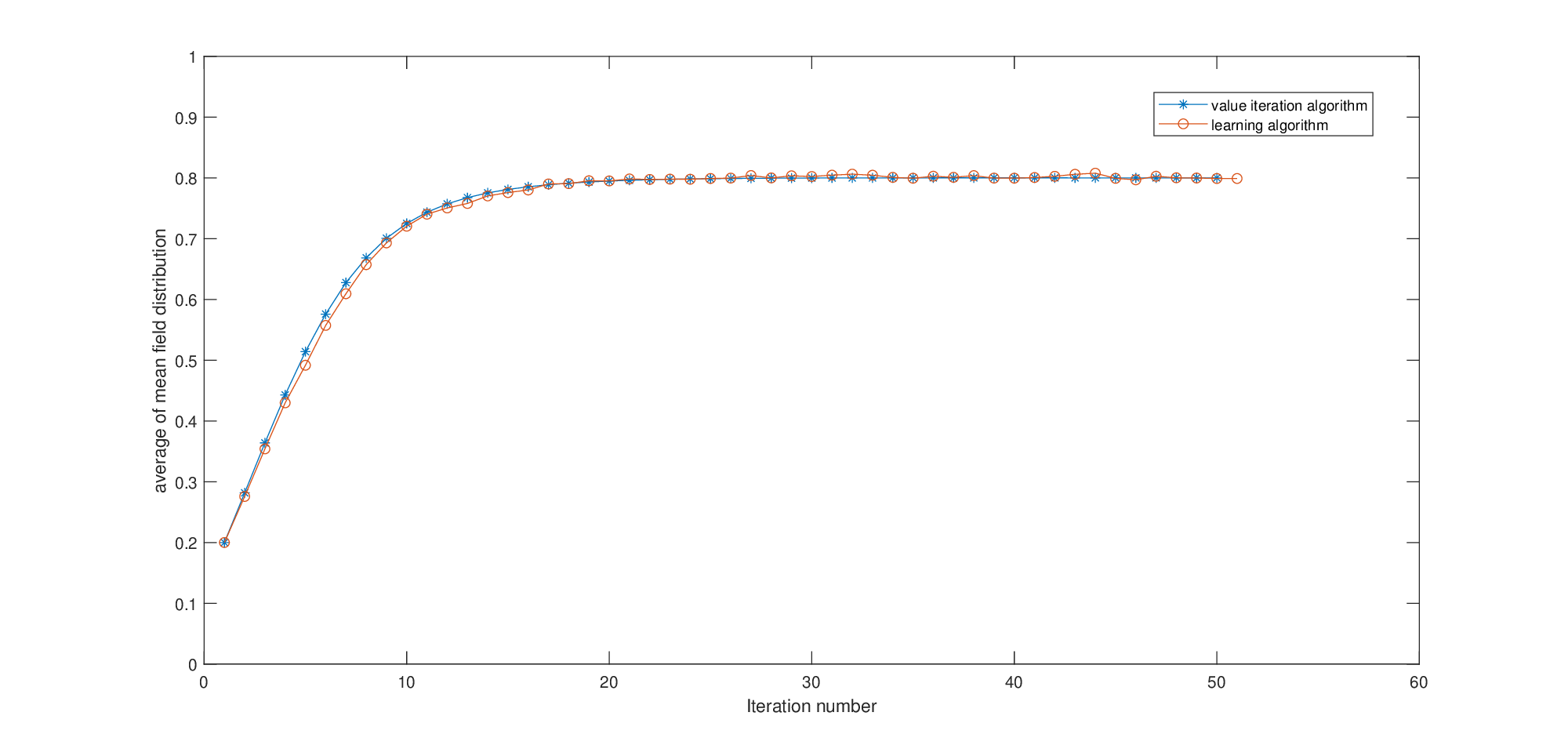}
	\caption{Comparison of state-measures: discounted-cost} \label{meanfield_comp}
\end{figure}
	
\begin{figure}[h]
\centering
	\includegraphics[scale=0.5]{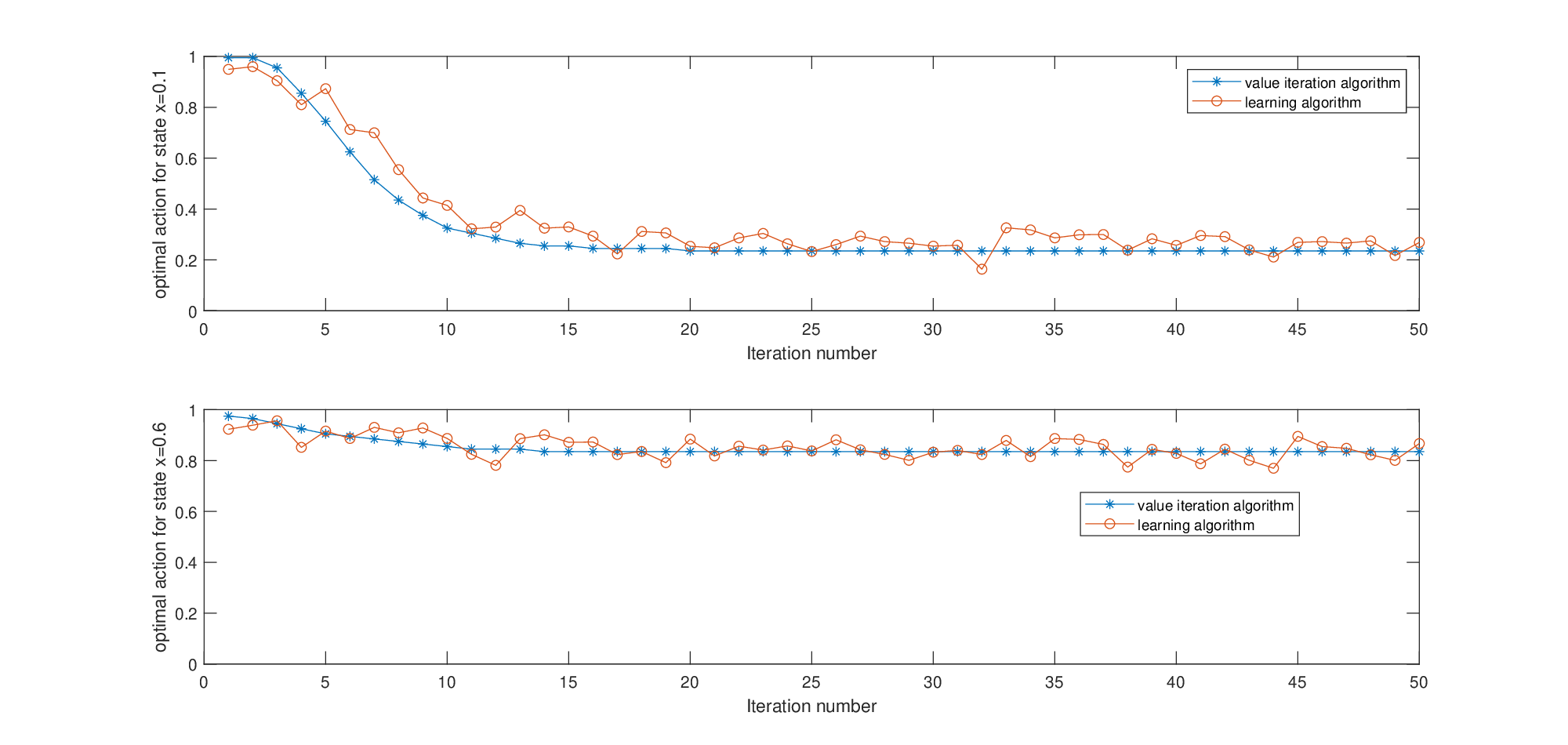}
	\caption{Comparison of policies: discounted-cost} \label{policy_comp}
\end{figure}

\begin{figure}[h]
\centering
	\includegraphics[scale=0.5]{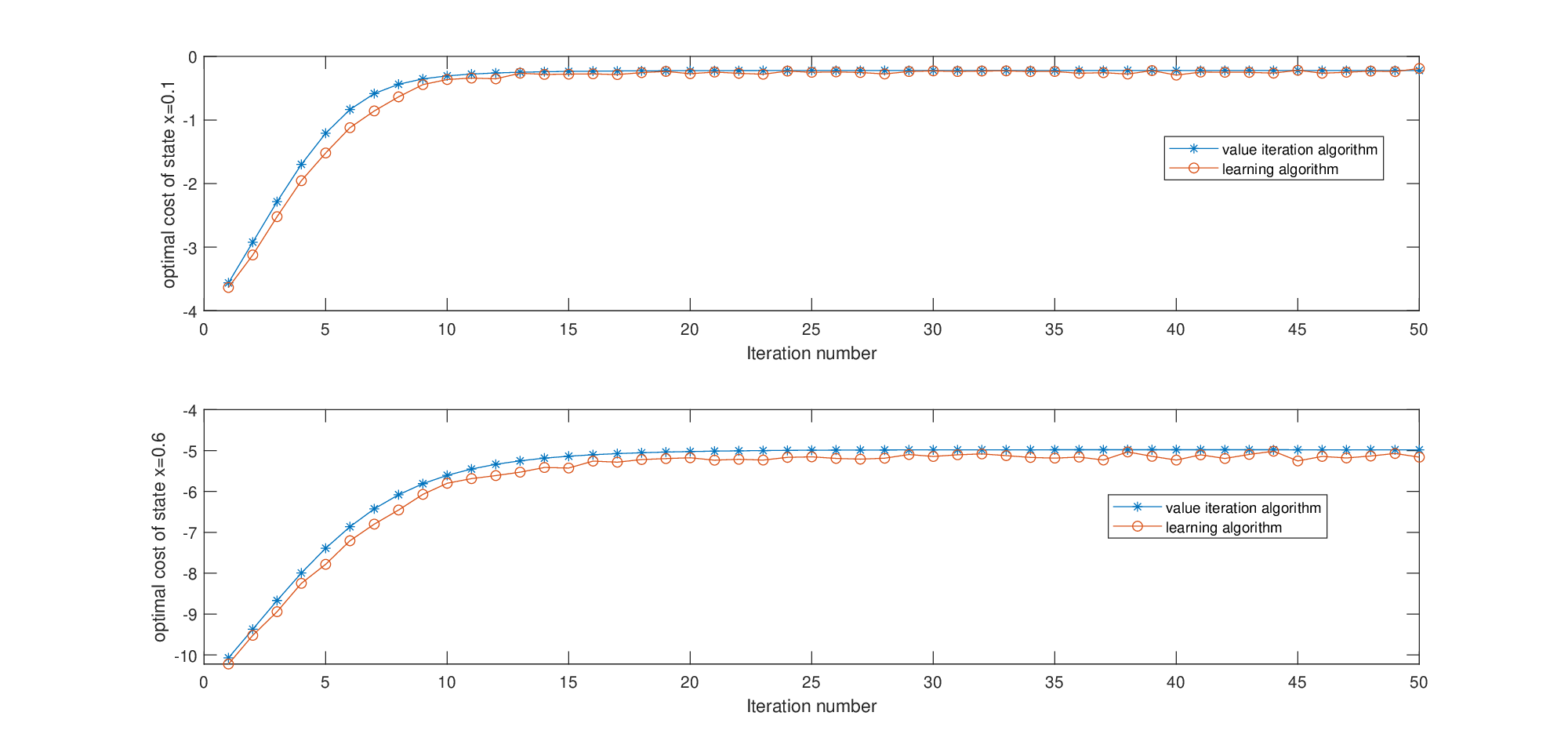}
	\caption{Comparison of costs: discounted-cost} \label{cost_comp}
\end{figure}

\subsection{Average Cost}\label{num_ex_average}
	
We consider a mean-field game with state space $\sX = \{0,1\}$ and action space $\sA = [0,1]$. The transition probability $p: \sX \times \sA \rightarrow \P(\sX)$ is independent of the mean-field term and is given by 
$$
p(\,\cdot\,|x,a) = l_0(\,\cdot\,|x,) \cdot a + l_1(\,\cdot\,|x) \cdot (1-a),
$$
where 
\begin{align}
	&l_0(1|0) = \eta, \quad l_0(1|1) = 1-\alpha, \nonumber \\
	&l_1(1|0)=\kappa, \quad l_1(1|1) = 1-\xi. \nonumber
\end{align}
The one-stage cost function $c:\sX \times \sA \times \P(\sX) \rightarrow [0,\infty)$ depends on the mean-field term and is defined to be
$$
	c(x,a,\mu) = \tau \, \langle \mu \rangle \, x+\lambda \, (1-\langle \mu \rangle) \, (1-a) + \gamma \, a^2, 
$$
where $\langle \mu \rangle$ is the mean of the distribution $\mu$ on $\sX$. It can be verified that Assumption~\ref{as1} holds in this example. We use the following values of the parameters:
\begin{align}
	\eta&=0.7, \quad \alpha = 0.1, \quad \kappa=0.1, \quad \xi=0.8 \nonumber \\
	\tau&=0.1, \quad \lambda=0.4, \quad \gamma=0.2. \nonumber 
\end{align}
With these parameters, it is also straightforward to check that  Assumption~\ref{av-as2}-(a) holds. We run the learning algorithm using the following parameters: $N=1000$, $L=50$, $M=1000$, $K=50$. Output of the learning algorithms contain the average of the state-measure (i.e., mean-field distribution) and mean-field equilibrium policies. In the fitted $Q$-iteration algorithm, we pick the function class $\F$ as two-layer neural networks with $20$ hidden units. We use the neural network fitting tool of MATLAB. In particular, we use `fittnet', `train', and `net' functions of MATLAB, where `Levenberg-Marquardt' is picked as the training algorithm and the transfer function is chosen as `hyperbolic tangent sigmoid transfer function'. The parameters of the neural network fitting tool of MATLAB are set to default values. We also run the value iteration algorithm using MFE operator $H^{\av}$ to find the correct average of state-measure and mean-field equilibrium policies. Then, we compare the learned outputs with the outputs of the value iteration algorithm. Figures~\ref{av-meanfield_comp} and \ref{av-policy_comp} show this comparison for the average-cost. It can be seen that learned outputs converge to the outputs of the value iteration algorithm.

\begin{figure}[h]
\centering
	\includegraphics[scale=0.5]{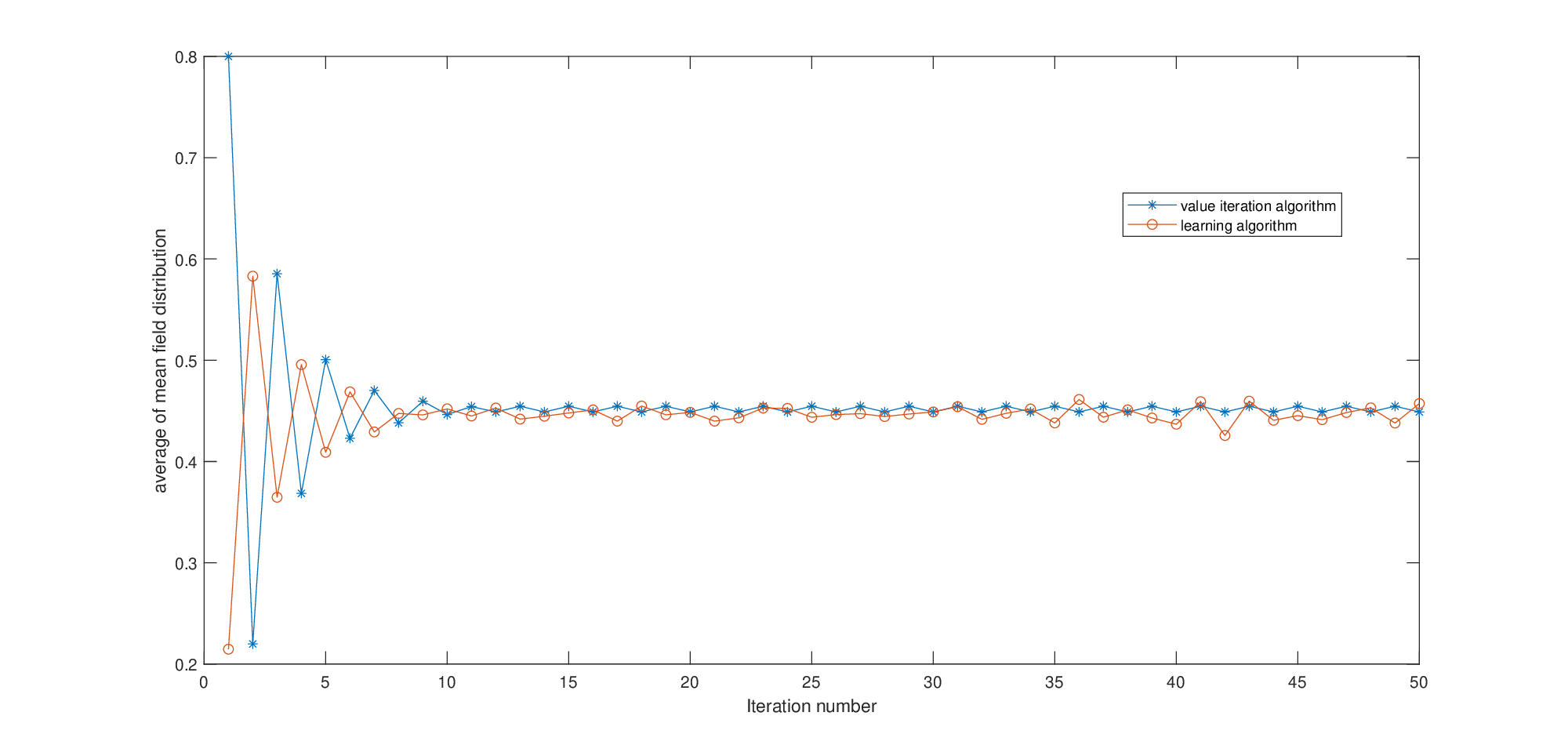}
	\caption{Comparison of state-measures: average-cost} \label{av-meanfield_comp}
\end{figure}
	
\begin{figure}[h]
\centering
	\includegraphics[scale=0.5]{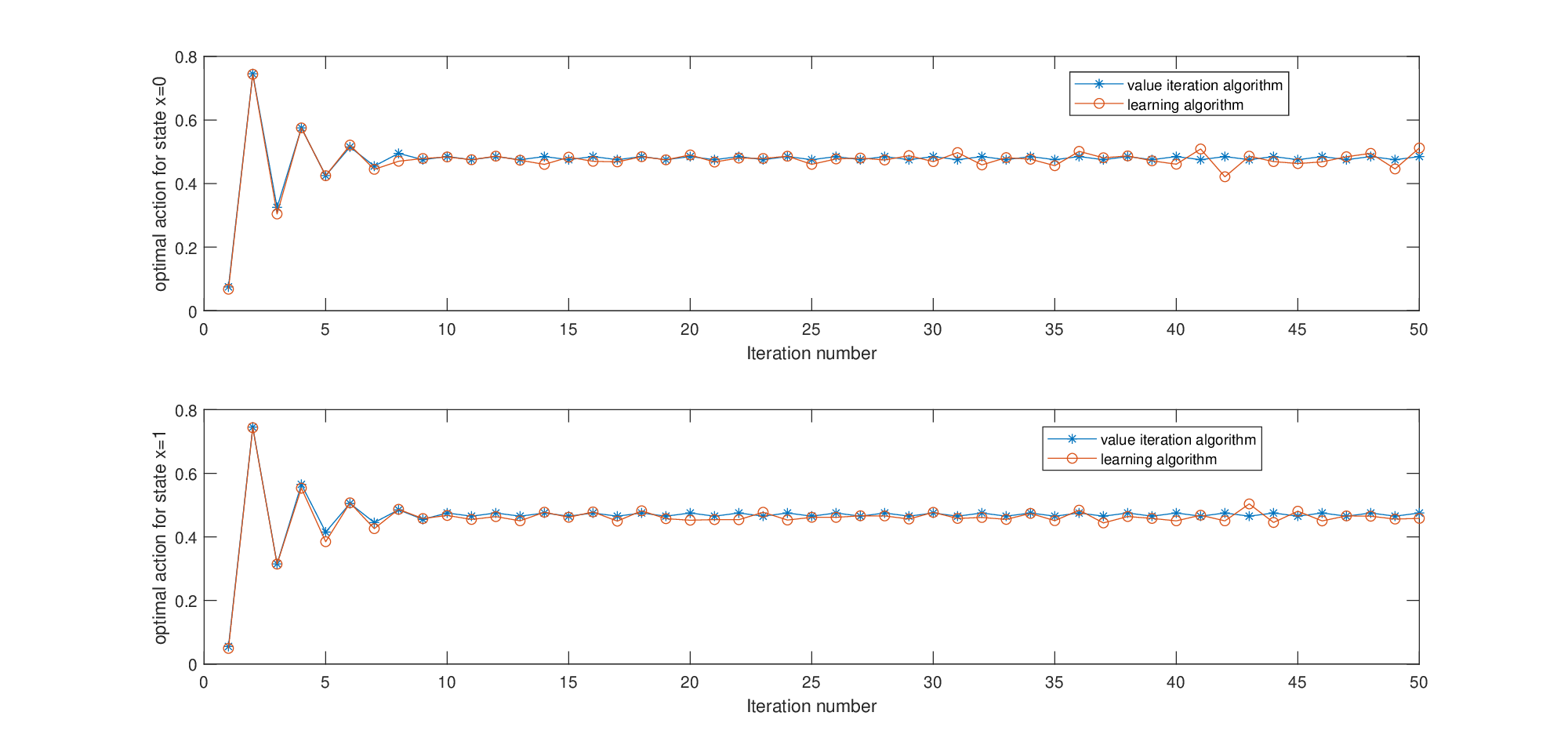}
	\caption{Comparison of policies: average-cost} \label{av-policy_comp}
\end{figure}
	

\section{Conclusion}\label{conc}

This paper has established a learning algorithm for discrete-time mean-field games subject to discounted-cost and average-cost criteria. Under certain regularity conditions on system components, we have proved that the policy obtained from the learning algorithm converges to the policy in mean-field equilibrium with some probabilistic convergence rate. We have then used the learned policy to construct an approximate Nash equilibrium for the finite-agent game problem. 

\section*{Appendix}

In this appendix, we state some auxiliary results that will be frequently used in the paper. The first result gives a bound on $l_{\infty}$-norm  of uniformly Lipschitz continuous function $g(x,a)$ with respect to the action $a$ in terms of its $l_2$-norm. 

\begin{lemma}\label{li-l2}
Let $g: \sX \times \sA \rightarrow \R$ be a uniformly Lipschitz continuous function of the action $a$ with Lipschitz constant $L$. Then, under Assumption~\ref{as1}-(c), we have 
$$
\|g\|_{\infty}  \leq \max\left( \left[\frac{m(\sA) \, (\dim_{\sA}+1)! \, \zeta_0}{\alpha \, (2/L)^{\dim_{\sA}}} \, \|g\|_{\nu} \right]^{1/(\dim_{\sA}+1)}, (\dim_{\sA}+1) \, \zeta_0 \, \|g\|_{\nu} \right),
$$
where $\alpha  > 0$ is the constant in Assumption~\ref{as1}-(c) and $\displaystyle \zeta_0 \coloneqq \frac{1}{\sqrt{\min_x \nu(x)}}$.
\end{lemma}

\proof
Under Assumption~\ref{as1}-(c), by following the same steps as in the proof of {\cite[Lemma D.2]{AnMuSz07-t}}, for all $(x,a) \in \sX \times \sA$, we obtain the following 
$$
\left( \int_{\sA} |g(x,\hat{a})|^2 \, m_{\sA}(d\hat{a}) \right)^{1/2} \geq \min\left( \left[\frac{\alpha \, (2/L)^{\dim_{\sA}}}{m(\sA) \, (\dim_{\sA}+1)!} \, |g(x,a)| \right]^{(\dim_{\sA}+1)}, \frac{ |g(x,a)|}{(\dim_{\sA}+1)} \right)
$$
Note that we also have 
\begin{align}
&\left( \sum_x \int_{\sA} |g(x,a)|^2 \, m_{\sA}(da) \, \nu(x) \right)^{1/2} \geq \sqrt{\min_x \nu(x)} \, \sup_x \, \left( \int_{\sA} |g(x,a)|^2 \, m_{\sA}(da) \right)^{1/2}.  \nonumber
\end{align}
Therefore, above inequalities lead to 
$$
\|g\|_{\infty}  \leq \max\left( \left[\frac{m(\sA) \, (\dim_{\sA}+1)! \, \zeta_0}{\alpha \, (2/L)^{\dim_{\sA}}} \, \|g\|_{\nu} \right]^{1/(\dim_{\sA}+1)}, (\dim_{\sA}+1) \, \zeta_0 \, \|g\|_{\nu} \right).
$$
\endproof

\begin{remark}\label{simplify}
In the paper, to simplify the notation, we will always assume that 
$$
\left[\frac{m(\sA) \ (\dim_{\sA}+1)! \, \zeta_0}{\alpha \, (2/L)^{\dim_{\sA}}} \, \|g\|_{\nu} \right]^{1/(\dim_{\sA}+1)} \geq (\dim_{\sA}+1) \, \zeta_0 \, \|g\|_{\nu}.
$$
Therefore, the bound in Lemma~\ref{li-l2} will always be in the following form:
$$
\|g\|_{\infty}  \leq \left[\frac{m(\sA) \ (\dim_{\sA}+1)! \, \zeta_0}{\alpha \, (2/L)^{\dim_{\sA}}} \, \|g\|_{\nu} \right]^{1/(\dim_{\sA}+1)}.
$$
\end{remark}

Before we state the next result, we need to give some definitions. Let $\sE$ be some set. Let ${\cal G}$ be a set of real-valued functions on $\sE$ taking values in $[0,K]$. For any $e^{1:N} \coloneqq \{e_i\}_{i=1}^N \in \sE^N$, define the following semi-metric on ${\cal G}$:
$$
d_{e^{1:N}}(g,h) \coloneqq \frac{1}{N} \, \sum_{i=1}^N |g(e_i) - h(e_i)|. 
$$
Then, for any $\varepsilon > 0$, let $N_1\left(\varepsilon,\{e_i\}_{i=1}^N,{\cal G}\right)$ denote the $\varepsilon$-covering number of ${\cal G}$ in terms of semi-metric $d_{e^{1:N}}$ \cite[pp. 14]{Vid10}. Moreover, let $V_{\cal G}$ denote the pseudo-dimension of the function class ${\cal G}$ \cite[Definition 4.2, pp. 120]{Vid10}. 

\begin{lemma}{\cite[Proposition E.3]{AnMuSz07-t}}\label{cov-num}
For any $e^{1:N}$, we have 
$$
N_1\left(\varepsilon,\{e_i\}_{i=1}^N,{\cal G}\right) \leq e \, (V_{\cal G}+1) \, \left(\frac{2 e K}{\varepsilon}\right)^{V_{\cal G}}.
$$
\end{lemma}

Let $P(\,\cdot\,|x)$ be a transition probability on $\sX$ with the following contraction coefficient
$$
\theta_P \coloneqq \frac{1}{2} \sup_{x,z} \|P(\,\cdot\,|x)-P(\,\cdot\,|z)\|_1. 
$$
Then, the following result holds.

\begin{lemma}{\cite[Lemma A.2]{KoRa08}}\label{KoRa08}
Let $\mu,\nu \in \P(\sX)$. Then, 
$$
\sum_y \left| \sum_x P(y|x) \, \mu(x) - \sum_x P(y|x) \, \nu(x) \right| \leq \theta_P \, \|\mu-\nu\|_1. 
$$
\end{lemma}
In other words, if we define $\mu P(\,\cdot\,) \coloneqq \sum_x P(\,\cdot\,|x) \, \mu(x) \in \P(\sX)$ and $\nu P(\,\cdot\,) \coloneqq \sum_x P(\,\cdot\,|x) \, \nu(x) \in \P(\sX)$, then we have $\|\mu P-\nu P\|_1 \leq \theta_P \, \|\mu-\nu\|_1$. Indeed, the last inequality explains why $\theta_P$ is called contraction coefficient.

\acks{
This research was supported by The Scientific and Technological Research Council of Turkey (T\"{U}B\.{I}TAK) B\.{I}DEB 2232 Research Grant.}


\end{document}